\documentclass[12pt]{article}

\usepackage{microtype}
\usepackage{fullpage}
\usepackage{xspace,xcolor}
\usepackage{enumitem}
\usepackage{mathtools}

\usepackage{amsmath,amsthm,amssymb}

\newcommand{\maxCol}{\mathrm{51}}
\newcommand{\lmin}{\mathit{Lmin}}
\newcommand{\lmax}{\mathit{Lmax}}
\newcommand{\EarlyStopCV}{\textsc{EarlyStopCV}\xspace}
\newcommand{\ChooseNewColour}{\textsc{ChooseNewPhase1Colour}\xspace}
\newcommand{\CVChoice}{\textsc{CVChoice}\xspace}

\newcommand{\delay}{\ensuremath{\mathit{delay}(\alpha,\beta)}}
\newcommand{\colv}[2]{\ensuremath{\mathit{col}(v_{#1},#2)}}
\newcommand{\sweeprad}[1]{\ensuremath{\mathit{SR}(#1)}}
\newcommand{\icrit}{\ensuremath{i_{\mathrm{crit}}}}
\newcommand{\sleft}{\ensuremath{s_{\mathrm{left}}}}
\newcommand{\sright}{\ensuremath{s_{\mathrm{right}}}}
\newcommand{\rvalg}{\ensuremath{\mathcal{A}_{\textrm{rv}}}}
\newcommand{\colalg}{\ensuremath{\mathcal{A}_{\textrm{col}}}}
\newcommand{\threecolalg}{\ensuremath{\mathcal{A}_{\textrm{3col}}}}
\newcommand{\rvadjalg}{\ensuremath{\mathcal{A}_{\textrm{rv-adj}}}}
\newcommand{\rvcanonalg}{\ensuremath{\mathcal{A}_{\textrm{rv-canon}}}}
\newcommand{\rvknownDalg}{\ensuremath{\mathcal{A}_{\textrm{rv-D}}}}
\newcommand{\rvunknownDalg}{\ensuremath{\mathcal{A}_{\textrm{rv-noD}}}}

\newcommand{\ep}[2]{\ensuremath{{E}^{#2}_{#1}}}
\newcommand{\Pg}[2]{\ensuremath{\mathcal{P}^{#2}_{#1}}}
\newcommand{\klog}[1]{\ensuremath{\kappa\log^*(#1)}}

\DeclarePairedDelimiter{\floor}{\lfloor}{\rfloor}
\newcommand{\dcrit}{\ensuremath{d_{\textrm{crit}}}}
\newcommand{\first}[2]{\ensuremath{F(#1,#2)}}

\usepackage{hyperref}

\newtheorem{theorem}{Theorem}[section]
\newtheorem{lemma}[theorem]{Lemma}
\newtheorem{corollary}[theorem]{Corollary}
\newtheorem{claim}{Claim}
\newtheorem{proposition}[theorem]{Proposition}
\newtheorem{definition}[theorem]{Definition}
\newtheorem{fact}[theorem]{Fact}

\usepackage{cleveref}
\crefname{claim}{Claim}{Claims}
\crefname{fact}{Fact}{Facts}

\usepackage{algorithm}
\usepackage{algpseudocode}

\begin{document}
	\def\thefootnote{\fnsymbol{footnote}}
	
	\title{Fast Deterministic Rendezvous in Labeled Lines} %TODO Please add
	
	\author{
		Avery Miller\footnotemark[1]
		\and Andrzej Pelc\footnotemark[2]
	}
	
	\footnotetext[1]{Department of Computer Science, University of Manitoba, Winnipeg, Manitoba, R3T 2N2, Canada. {\tt avery.miller@umanitoba.ca}. Supported by NSERC Discovery Grant RGPIN--2017--05936.}
	\footnotetext[2]{D\'epartement d'informatique, Universit\'e du Qu\'ebec en Outaouais, Gatineau,
		Qu\'ebec J8X 3X7, Canada. {\tt pelc@uqo.ca}. Partially supported by NSERC Discovery Grant RGPIN--2018--03899
		and by the Research Chair in Distributed Computing at the
		Universit\'e du Qu\'ebec en Outaouais.}

	\maketitle

\begin{abstract}
Two mobile agents, starting from different nodes of a network modeled as a graph, and woken up at possibly different times, have to meet at the same node. This problem is known as {\em rendezvous}. Agents move in synchronous rounds. In each round, an agent can either stay idle or move to an adjacent node. We consider deterministic rendezvous in the infinite line, i.e., the infinite graph with all nodes of degree 2. Each node has a distinct label which is a positive integer. An agent currently located at a node can see its label and both ports 0 and 1 at the node. The time of rendezvous is the number of rounds until meeting, counted from the starting round of the earlier agent. We consider three scenarios. In the first scenario, each agent knows its position in the line, i.e., each of them knows its initial distance from the smallest-labeled node, on which side of this node it is located, and the direction towards it. For this scenario, we design a rendezvous algorithm working in time $O(D)$, where $D$ is the initial distance between the agents. This complexity is clearly optimal. In the second scenario, each agent initially knows only the label of its starting node and the initial distance $D$ between the agents. In this scenario, we design a rendezvous algorithm working in time $O(D\log^*\ell)$, where $\ell$ is the larger label of the starting nodes. We also prove a matching lower bound $\Omega(D\log^*\ell)$. Finally, in the most general scenario, where each agent initially knows only the label of its starting node, we design a rendezvous algorithm working in time $O(D^2(\log^*\ell)^3)$, which is thus at most cubic in the lower bound. All our results remain valid (with small changes) for arbitrary finite lines and for cycles. Our algorithms are drastically better than approaches that use graph exploration, which have running times that depend on the size or diameter of the graph. Our main methodological tool, and the main novelty of the paper, is a two way reduction: from fast colouring of the infinite labeled line using a constant number of colours in the $\mathcal{LOCAL}$ model to fast rendezvous in this line, and vice-versa. In one direction we use fast node colouring to quickly break symmetry between the identical agents. In the other direction, a lower bound on colouring time implies a lower bound on the time of breaking symmetry between the agents, and hence a lower bound on their meeting time.
\end{abstract}

\newpage

%--------------------------------------------------
%--------------------------------------------------
\section{Introduction}
%--------------------------------------------------
%--------------------------------------------------

%--------------------------------------------------
%--------------------------------------------------
\subsection{Background}
%--------------------------------------------------
%--------------------------------------------------
Two mobile agents, starting from different nodes of a network modeled as a graph, and woken up at possibly different times, have to meet at the same node.
This problem is known as {\em rendezvous}.
The autonomous mobile entities (agents) may be natural, such as people who want to meet in a city whose streets form a network. They may also represent 
human-made objects, such as software agents in computer networks or mobile robots navigating in a network of corridors in a mine or a building. 
Agents might want to meet to share previously collected data 
or to coordinate future network maintenance tasks. 

%--------------------------------------------------
%--------------------------------------------------
\subsection{Model and problem description}
%--------------------------------------------------
%--------------------------------------------------
We consider deterministic rendezvous in the infinite line, i.e., the infinite connected graph with all nodes of degree 2. Each node has a distinct label which is a positive integer. An agent currently located at a node can see its label and both ports 0 and 1 at the node. Technically, agents are anonymous, but each agent could adopt the label of its starting node as its label. Agents have synchronized clocks ticking at the same rate, measuring {\em rounds}.  Agents are woken up by an adversary in possibly different rounds. The clock of each agent starts in its wakeup round. In each round, an agent can either stay idle or move to an adjacent node. After moving, an agent knows on which port number it arrived at its current node. When agents cross, i.e., traverse an edge simultaneously in different directions, they do not notice this fact. No limitation is imposed on the memory of the agents. Computationally, they are modeled as Turing Machines. The time of rendezvous is the number of rounds until meeting, counted from the starting round of the earlier agent.

In most of the literature concerning rendezvous, nodes are assumed to be anonymous. The usual justification of this weak assumption is that in labeled graphs (i.e., graphs whose nodes are assigned distinct positive integers) each agent can explore the entire graph and then meet the other one at the node with the smallest label. Thus, the rendezvous problem in labeled graphs can be reduced to exploration. While this simple strategy is correct for finite graphs, it requires time at least equal to the size of the graph, and hence very inefficient for large finite graphs, even if the agents start at a very small initial distance. In infinite graphs, this strategy is incorrect: indeed, if label 1 is not used in some labeling, an agent may never learn what is the smallest used label. For large environments that would be impractical to exhaustively search, it is desirable and natural to relate rendezvous time to the initial distance between the agents, rather than to the size of the graph or the value of the largest node label. This motivates our choice to study the infinite labeled line: we must avoid algorithms that depend on exhaustive search, and there cannot be an efficient algorithm that is dependent on a global network property such as size, largest node label, etc. In particular, the initial distance between the agents and the labels of encountered nodes should be the only parameters available to measure the efficiency of algorithms. Results for the infinite line can then be applied to understand what's possible in large path-like environments.

We consider three scenarios, depending on the amount of knowledge available {\em a priori} to the agents. In the first scenario, each agent knows its position in the line, i.e., each of them knows its distance from the smallest-labeled node, on which side of this node it is located, and the direction towards it. This scenario is equivalent to assuming that the labeling is very particular, such that both above pieces of information can be deduced by the agent by inspecting the label of its starting node and port numbers at it. One example of such a labeling is $...8,6,4,2,1,3,5,7,...$ (recall that labels must be positive integers), with the following port numbering. For nodes with odd labels, port 1 always points to the neighbour with larger label and port 0 points to the neighbour with smaller label, while for nodes with even label, port 1 always points to the neighbour with smaller label and port 0 points to the neighbour with larger label. We will call the line with the above node labeling and port numbering the {\em canonical line}, and use, for any node, the term ``right of the node'' for the direction corresponding to port 1 and the term ``left of the node'' for the direction corresponding to port 0. In the second scenario that we consider, the labeling and port numbering are arbitrary and each agent knows {\em a priori} only the label of its starting node and the initial distance $D$ between the agents. Finally, in the third scenario, the labeling and port numbering are arbitrary and each agent knows {\em a priori} only the label of its starting node.

%--------------------------------------------------
%--------------------------------------------------
\subsection{Our results}
%--------------------------------------------------
%--------------------------------------------------
We start with the scenario of the canonical line. This scenario was previously considered in \cite{CCGKM}, where the authors give a rendezvous algorithm with optimal complexity $O(D)$, where $D$ is the (arbitrary) initial distance between the agents. However, they use the strong assumption that both agents are woken up in the same round. In our first result we get rid of this assumption: we design
a rendezvous algorithm with complexity $O(D)$, regardless of the delay between the wakeup of the agents. This complexity is clearly optimal. 

The remaining results are expressed in terms of the \emph{base-2 iterated logarithm} function, denoted by $\log^*$. The value of $\log^*(n)$ is 0 if $n \leq 1$, and otherwise its value is $1+\log^*(\log_2 n)$. This function grows extremely slowly with respect to $n$, e.g., $\log^*(2^{65536})=5$.

For the second scenario (arbitrary labeling with known initial distance $D$) we design a rendezvous algorithm working in time $O(D\log^*\ell)$, where $\ell$ is the larger label of the two starting nodes. We also prove a matching lower bound $\Omega(D\log^*\ell)$, by constructing an infinite labeled line in which this time is needed, even if $D$ is known and if agents start simultaneously. As a corollary, we get the following impossibility result: for any deterministic algorithm and for any finite time $T$, there exists an infinite labeled line such that two agents starting from some pair of adjacent nodes cannot meet in time $T$.

Finally, for the most general scenario, where each agent knows {\em a priori} only the label of its starting node, we design a rendezvous algorithm working in time
$O(D^2(\log^*\ell)^3)$, for arbitrary unknown $D$. This complexity is thus at most cubic in the above lower bound that holds even if $D$ is known.

It should be stressed that the complexities of our algorithms in the second and third scenarios depend on $D$ and on the larger label of the two starting nodes. No algorithm whose complexity depends on $D$ and on the maximum label in even a small vicinity of the starting nodes would be satisfactory, since neighbouring nodes can have labels differing in a dramatic way,
e.g., consider two agents that start at adjacent nodes, and these nodes are labeled with small integers like 2 and 3, but other node labels in their neighbourhoods are extremely large. Our approach demonstrates that this is an instance that can be solved very quickly, but an approach whose running time depends on the labels of other nodes in the neighbourhood can result in arbitrarily large running times (which are very far away from the lower bound).

An alternative way of looking at the above three scenarios is the following. In the first scenario, the agent is given {\em a priori} the entire (ordered) labeling of the line. Of course, since this is an infinite object, the labeling cannot be given to the agent as a whole, but the agent can {\em a priori} get the answer to any question about it, in our case answers to two simple questions: how far is the smallest-labeled node and in which direction. In the second and third scenarios, the agent gets information about the label of a given node only when visiting it. It is instructive to consider another, even weaker, scenario: the agent learns the label of its initial node but the other labels are never revealed to it. This weak scenario is equivalent to the scenario of labeled agents operating in an anonymous line: the labels of the agents are guaranteed to be different integers (they are the labels of the starting nodes) but all other nodes look the same. It follows from \cite{DFKP} that in this scenario the optimal time of rendezvous is $\Theta(D\log \ell)$, where $\ell$ is, as usual, the larger label of the starting nodes.\footnote{In \cite {DFKP} cycles are considered instead of the infinite line, and the model and result are slightly different but obtaining, as a corollary, this complexity in our model is straightforward.} 
Hence, at least for known $D$, we get strict separations between optimal rendezvous complexities, according to three ways of getting knowledge about the labeling of the line: $\Theta(D)$, if all knowledge is given at once, 
$\Theta(D\log^* \ell)$ if knowledge about the labeling is given as the agents visit new nodes, and $\Theta(D\log \ell)$, if no knowledge about the labeling is ever given to the agents, apart from the label of the starting node. Of course, in the ultimately weak scenario where even the label of the starting node is hidden from the agents, deterministic rendezvous would become impossible: there is no deterministic algorithm guaranteeing rendezvous of anonymous agents in an anonymous infinite line, even if they start simultaneously from adjacent nodes.

It should be mentioned that the $O(D\log \ell)$-time rendezvous algorithm from  \cite{DFKP} is also valid in our scenario of arbitrary labeled lines, with labels of nodes visible to the agents at their visit and with unknown $D$: the labels of the visited nodes can be simply ignored by the algorithm from  \cite{DFKP}. This complexity is incomparable with our complexity $O(D^2(\log^* \ell)^3)$. For bounded $D$ our algorithm is much faster but for large $D$ compared to $\ell$ it is less efficient.

As usual in models where agents cannot detect crossing on the same edge when moving in opposite directions, we guarantee rendezvous by creating time intervals in which one agent is idle at its starting node and the other searches a sufficiently large neighbourhood to include this node. Since the adversary can choose the starting rounds of the agents, it is difficult to organize these time intervals to satisfy the above requirement.
Our main methodological tool, and the main novelty of the paper is a two way reduction: from fast colouring of the infinite labeled line using a constant number of colours in the $\mathcal{LOCAL}$ model to fast rendezvous in this line, and vice-versa. In one direction we use fast node colouring to quickly break symmetry between the identical agents. In the other direction, a lower bound on colouring time implies a lower bound on the time of breaking symmetry between the agents, and hence a lower bound on their meeting time.  

As part of our approach to solve rendezvous using node colouring, we provide a result (based on the idea of Cole and Vishkin \cite{CV}) that might be of independent interest: for the $\mathcal{LOCAL}$ model, we give a deterministic distributed algorithm $\EarlyStopCV$ that properly 3-colours any infinite line whose nodes are initially labeled with integers greater than 1 that form a proper colouring\footnote{Note that any labeling of nodes by arbitrary distinct integers greater than 1 is an example of such a proper colouring.}, such that the execution of $\EarlyStopCV$ at any node $x$ with initial label $\mathrm{ID}_x$ terminates in time $O(\log^*(\mathrm{ID}_x))$, and the algorithm does not require that the nodes have any initial knowledge about the network other than their own label. %This result also applies to other networks such as cycles and finite paths.

All our results remain valid for finite lines and cycles: in the first scenario without any change, and in the two other scenarios with small changes. As previously mentioned, it is always possible to meet in a labeled line or cycle of size $n$ in time $O(n)$ by exploring the entire graph and going to the smallest-labeled node. Thus, in the  second scenario the upper and lower bounds on complexity change to $O(\min(n,D\log^*\ell))$
and $\Omega(\min(n,D\log^*\ell))$, respectively, and in the third, most general scenario, the complexity of (a slightly modified version of) our algorithm changes to $O(\min(n,D^2(\log^*\ell)^3))$.

%--------------------------------------------------
%--------------------------------------------------
\subsection{Related work}
%--------------------------------------------------
%--------------------------------------------------
Rendezvous has been studied both under randomized and deterministic scenarios.
A  survey of  randomized rendezvous under various models  can be found in
\cite{alpern02b}, cf. also  \cite{alpern95a,alpern02a,anderson90,baston98}. 
Deterministic rendezvous in networks has been surveyed in \cite{Pe,Pe2}.
Many authors
considered geometric scenarios (rendezvous in the real line, e.g.,  \cite{baston98,baston01},
or in the plane, e.g., \cite{anderson98a,anderson98b,BDPP,CGKK}).
Gathering more than two agents was studied, e.g., in \cite{fpsw,lim96,thomas92}.

In the deterministic setting, many authors studied the feasibility and time complexity of synchronous rendezvous in networks. For example, deterministic rendezvous of agents equipped with tokens used to mark nodes was considered, e.g., in~\cite{KKSS}. In most of the papers concerning rendezvous in networks, nodes of the network are assumed to be unlabeled and agents are not allowed to mark the nodes in any way.
In this case, the symmetry is usually broken by assigning the agents distinct labels and assuming that each agent knows its own label but not the label of the other agent.
Deterministic rendezvous of labeled agents in rings was investigated, e.g., in \cite{DFKP,KM} and in arbitrary graphs in  \cite{DFKP,KM,TSZ07}.
%In \cite{DFKP}, the authors gave tight upper and lower bounds
%of $\Theta (D\log \ell)$ on the 
%time of rendezvous when agents start simultaneously, where $D$ is the initial distance between agents and $\ell$ is the smaller label. 
%They also gave a lower bound of $\Omega(n+D\log \ell)$ on the time of rendezvous with arbitrary delay in $n$-node rings.  In \cite{KM} an upper bound $O(n\log \ell)$
%on the time of rendezvous was given, even without knowledge of $n$.
In \cite{DFKP}, the authors present a rendezvous algorithm whose running time is polynomial in the size of the graph, in the length of the shorter
label and in the delay between the starting times of the agents. In \cite{KM,TSZ07}, rendezvous time is polynomial in the first two of these parameters and independent of the delay.
In \cite{BP1,BP2} rendezvous was considered, respectively, in unlabeled infinite trees, and in arbitrary connected unlabeled infinite graphs, but the complexities depended, among others, on the logarithm of the size of the label space.
Gathering many anonymous agents in unlabeled networks was investigated in \cite{DP}. In this weak scenario, not all initial configurations of agents are possible to gather, and the authors of \cite{DP} characterized all such configurations and provided universal gathering algorithms for them. On the other hand, time of rendezvous in labeled networks was studied, e.g., in \cite{MP}, in the context of algorithms with advice. In \cite{CCGKM}, the authors studied rendezvous under a very strong assumption that each agent has a map of the network and knows its position in it.

Memory required by the agents to achieve deterministic rendezvous was studied in \cite{FP2} for trees and in  \cite{CKP} for general graphs.
Memory needed for randomized rendezvous in the ring was discussed, e.g., in~\cite{KKPM08}. 

Apart from the synchronous model used in this paper, several authors investigated asynchronous rendezvous and approach in the plane \cite{BBDDP,CFPS,fpsw} and in networks modeled as graphs
\cite{BCGIL,DPV,DGKKP}.
In the latter scenario, the agent chooses the edge to traverse, but the adversary controls the speed of the agent. Under this assumption, rendezvous
at a node cannot be guaranteed even in the two-node graph. Hence the rendezvous requirement is relaxed to permit the agents to meet inside an edge.

%--------------------------------------------------
%--------------------------------------------------
\subsection{Roadmap} 
%--------------------------------------------------
%--------------------------------------------------
In \Cref{sec:tools}, we state two results concerning fast colouring of labeled lines. In \Cref{sec:canon}, we present and analyze our optimal rendezvous algorithm for the canonical line. 
In \Cref{sec:knowndist}, we present our optimal rendezvous algorithm for arbitrary labeled lines with known initial distance between the agents, we analyze the algorithm, and we
give the proof of the matching lower bound. In \Cref{sec:unknowndist}, we give our rendezvous algorithm for arbitrary labeled lines with unknown initial distance between the agents. In \Cref{sec:linecolouring}, we give a fast colouring algorithm for the infinite line, as announced in \Cref{sec:tools}. In \Cref{sec:conclusion}, we conclude the paper and present some open problems.

%--------------------------------------------------
%--------------------------------------------------
\section{Tools}\label{sec:tools}
%--------------------------------------------------
%--------------------------------------------------

We will use two results concerning distributed node colouring of lines and cycles in the $\mathcal{LOCAL}$ communication model \cite{Pel}. In this model, communication proceeds in synchronous
rounds and all nodes start simultaneously. In each round, each node
can exchange arbitrary messages with all of its neighbours and perform arbitrary local computations. A \emph{proper colouring} of a graph is an assignment of one integer to each node of the graph such that adjacent nodes have different labels. Clearly, if a graph's node labels are distinct integers, then they form a proper colouring. In the graph $k$-colouring problem, the goal is for each node in the graph to adopt one of $k$ colours in such a way that the adopted colours form a proper colouring.

The first result is based on the 3-colouring algorithm of Cole and Vishkin \cite{CV}, but improves on their result in two ways. First, our algorithm does not require that the nodes know the size of the network or the largest label in the network. Second, the running time of our algorithm at any node $x$ with initial label $\mathrm{ID}_x$ is $O(\log^*(\mathrm{ID}_x))$. The running-time guarantee is vital for later results in this paper, and it is not provided by the original $O(\log^*(n))$ algorithm of Cole and Vishkin since the correctness of their algorithm depends on the fact that all nodes execute the algorithm for the same number of rounds. Our algorithm from \Cref{CValg}, called $\EarlyStopCV$, is described and analyzed in \Cref{sec:linecolouring}.

\begin{proposition}\label{CValg}
	There exists an integer constant $\kappa > 1$ and a deterministic distributed algorithm $\EarlyStopCV$ such that, for any infinite line $G$ with nodes labeled with integers greater than 1 that form a proper colouring, $\EarlyStopCV$ 3-colours $G$ and the execution at any node $x$ with label $\mathrm{ID}_x$ terminates in time at most $\kappa\log^*(\mathrm{ID}_x)$.
\end{proposition}

The second result is a lower bound due to Linial \cite{LS, Li}. Note that Linial's original result was formulated for cycles labeled with integers in the range $1,\dots,n$, but the simplified proof in \cite{LS} can be adapted to hold in our formulation below.

\begin{proposition}\label{linial}
	Fix any positive integer $n$ and any set $\mathcal{I}$ of $n$ integers. For any deterministic algorithm $\mathcal{A}$ that 3-colours any path on $n$ nodes with distinct labels from $\mathcal{I}$, there is such a path and at least one node for which algorithm $\mathcal{A}$'s execution takes time at least $\frac{1}{2}\log^*(n)-1$.
\end{proposition}

Finally, we highlight an important connection between the agent-based computational model considered in this paper (where algorithm executions may start in different rounds) and the node-based $\mathcal{LOCAL}$ model (where all algorithm executions start in the same round). Consider a fixed network $G$ in which the nodes are labeled with fixed distinct integers. From the communication constraints imposed by the $\mathcal{LOCAL}$ model, for any deterministic algorithm $\mathcal{A}$, each node $v$'s behaviour in the first $i$ rounds is completely determined by the labeled $(i-1)$-neighbourhood of $v$ in $G$ (i.e., the induced labeled subgraph of $G$ consisting of all nodes within distance $i-1$ from $v$). By \Cref{CValg}, we know that each node $x$ executing $\EarlyStopCV$ in $G$ determines its final colour within $\kappa\log^*(\mathrm{ID}_x)$ rounds, so its colour is completely determined by its labeled $(\kappa\log^*(\mathrm{ID}_x)-1)$-neighbourhood in $G$. So, in the agent-based computational model, an agent operating in the same labeled network $G$ with knowledge of $\EarlyStopCV$ and the value of $\kappa$ from \Cref{CValg} could, starting in any round: visit all nodes within a distance $\kappa\log^*(\mathrm{ID}_x)-1$ from $x$, then simulate in its local memory (in one round) the behaviour of $x$ in the first $\kappa\log^*(\mathrm{ID}_x)$ rounds of $\EarlyStopCV$ in the $\mathcal{LOCAL}$ model, and thus determine the colour that would be chosen by node $x$ as if all nodes in $G$ had executed $\EarlyStopCV$ in the $\mathcal{LOCAL}$ model starting in round 1.

\begin{proposition}\label{simulate}
	Consider a fixed infinite line $G$ such that the nodes are labeled with distinct integers greater than 1. Consider any two nodes labeled $v_\alpha$ and $v_\beta$ in $G$. Suppose that all nodes in $G$ execute algorithm $\EarlyStopCV$ in the $\mathcal{LOCAL}$ model starting in round 1, and let $c_\alpha$ and $c_\beta$ be the colours that nodes $v_\alpha$ and $v_\beta$, respectively, output at the end of their executions.
	Next, consider two agents $\alpha$ and $\beta$ that start their executions at nodes labeled $v_\alpha$ and $v_\beta$ in $G$, respectively (and possibly in different rounds). Suppose that there is a round $r_\alpha$ in which agent $\alpha$ knows the $(\kappa\log^*(v_\alpha))$-neighbourhood of node $v_\alpha$ in $G$, and suppose that there is a round $r_\beta$ in which agent $\beta$ knows the $(\kappa\log^*(v_\beta))$-neighbourhood of node $v_\beta$ in $G$ (and note that we may have $r_\alpha \neq r_\beta$). Then, agent $\alpha$ can compute $c_\alpha$ in round $r_\alpha$, and, agent $\beta$ can compute $c_\beta$ in round $r_\beta$.
\end{proposition}

%--------------------------------------------------
%--------------------------------------------------
\section{The canonical line}\label{sec:canon}
%--------------------------------------------------
%--------------------------------------------------

In this section, we describe an algorithm $\rvcanonalg$ that solves rendezvous on the canonical line in time $O(D)$ when two agents start at arbitrary positions and the delay between the rounds in which they start executing the algorithm is arbitrary. The agents do not know the initial distance $D$ between them, and do not know the delay between the starting rounds. 

Denote by $\alpha$ and $\beta$ the two agents. Denote by $v_\alpha$ and $v_\beta$ the starting nodes of $\alpha$ and $\beta$, respectively. Denote by $\mathcal{O}$ the node on the canonical line with the smallest label. For $a \in \{\alpha,\beta\}$, we will write $d(v_a,\mathcal{O})$ to denote the distance between $v_a$ and $\mathcal{O}$.

\subsection{Algorithm description}
Algorithm $\rvcanonalg$ proceeds in phases, numbered starting at 0. Each phase's description has two main components. The first component is a colouring of all nodes on the line. At a high level, in each phase $i \geq 0$, the line is partitioned into segments consisting of $2^i$ consecutive nodes each, and the set of segments is properly coloured, i.e., all nodes in the same segment get the same colour, and two neighbouring nodes in different segments get different colours. The second component describes how an agent behaves when executing the phase. At a high level, the phase consists of equal-sized blocks of rounds, and in each block, an agent either stays idle at its starting node for all rounds in the block, or, it spends the block performing a search of nearby nodes in an attempt to find the other agent (and if not successful, returns back to its starting node). Whether an agent idles or searches in a particular block of the phase depends on the colour of its starting node. The overall idea is: there exists a phase in which the starting nodes of the agents will be coloured differently, and this will result in one agent idling while the other searches, which will result in rendezvous.

The above intuition overlooks two main difficulties. The first difficulty is that the agents do not know the initial distance between them, so they do not know how far they have to search when trying to find the other agent. To deal with this issue, the agents will use larger and larger ``guesses'' in each subsequent phase of the algorithm, and eventually the radius of their search will be large enough. The second difficulty is that the agents do not necessarily start the algorithm in the same round, so the agents' executions of the algorithm (i.e., the phases and blocks) can misalign in arbitrary ways. This makes it difficult to ensure that there is a large enough set of contiguous rounds during which one agent remains idle while the other agent searches. To deal with this issue, we carefully choose the sizes of blocks and phases, as well as the type of behaviour (idle vs. search) carried out in each block.

We now give the full details of an arbitrary phase $i \geq 0$ in the algorithm's execution.

\textbf{Colouring:} When an agent starts executing phase $i$, it first determines the colour of its starting node. From a global perspective, the idea is to assign colours to nodes on the infinite line in the following way: 
\begin{enumerate}
	\item Partition the nodes into segments consisting of $2^i$ nodes each, such that node $\mathcal{O}$ is the leftmost node of its segment. Denote the segment containing $\mathcal{O}$ by $S_0$, denote each subsequent segment to its right using increasing indices (i.e., $S_1,S_2,\ldots$) and denote each subsequent segment to its left using decreasing indices (i.e., $S_{-1}, S_{-2},\ldots$).
	\item For each segment with even index, colour all nodes in the segment with ``red''. For each segment with odd index, colour all nodes in the segment with ``blue''. As a result, neighbouring segments are always assigned different colours.
\end{enumerate}
However, the agent executing phase $i$ does not compute this colouring for the entire line, it need only determine the colour of its starting node. To do so, it uses its knowledge about the distance and direction from its starting node to node $\mathcal{O}$. In particular,
\begin{itemize}
	\item if the agent's starting node $s$ is $\mathcal{O}$ or to the right of $\mathcal{O}$, it computes the index of the segment in which $s$ is contained, i.e., $\mathit{index} = \lfloor d(s,\mathcal{O}) / 2^i \rfloor$. If $\mathit{index}$ is even, then $s$ has colour red, and otherwise $s$ has colour blue.
	\item if the agent's starting node $s$ is to the left of $\mathcal{O}$, it computes the index of the segment in which $s$ is contained, i.e., $\mathit{index} = -\lceil (d(s,\mathcal{O})) / 2^i \rceil$. If $\mathit{index}$ is even, then $s$ has colour red, and otherwise $s$ has colour blue.
\end{itemize}

\textbf{Behaviour:} Phase $i$ consists of $44\cdot 2^{i+1}$ rounds, partitioned into 11 equal-sized blocks of $4\cdot 2^{i+1}$ rounds each. The number $2^{i+1}$ has a special significance: it is the search radius used by an agent during phase $i$ whenever it is searching for the other agent. We use the notation $\sweeprad{i}$ to represent the value $2^{i+1}$ in the remainder of the algorithm's description and analysis. In each of the 11 blocks of the phase, the agent behaves in one of two ways: if a block is designated as a \emph{waiting block}, the agent stays at its starting node $v$ for all $4\cdot \sweeprad{i}$ rounds of the block; otherwise, a block is designated as a \emph{searching block}, in which the agent moves right $\sweeprad{i}$ times, then left $2\cdot\sweeprad{i}$ times, then right $\sweeprad{i}$ times. In other words, during a searching block, the agent explores all nodes within its immediate $\sweeprad{i}$-neighbourhood, and ends up back at its starting node. Whether a block is designated as `waiting' or `searching' depends on the agent's starting node colour in phase $i$. In particular, if its starting node is red, then blocks 1,8,9 are searching blocks and all others are waiting blocks; otherwise, if its starting node is blue, then blocks 1,10,11 are searching blocks and all others are waiting blocks.
This concludes the description of $\rvcanonalg$, whose pseudocode appears in \Cref{canonpseudo}.

\begin{algorithm}
	\small
	\caption{Algorithm \rvcanonalg, executed by an agent $a$ from its starting node $v_a$}
	\label{canonpseudo}
	\begin{algorithmic}[1]
		\ForAll{$i \leftarrow 0,\ldots$} \Comment{phase counter}
		\State {\color{gray} \%\% Determine the colour of my starting node $v_a$}
		\If{located at or to the right of $\mathcal{O}$}
		\State $\mathtt{index} \leftarrow \lfloor d(v_a,\mathcal{O}) / 2^i \rfloor$
		\If{$\mathtt{index}$ is even}
		\State $\mathtt{mycolour} \leftarrow \mathrm{red}$
		\Else
		\State $\mathtt{mycolour} \leftarrow \mathrm{blue}$
		\EndIf
		\Else
		\State $\mathtt{index} \leftarrow -\lceil d(v_a,\mathcal{O})/ 2^i\rceil$
		\If{$\mathtt{index}$ is even}
		\State $\mathtt{mycolour} \leftarrow \mathrm{red}$
		\Else
		\State $\mathtt{mycolour} \leftarrow \mathrm{blue}$
		\EndIf
		\EndIf
		\State {\color{gray} \%\% Execute the blocks of this phase}
		\ForAll{$b \leftarrow 1,\ldots,11$} \Comment{block counter}
		\If{$[(\mathtt{mycolour} = red)$ \textbf{and} $(b \in \{1,8,9\})]$ \textbf{or} $[(\mathtt{mycolour} = blue)$ \textbf{and} $(b \in \{1,10,11\})]$}
		\State {\color{gray} \%\% searching block}
		\State move right $2^{i+1}$ times
		\State move left $2^{i+1}$ times
		\State move left $2^{i+1}$ times
		\State move right $2^{i+1}$ times
		\Else
		\State {\color{gray} \%\% waiting block}
		\State wait at starting node $v_a$ for $4\cdot 2^{i+1}$ rounds
		\EndIf
		\EndFor
		\EndFor
	\end{algorithmic}
\end{algorithm}

\subsection{Algorithm analysis}

In this section, we prove that Algorithm $\rvcanonalg$ solves rendezvous within $O(D)$ rounds for arbitrary starting positions and arbitrary starting rounds. There are two key facts that we exploit in the analysis. The first fact is that the search radius $\sweeprad{i}$ doubles from one phase to the next, which addresses the difficulty that the agents do not know the initial distance between them. More specifically, after some number of phases, there is a phase such that the agent's search radius is sufficiently large to reach the other agent's starting node. We call the first such phase the \emph{critical} phase, and we denote the phase number of the critical phase by $\icrit$. As the distance between the starting nodes is $D$, we can directly compute the value of $\icrit$ by solving $2^{\icrit+1} \geq D$, which yields $\icrit = \lceil\log_2(D)\rceil - 1$. However, recall that the agents do not know the value of $D$, so they cannot determine the value of $\icrit$ while executing the algorithm.

The second fact is that each phase colours the canonical line differently, since the segment size doubles from one phase to the next. In some phases, the starting nodes of the agents are coloured the same, but, we will show that there is at least one phase during which the two starting nodes are coloured differently, which allows us to break symmetry between the behaviours of the agents. More formally, for agent $a \in \{\alpha,\beta\}$, we denote by $\colv{a}{i}$ the colour assigned to the agent $a$'s starting node in phase $i$. We say that a phase $i$ is \emph{asymmetric} if $\colv{\alpha}{i} \neq \colv{\beta}{i}$, and otherwise we say it is \emph{symmetric}. 

The following observations give sufficient conditions for when the critical phase is asymmetric. Of the two starting nodes of the agents, we denote by $\sleft$ the leftmost starting node, and denote by $\sright$ the rightmost starting node. We consider the two cases where the two starting nodes are on the same side of node $\mathcal{O}$. The remaining case (one starting node on each side) will be relatively easy to deal with later.

\begin{lemma}\label{lem:bothright}
	Suppose that both $\sleft$ and $\sright$ are at or to the right of node $\mathcal{O}$. If $D \leq 2^{\icrit+1} - 1 - (d(\sleft,\mathcal{O}) \bmod 2^{\icrit})$, then phase $\icrit$ is asymmetric.
\end{lemma}
\begin{proof}
	Suppose that $D \leq 2^{\icrit+1} - 1 - (d(\sleft,\mathcal{O}) \bmod 2^{\icrit})$.
	
	Recall from the description of the algorithm that if an agent's starting node is at or to the right of node $\mathcal{O}$, then, when executing phase $\icrit$, it determines the colour of its starting node $s$ by calculating the parity of $\lfloor d(s,\mathcal{O}) / 2^{\icrit} \rfloor$. We set out to prove that $\lfloor d(\sright,\mathcal{O}) / 2^{\icrit} \rfloor - \lfloor d(\sleft,\mathcal{O}) / 2^{\icrit} \rfloor = 1$, which is sufficient to prove the result since it implies that the nodes $\sleft$ and $\sright$ are coloured differently in phase $\icrit$.
	
	First, we set out to prove that $\lfloor d(\sright,\mathcal{O}) / 2^{\icrit} \rfloor - \lfloor d(\sleft,\mathcal{O}) / 2^{\icrit} \rfloor \geq 1$. As both agents are located at or to the right of node $\mathcal{O}$, note that the initial distance $D$ between the two agents can be calculated by $d(\sright,\mathcal{O}) - d(\sleft,\mathcal{O})$. Next, by the definition of $\icrit$, it is the smallest integer such that $2^{\icrit+1} \geq D$, which implies that $2^{\icrit} < D$. Therefore, we can conclude that $d(\sright,\mathcal{O}) - d(\sleft,\mathcal{O}) > 2^{\icrit}$. Dividing through by $2^{\icrit}$, we get that $(d(\sright,\mathcal{O}) / 2^{\icrit}) - (d(\sleft,\mathcal{O}) / 2^{\icrit}) > 1$, which implies that $\lfloor d(\sright,\mathcal{O}) / 2^{\icrit} \rfloor - \lfloor d(\sleft,\mathcal{O}) / 2^{\icrit} \rfloor > 0$. Since all terms on the left side are integers, we conclude that $\lfloor d(\sright,\mathcal{O}) / 2^{\icrit} \rfloor - \lfloor d(\sleft,\mathcal{O}) / 2^{\icrit} \rfloor \geq 1$, as desired.
	
	Next, we set out to prove that $\lfloor d(\sright,\mathcal{O}) / 2^{\icrit} \rfloor - \lfloor d(\sleft,\mathcal{O}) / 2^{\icrit} \rfloor \leq 1$. Once again, note that the initial distance $D$ between the two agents can be calculated by $d(\sright,\mathcal{O}) - d(\sleft,\mathcal{O})$. Using the Quotient-Remainder Theorem, we can rewrite each of $d(\sright,\mathcal{O})$ and $d(\sleft,\mathcal{O})$ as follows:
	\begin{itemize}
		\item $d(\sright,\mathcal{O}) = 2^{\icrit}\cdot\lfloor d(\sright,\mathcal{O}) / 2^{\icrit} \rfloor + (d(\sright,\mathcal{O}) \bmod 2^{\icrit})$
		\item $d(\sleft,\mathcal{O}) = 2^{\icrit}\cdot\lfloor d(\sleft,\mathcal{O}) / 2^{\icrit} \rfloor + (d(\sleft,\mathcal{O}) \bmod 2^{\icrit})$
	\end{itemize}
	Substituting these expressions into the equation $d(\sright,\mathcal{O}) - d(\sleft,\mathcal{O}) = D$ gives
%	\begin{align*}
%		& (2^{\icrit}\cdot\lfloor d(\sright,\mathcal{O}) / 2^{\icrit} \rfloor + (d(\sright,\mathcal{O}) \bmod 2^{\icrit}))\\
%		- & (2^{\icrit}\cdot\lfloor d(\sleft,\mathcal{O}) / 2^{\icrit} \rfloor + (d(\sleft,\mathcal{O}) \bmod 2^{\icrit}))
%	\end{align*}
%	Setting this equal to $D$ and re-arranging, we get that
	\begin{align*}
		& 2^{\icrit}\cdot(\lfloor d(\sright,\mathcal{O}) / 2^{\icrit} \rfloor - \lfloor d(\sleft,\mathcal{O}) / 2^{\icrit} \rfloor) \\
		=\ & D - (d(\sright,\mathcal{O}) \bmod 2^{\icrit}) + (d(\sleft,\mathcal{O}) \bmod 2^{\icrit})
	\end{align*}
	By assumption, we have $D \leq 2^{\icrit+1} - 1 - (d(\sleft,\mathcal{O}) \bmod 2^{\icrit})$. Making this substitution into the previous equation and simplifying, we conclude that
	\[
	2^{\icrit}\cdot(\lfloor d(\sright,\mathcal{O}) / 2^{\icrit} \rfloor - \lfloor d(\sleft,\mathcal{O}) / 2^{\icrit} \rfloor) \leq 2^{\icrit+1} - 1 - (d(\sright,\mathcal{O}) \bmod 2^{\icrit})
	\]
	Finally, note that $d(\sright,\mathcal{O}) \bmod 2^{\icrit} \geq 0$, so the right side of this inequality is strictly less than $2^{\icrit+1}$. Dividing both sides by $2^{\icrit}$ gives
	\[
	\lfloor d(\sright,\mathcal{O}) / 2^{\icrit} \rfloor - \lfloor d(\sleft,\mathcal{O}) / 2^{\icrit} \rfloor < 2
	\]
	As all terms on the left side are integers, we get $\lfloor d(\sright,\mathcal{O}) / 2^{\icrit} \rfloor - \lfloor d(\sleft,\mathcal{O}) / 2^{\icrit} \rfloor \leq 1$, as desired.
\end{proof}
\begin{lemma}\label{lem:bothleft}
	Suppose that both $\sleft$ and $\sright$ are to the left of node $\mathcal{O}$. If $D \leq 2^{\icrit+1} - 1 - ((d(\sright,\mathcal{O}) -1)\bmod 2^{\icrit})$, then phase $\icrit$ is asymmetric.
\end{lemma}
\begin{proof}
	Suppose that $D \leq 2^{\icrit+1} - 1 - ((d(\sright,\mathcal{O}) -1)\bmod 2^{\icrit})$.
	
	Recall from the description of the algorithm that if an agent's starting node is to the left of node $\mathcal{O}$, then, when executing phase $\icrit$, it determines the colour of its starting node $s$ by calculating the parity of $-\lceil d(s,\mathcal{O}) / 2^{\icrit} \rceil$. We set out to prove that $(-\lceil d(\sright,\mathcal{O}) / 2^{\icrit} \rceil) - (-\lceil d(\sleft,\mathcal{O}) / 2^{\icrit} \rceil) = 1$, which is sufficient to prove the result since it implies that the nodes $\sleft$ and $\sright$ are coloured differently in phase $\icrit$.
	
	First, we set out to prove that $(-\lceil d(\sright,\mathcal{O}) / 2^{\icrit} \rceil) - (-\lceil d(\sleft,\mathcal{O}) / 2^{\icrit} \rceil) \geq 1$. As both starting nodes are to the left of node $\mathcal{O}$, note that the initial distance $D$ between the two agents can be calculated by $d(\sleft,\mathcal{O}) - d(\sright,\mathcal{O})$. Next, by the definition of $\icrit$, it is the smallest integer such that $2^{\icrit+1} \geq D$, which implies that $2^{\icrit} < D$. Therefore, we can conclude that $d(\sleft,\mathcal{O}) - d(\sright,\mathcal{O}) > 2^{\icrit}$. Dividing through by $2^{\icrit}$, we get that $(d(\sleft,\mathcal{O}) / 2^{\icrit}) - (d(\sright,\mathcal{O}) / 2^{\icrit}) > 1$, which implies that $\lceil d(\sleft,\mathcal{O}) / 2^{\icrit} \rceil - \lceil d(\sright,\mathcal{O}) / 2^{\icrit} \rceil > 0$. Since all terms on the left side are integers, we conclude that $\lceil d(\sleft,\mathcal{O}) / 2^{\icrit} \rceil - \lceil d(\sright,\mathcal{O}) / 2^{\icrit} \rceil \geq 1$.
	This can be re-written as $(- \lceil d(\sright,\mathcal{O}) / 2^{\icrit} \rceil)-(-\lceil d(\sleft,\mathcal{O}) / 2^{\icrit} \rceil)  \geq 1$, as desired.
	
	Next, we set out to prove that $(-\lceil d(\sright,\mathcal{O}) / 2^{\icrit} \rceil) - (-\lceil d(\sleft,\mathcal{O}) / 2^{\icrit} \rceil) \leq 1$. Once again, note that the initial distance $D$ between the two agents is $d(\sleft,\mathcal{O}) - d(\sright,\mathcal{O})$. This expression is equal to $(d(\sleft,\mathcal{O})-1) - (d(\sright,\mathcal{O})-1)$. Using the Quotient-Remainder Theorem, we can rewrite each of $d(\sleft,\mathcal{O})-1$ and $d(\sright,\mathcal{O})-1$ as follows:
	\begin{itemize}
		\item $d(\sleft,\mathcal{O})-1 = 2^{\icrit}\cdot\lfloor (d(\sleft,\mathcal{O})-1) / 2^{\icrit} \rfloor + ((d(\sleft,\mathcal{O})-1) \bmod 2^{\icrit})$
		\item $d(\sright,\mathcal{O})-1 = 2^{\icrit}\cdot\lfloor (d(\sright,\mathcal{O})-1) / 2^{\icrit} \rfloor + ((d(\sright,\mathcal{O})-1) \bmod 2^{\icrit})$
	\end{itemize}
	Substituting these expressions into the equation $(d(\sleft,\mathcal{O})-1) - (d(\sright,\mathcal{O})-1) = D$ gives
%	{\footnotesize
%		\[
%		(2^{\icrit}\cdot\lfloor (d(\sleft,\mathcal{O})-1) / 2^{\icrit} \rfloor + ((d(\sleft,\mathcal{O})-1) \bmod 2^{\icrit})) - (2^{\icrit}\cdot\lfloor (d(\sright,\mathcal{O})-1) / 2^{\icrit} \rfloor + ((d(\sright,\mathcal{O})-1) \bmod 2^{\icrit}))
%		\]
%	}
%	Setting this equal to $D$ and re-arranging, we get that
%	{\footnotesize
		\begin{align*}
		& 2^{\icrit}\cdot(\lfloor (d(\sleft,\mathcal{O})-1) / 2^{\icrit} \rfloor - \lfloor (d(\sright,\mathcal{O})-1) / 2^{\icrit} \rfloor)\\
		=\ & D - ((d(\sleft,\mathcal{O})-1) \bmod 2^{\icrit}) + ((d(\sright,\mathcal{O})-1) \bmod 2^{\icrit})
		\end{align*}
%	}
	Using the identity $\lfloor (n-1)/m \rfloor = \lceil n/m \rceil-1$ and simplifying, the equation becomes
%	{\footnotesize
		\begin{align*}
		& 2^{\icrit}\cdot(\lceil d(\sleft,\mathcal{O}) / 2^{\icrit} \rceil - \lceil d(\sright,\mathcal{O}) / 2^{\icrit} \rceil) \\
		=\ & D - ((d(\sleft,\mathcal{O})-1) \bmod 2^{\icrit}) + ((d(\sright,\mathcal{O})-1) \bmod 2^{\icrit})
		\end{align*}
%	}
	By assumption, we have $D \leq 2^{\icrit+1} - 1 - ((d(\sright,\mathcal{O}) -1)\bmod 2^{\icrit})$. Making this substitution into the previous equation and simplifying, we conclude that
	{\small
		\[
		2^{\icrit}\cdot(\lceil d(\sleft,\mathcal{O}) / 2^{\icrit} \rceil - \lceil d(\sright,\mathcal{O}) / 2^{\icrit} \rceil) \leq 2^{\icrit+1} - 1 - ((d(\sleft,\mathcal{O})-1) \bmod 2^{\icrit})
		\]
	}
	Finally, note that $d(\sleft,\mathcal{O}) \bmod 2^{\icrit} \geq 0$, so the right side of this inequality is strictly less than $2^{\icrit+1}$. Thus, we have
			\[
	2^{\icrit}\cdot(\lceil d(\sleft,\mathcal{O}) / 2^{\icrit} \rceil - \lceil d(\sright,\mathcal{O}) / 2^{\icrit} \rceil) < 2^{\icrit+1}
	\]
	Dividing both sides by $2^{\icrit}$ gives
	\[
	\lceil d(\sleft,\mathcal{O}) / 2^{\icrit} \rceil - \lceil d(\sright,\mathcal{O}) / 2^{\icrit} \rceil < 2
	\]
	As all terms on the left side are integers, we get $\lceil d(\sleft,\mathcal{O}) / 2^{\icrit} \rceil - \lceil d(\sright,\mathcal{O}) / 2^{\icrit} \rceil \leq 1$. This inequality can be re-written to give the desired inequality $(- \lceil d(\sright,\mathcal{O}) / 2^{\icrit} \rceil) - (-\lceil d(\sleft,\mathcal{O}) / 2^{\icrit} \rceil)  \leq 1$.
\end{proof}

In cases where the critical phase is symmetric, there are situations where rendezvous cannot be solved by the end of phase $\icrit$ (e.g., if the two agents start the algorithm in the same round, their behaviour during phase $\icrit$ will be exactly the same, so rendezvous will not occur). However, we next observe that if phase $\icrit$ is symmetric, then we can guarantee that phase $\icrit+1$ is asymmetric, and we will later use this fact to prove that rendezvous occurs.
\begin{lemma}\label{nextasym}
	If phase $\icrit$ is symmetric, then phase $\icrit+1$ is asymmetric.
\end{lemma}
\begin{proof}
	We consider three cases depending on the relative positions of $\sleft$ and $\sright$ with respect to $\mathcal{O}$. 
	\begin{itemize}
		\item \textbf{Case 1:  both $\sleft$ and $\sright$ are at or to the right of node $\mathcal{O}$}
		
		Suppose that phase $\icrit$ is symmetric. By \Cref{lem:bothright}, it follows that $D > 2^{\icrit+1} - 1 - (d(\sleft,\mathcal{O}) \bmod 2^{\icrit})$.
		
		Recall from the description of the algorithm that if an agent's starting node is at or to the right of node $\mathcal{O}$, then, when executing phase $\icrit+1$, it determines the colour of its starting node $s$ by calculating the parity of $\lfloor d(s,\mathcal{O}) / 2^{\icrit+1} \rfloor$. We set out to prove that $\lfloor d(\sright,\mathcal{O}) / 2^{\icrit+1} \rfloor - \lfloor d(\sleft,\mathcal{O}) / 2^{\icrit+1} \rfloor = 1$, which is sufficient to prove the result since it implies that the nodes $\sleft$ and $\sright$ are coloured differently in phase $\icrit+1$.
		
		First, we set out to prove that $\lfloor d(\sright,\mathcal{O}) / 2^{\icrit+1} \rfloor - \lfloor d(\sleft,\mathcal{O}) / 2^{\icrit+1} \rfloor \leq 1$. As both agents are located at or to the right of node $\mathcal{O}$, note that the initial distance $D$ between the two agents can be calculated by $d(\sright,\mathcal{O}) - d(\sleft,\mathcal{O})$. Next, by the definition of $\icrit$, it is the smallest integer such that $2^{\icrit+1} \geq D$. Therefore, we can conclude that $d(\sright,\mathcal{O}) - d(\sleft,\mathcal{O}) \leq 2^{\icrit+1}$. Dividing through by $2^{\icrit+1}$, we get that $(d(\sright,\mathcal{O}) / 2^{\icrit+1}) - (d(\sleft,\mathcal{O}) / 2^{\icrit+1}) \leq 1$, which implies the inequality $\lfloor d(\sright,\mathcal{O}) / 2^{\icrit+1} \rfloor - \lfloor d(\sleft,\mathcal{O}) / 2^{\icrit+1} \rfloor \leq 1$, as desired.
		
		Next, we set out to prove that $\lfloor d(\sright,\mathcal{O}) / 2^{\icrit+1} \rfloor - \lfloor d(\sleft,\mathcal{O}) / 2^{\icrit+1} \rfloor \geq 1$. Once again, note that the initial distance $D$ between the two agents can be calculated by $d(\sright,\mathcal{O}) - d(\sleft,\mathcal{O})$. Using the Quotient-Remainder Theorem, we can rewrite each of $d(\sright,\mathcal{O})$ and $d(\sleft,\mathcal{O})$ as follows:
		\begin{itemize}
			\item[$\bullet$] $d(\sright,\mathcal{O}) = 2^{\icrit+1}\cdot\lfloor d(\sright,\mathcal{O}) / 2^{\icrit+1} \rfloor + (d(\sright,\mathcal{O}) \bmod 2^{\icrit+1})$
			\item[$\bullet$] $d(\sleft,\mathcal{O}) = 2^{\icrit+1}\cdot\lfloor d(\sleft,\mathcal{O}) / 2^{\icrit+1} \rfloor + (d(\sleft,\mathcal{O}) \bmod 2^{\icrit+1})$
		\end{itemize}
		Substituting these expressions into the equation $d(\sright,\mathcal{O}) - d(\sleft,\mathcal{O}) = D$ gives
%		{\footnotesize
%			\[
%			(2^{\icrit+1}\cdot\lfloor d(\sright,\mathcal{O}) / 2^{\icrit+1} \rfloor + (d(\sright,\mathcal{O}) \bmod 2^{\icrit+1})) - (2^{\icrit+1}\cdot\lfloor d(\sleft,\mathcal{O}) / 2^{\icrit+1} \rfloor + (d(\sleft,\mathcal{O}) \bmod 2^{\icrit+1}))
%			\]
%		}
%		Setting this equal to $D$ and re-arranging, we get that
%		{\footnotesize
			\begin{align*}
			& 2^{\icrit+1}\cdot(\lfloor d(\sright,\mathcal{O}) / 2^{\icrit+1} \rfloor - \lfloor d(\sleft,\mathcal{O}) / 2^{\icrit+1} \rfloor)\\
			=\ & D - (d(\sright,\mathcal{O}) \bmod 2^{\icrit+1}) + (d(\sleft,\mathcal{O}) \bmod 2^{\icrit+1})
			\end{align*}
%			}
		By assumption, we have $D > 2^{\icrit+1} - 1 - (d(\sleft,\mathcal{O}) \bmod 2^{\icrit})$, and re-arranging gives $(d(\sleft,\mathcal{O}) \bmod 2^{\icrit}) > 2^{\icrit+1} - 1 - D$. Applying the general fact that $(k \bmod 2^{j+1}) \geq (k \bmod 2^{j})$, we conclude that $(d(\sleft,\mathcal{O}) \bmod 2^{\icrit+1}) > 2^{\icrit+1} - 1 - D$.  Making this substitution to the previous equation and simplifying, we conclude that
		\begin{align*}
		& 2^{\icrit+1}\cdot(\lfloor d(\sright,\mathcal{O}) / 2^{\icrit+1} \rfloor - \lfloor d(\sleft,\mathcal{O}) / 2^{\icrit+1} \rfloor)\\
		>\ & 2^{\icrit+1} - 1 - (d(\sright,\mathcal{O}) \bmod 2^{\icrit+1})
		\end{align*}
		Finally, note that $(d(\sright,\mathcal{O}) \bmod 2^{\icrit+1}) \leq 2^{\icrit+1}-1$, so we get 
		\[
		2^{\icrit+1}\cdot(\lfloor d(\sright,\mathcal{O}) / 2^{\icrit+1} \rfloor - \lfloor d(\sleft,\mathcal{O}) / 2^{\icrit+1} \rfloor) > 0
		\]
%		
%		As the left side is an integer, we have
%		
%		\[
%		2^{\icrit+1}\cdot(\lfloor d(\sright,\mathcal{O}) / 2^{\icrit+1} \rfloor - \lfloor d(\sleft,\mathcal{O}) / 2^{\icrit+1} \rfloor) \geq 1
%		\]
		Dividing both sides by $2^{\icrit+1}$ gives
		\[
		\lfloor d(\sright,\mathcal{O}) / 2^{\icrit+1} \rfloor - \lfloor d(\sleft,\mathcal{O}) / 2^{\icrit+1} \rfloor > 0
		\]
		Both terms on the left side are integers, so $\lfloor d(\sright,\mathcal{O}) / 2^{\icrit+1} \rfloor - \lfloor d(\sleft,\mathcal{O}) / 2^{\icrit+1} \rfloor \geq 1$, as desired.
		
		\item \textbf{Case 2:  both $\sleft$ and $\sright$ are to the left of node $\mathcal{O}$}
		
		Suppose that phase $\icrit$ is symmetric. By \Cref{lem:bothleft}, it follows that $D > 2^{\icrit+1} - 1 - ((d(\sright,\mathcal{O}) -1)\bmod 2^{\icrit})$.
		
		Recall from the description of the algorithm that if an agent's starting node is to the left of node $\mathcal{O}$, then, when executing phase $\icrit+1$, it determines the colour of its starting node $s$ by calculating the parity of $-\lceil d(s,\mathcal{O}) / 2^{\icrit+1} \rceil$. We set out to prove that $(-\lceil d(\sright,\mathcal{O}) / 2^{\icrit+1} \rceil) - (-\lceil d(\sleft,\mathcal{O}) / 2^{\icrit+1} \rceil) = 1$, which is sufficient to prove the result since it implies that the nodes $\sleft$ and $\sright$ are coloured differently in phase $\icrit+1$.
		
		First, we prove that $(-\lceil d(\sright,\mathcal{O}) / 2^{\icrit+1} \rceil) - (-\lceil d(\sleft,\mathcal{O}) / 2^{\icrit+1} \rceil) \leq 1$. As both starting nodes are to the left of node $\mathcal{O}$, note that the initial distance $D$ between the two agents can be calculated by $d(\sleft,\mathcal{O}) - d(\sright,\mathcal{O})$. Next, by the definition of $\icrit$, it is the smallest integer such that $2^{\icrit+1} \geq D$. Therefore, $d(\sleft,\mathcal{O}) - d(\sright,\mathcal{O}) \leq 2^{\icrit+1}$. Dividing through by $2^{\icrit+1}$, we get that $(d(\sleft,\mathcal{O}) / 2^{\icrit+1}) - (d(\sright,\mathcal{O}) / 2^{\icrit+1}) \leq 1$, which implies that $\lceil d(\sleft,\mathcal{O}) / 2^{\icrit+1} \rceil - \lceil d(\sright,\mathcal{O}) / 2^{\icrit+1} \rceil \leq 1$. This can be re-written as $(- \lceil d(\sright,\mathcal{O}) / 2^{\icrit+1} \rceil)-(-\lceil d(\sleft,\mathcal{O}) / 2^{\icrit+1} \rceil)  \leq 1$, as desired.
		
		Next, we prove that $(-\lceil d(\sright,\mathcal{O}) / 2^{\icrit+1} \rceil) - (-\lceil d(\sleft,\mathcal{O}) / 2^{\icrit+1} \rceil) \geq 1$. Once again, note that the initial distance $D$ between the two agents can be calculated by $d(\sleft,\mathcal{O}) - d(\sright,\mathcal{O})$. This expression is equal to $(d(\sleft,\mathcal{O})-1) - (d(\sright,\mathcal{O})-1)$. Using the Quotient-Remainder Theorem, we can rewrite each of $d(\sleft,\mathcal{O})-1$ and $d(\sright,\mathcal{O})-1$ as follows:
		\begin{itemize}
			\item[$\bullet$] $d(\sleft,\mathcal{O})-1 = 2^{\icrit+1}\cdot\lfloor (d(\sleft,\mathcal{O})-1) / 2^{\icrit+1} \rfloor + ((d(\sleft,\mathcal{O})-1) \bmod 2^{\icrit+1})$
			\item[$\bullet$] $d(\sright,\mathcal{O})-1 = 2^{\icrit+1}\cdot\lfloor (d(\sright,\mathcal{O})-1) / 2^{\icrit+1} \rfloor + ((d(\sright,\mathcal{O})-1) \bmod 2^{\icrit+1})$
		\end{itemize}
		Substituting these into the equation $(d(\sleft,\mathcal{O})-1) - (d(\sright,\mathcal{O})-1) = D$ gives
%		{\scriptsize
%			\[
%			(2^{\icrit+1}\cdot\lfloor (d(\sleft,\mathcal{O})-1) / 2^{\icrit+1} \rfloor + ((d(\sleft,\mathcal{O})-1) \bmod 2^{\icrit+1})) - (2^{\icrit+1}\cdot\lfloor (d(\sright,\mathcal{O})-1) / 2^{\icrit+1} \rfloor + ((d(\sright,\mathcal{O})-1) \bmod 2^{\icrit+1}))
%			\]
%		}
%		Setting this equal to $D$ and re-arranging, we get that
%		{\scriptsize
			\begin{align*}
			& 2^{\icrit+1}\cdot(\lfloor (d(\sleft,\mathcal{O})-1) / 2^{\icrit+1} \rfloor - \lfloor (d(\sright,\mathcal{O})-1) / 2^{\icrit+1} \rfloor)\\
			 =\ & D - ((d(\sleft,\mathcal{O})-1) \bmod 2^{\icrit+1}) + ((d(\sright,\mathcal{O})-1) \bmod 2^{\icrit+1})
			\end{align*}
%		}
		Using the identity $\lfloor (n-1)/m \rfloor = \lceil n/m \rceil-1$ and simplifying, the equation becomes
%		{\scriptsize
			\begin{align*}
			& 2^{\icrit+1}\cdot(\lceil d(\sleft,\mathcal{O}) / 2^{\icrit+1} \rceil - \lceil d(\sright,\mathcal{O}) / 2^{\icrit+1} \rceil)\\
			 =\ & D - ((d(\sleft,\mathcal{O})-1) \bmod 2^{\icrit+1})) + ((d(\sright,\mathcal{O})-1) \bmod 2^{\icrit+1}))
			\end{align*}
%		}
		By assumption, we have $D > 2^{\icrit+1} - 1 - ((d(\sright,\mathcal{O}) -1)\bmod 2^{\icrit})$, and re-arranging gives $((d(\sright,\mathcal{O})-1) \bmod 2^{\icrit}) > 2^{\icrit+1} - 1 - D$. Further, applying the general fact that $(k \bmod 2^{j+1}) \geq (k \bmod 2^{j})$, we conclude that $((d(\sright,\mathcal{O})-1) \bmod 2^{\icrit+1}) > 2^{\icrit+1} - 1 - D$.  Making this substitution into the previous equation and simplifying, we conclude that
		\begin{align*}
		& 2^{\icrit+1}\cdot(\lceil d(\sleft,\mathcal{O}) / 2^{\icrit+1} \rceil - \lceil d(\sright,\mathcal{O}) / 2^{\icrit+1} \rceil)\\
		>\ & 2^{\icrit+1} - 1 - ((d(\sleft,\mathcal{O})-1) \bmod 2^{\icrit+1})
		\end{align*}

		Finally, note that $d(\sleft,\mathcal{O}) \bmod 2^{\icrit} \leq 2^{\icrit+1} - 1$, so
		\[
		2^{\icrit+1}\cdot(\lceil d(\sleft,\mathcal{O}) / 2^{\icrit+1} \rceil - \lceil d(\sright,\mathcal{O}) / 2^{\icrit+1} \rceil) > 0
		\]
%		As the left side of this inequality is an integer, we conclude that
%		\[
%		2^{\icrit+1}\cdot(\lceil d(\sleft,\mathcal{O}) / 2^{\icrit+1} \rceil - \lceil d(\sright,\mathcal{O}) / 2^{\icrit+1} \rceil) \geq 1
%		\]
		Dividing both sides by $2^{\icrit+1}$ gives
		\[
		\lceil d(\sleft,\mathcal{O}) / 2^{\icrit+1} \rceil - \lceil d(\sright,\mathcal{O}) / 2^{\icrit+1} \rceil > 0
		\]
		All terms on the left side are integers, so $\lceil d(\sleft,\mathcal{O}) / 2^{\icrit+1} \rceil - \lceil d(\sright,\mathcal{O}) / 2^{\icrit+1} \rceil \geq 1$. This can be written as $(- \lceil d(\sright,\mathcal{O}) / 2^{\icrit+1} \rceil) - (-\lceil d(\sleft,\mathcal{O}) / 2^{\icrit+1} \rceil)  \geq 1$, as desired.
		\item \textbf{Case 3: $\sleft$ is to the left of node $\mathcal{O}$, and $\sright$ is at or to the right of node $\mathcal{O}$}
		
		We prove that phase $\icrit+1$ is asymmetric. The value of $D$ can also be written as $d(\sleft,\mathcal{O}) + d(\sright,\mathcal{O})$. Also, by the definition of $\icrit$, we know that $2^{\icrit+1} \geq D$. So we conclude that $d(\sleft,\mathcal{O}) + d(\sright,\mathcal{O}) \leq 2^{\icrit+1}$. From this fact, we get:
		\begin{itemize}
			\item $d(\sleft,\mathcal{O}) \leq 2^{\icrit+1}$ since $d(\sright,\mathcal{O}) \geq 0$.
			\item $d(\sright,\mathcal{O}) < 2^{\icrit+1}$ since $d(\sleft,\mathcal{O}) > 0$ because $\sleft$ is to the left of $\mathcal{O}$.
		\end{itemize}
		From the algorithm's description, the agent that starts at $\sleft$ determines the colour of its starting node by computing the parity of $- \lceil d(\sleft,\mathcal{O}) / 2^{\icrit+1} \rceil)$. But we showed above that $d(\sleft,\mathcal{O}) \leq 2^{\icrit+1}$, so $- \lceil d(\sleft,\mathcal{O}) / 2^{\icrit+1} \rceil)$ evaluates to -1. From the algorithm's description, this means that the agent determines that its starting node's colour is blue.
		
		From the algorithm's description, the agent that starts at $\sright$ determines the colour of its starting node by computing the parity of $\lfloor d(\sright,\mathcal{O}) / 2^{\icrit+1} \rfloor$. But we showed above that $d(\sright,\mathcal{O}) < 2^{\icrit+1}$, so $\lfloor d(\sright,\mathcal{O}) / 2^{\icrit+1} \rfloor$ evaluates to 0. From the algorithm's description, this means that the agent determines that its starting node's colour is red.
		
		Since the two starting nodes are assigned different colours in phase $\icrit+1$, it follows that phase $\icrit+1$ is asymmetric, as desired.
	\end{itemize}
\end{proof}

\begin{theorem}
	Algorithm $\rvcanonalg$ solves rendezvous in $O(D)$ rounds on the canonical line, for arbitrary start positions and arbitrary start delay.
\end{theorem}
\begin{proof}
	The agent that starts executing the algorithm first is called the \emph{early agent}, and the other agent is called the \emph{late agent}. If both agents start executing the algorithm in the same round, then arbitrarily call one of them early and the other one late. Without loss of generality, we assume that $\alpha$ is the early agent. The number of rounds that elapse between the two starting rounds is denoted by $\delay$.   
	
	First, we address the case where the delay between the starting rounds of the agents is extremely large. This situation is relatively easy to analyze, since agent $\beta$ remains idle during agent $\alpha$'s entire execution up until they meet.
	\begin{claim}
		If $\delay > 48\cdot \sweeprad{\icrit}$, then rendezvous occurs by the end of phase $\icrit$ in agent $\alpha$'s execution.
	\end{claim}
	To prove the claim, we note that the total number of rounds that elapse in agent $\alpha$'s execution during phases $0,\ldots,\icrit-1$ is $\sum_{i=0}^{\icrit-1} 44\cdot\sweeprad{i} = 44\sum_{i=0}^{\icrit-1} 2^{i+1} < 44\cdot 2^{\icrit+1} = 44\cdot\sweeprad{\icrit}$. The first block of phase $\icrit$ is a searching block, so agent $\alpha$ explores its $\sweeprad{\icrit}$-neighbourhood during the $4\cdot\sweeprad{\icrit}$ rounds of the block. Since $D \leq \sweeprad{\icrit}$, it follows that agent $\alpha$ will be located at agent $\beta$'s starting node during this block. This occurs within the first $48\cdot\sweeprad{\icrit}$ rounds of $\alpha$'s execution, and by our assumption on $\delay$, we know that agent $\beta$ is located at its starting node during all of these rounds. It follows that rendezvous occurs during phase $\icrit$ of $\alpha$'s execution, which completes the proof of the claim.
	
	So, in the remainder of the proof, we assume that $\delay \leq 48\cdot \sweeprad{\icrit}$. We separate the proof into cases based on whether phase $\icrit$ is asymmetric or symmetric.
	\begin{itemize}
		\item \textbf{Case 1: Phase $\icrit$ is asymmetric}\\
		We consider subcases based on the delay between the agent's start times. 
		\begin{itemize}
			\item \underline{Case (i): $0 \leq \delay \leq 4\cdot\sweeprad{\icrit}$}
			
			There are two subcases, depending on the value of $\colv{\alpha}{\icrit}$. 
			
			First, suppose that $\colv{\alpha}{\icrit} = \mathrm{red}$, and consider the start of block 10 of $\beta$'s execution of phase $\icrit$. By the assumed condition on $\delay$, agent $\alpha$'s execution is between $0$ and $4\cdot\sweeprad{\icrit}$ rounds ahead. Since the block length in phase $\icrit$ is $4\cdot\sweeprad{\icrit}$ rounds, it follows that $\alpha$'s execution is at some round of block 10 of phase $\icrit$. As blocks 10 and 11 are waiting blocks for agent $\alpha$ in this phase, $\alpha$ will be idle at its starting node for the next $4\cdot\sweeprad{\icrit}$ rounds. Since phase $\icrit$ is asymmetric, it follows that $\colv{\beta}{\icrit} = \mathrm{blue}$, which means block 10 is a searching block for agent $\beta$. Therefore, in the next $4\cdot\sweeprad{\icrit}$ rounds of agent $\beta$'s execution, $\beta$ will search all nodes in the $\sweeprad{\icrit}$-neighbourhood of its starting node. As $D \leq \sweeprad{\icrit}$, it follows that $\beta$ will visit $\alpha$'s starting node in some round of block 10 or block 11 of $\alpha$'s execution of phase $\icrit$. Thus, rendezvous will occur by the end of agent $\alpha$'s execution of phase $\icrit$.
			
			Next, suppose that $\colv{\alpha}{\icrit} = \mathrm{blue}$, and consider the start of block 8 of $\beta$'s execution of phase $\icrit$. By the assumed condition on $\delay$, agent $\alpha$'s execution is between $0$ and $4\cdot\sweeprad{\icrit}$ rounds ahead. Since the block length in phase $\icrit$ is $4\cdot\sweeprad{\icrit}$ rounds, it follows that $\alpha$'s execution is at some round of block 8 of phase $\icrit$. As blocks 8 and 9 are waiting blocks for agent $\alpha$ in this phase, $\alpha$ will be idle at its starting node for the next $4\cdot\sweeprad{\icrit}$ rounds. Since phase $\icrit$ is asymmetric, it follows that $\colv{\beta}{\icrit} = \mathrm{red}$, which means block 8 is a searching block for agent $\beta$. Therefore, in the next $4\cdot\sweeprad{\icrit}$ rounds of agent $\beta$'s execution, $\beta$ will search all nodes in the $\sweeprad{\icrit}$-neighbourhood of its starting node. As $D \leq \sweeprad{\icrit}$, it follows that $\beta$ will visit $\alpha$'s starting node in some round of block 8 or block 9 of $\alpha$'s execution of phase $\icrit$. Thus, rendezvous will occur by the end of agent $\alpha$'s execution of phase $\icrit$.
			
			\item \underline{Case (ii): $4\cdot\sweeprad{\icrit} < \delay \leq 8\cdot\sweeprad{\icrit}$}
			
			Consider the start of block 1 of $\beta$'s execution of phase $\icrit$. By the assumed condition on $\delay$, agent $\alpha$'s execution is more than $4\cdot\sweeprad{\icrit}$ rounds ahead and at most $8\cdot\sweeprad{\icrit}$ rounds ahead. Since the block length in phase $\icrit$ is $4\cdot\sweeprad{\icrit}$, it follows that $\alpha$'s execution is more than 1 block ahead and at most 2 blocks ahead. In particular, it follows that $\alpha$'s execution is at some round in block 2 of phase $\icrit$. As blocks 2 and 3 are waiting blocks, it follows that $\alpha$ will be located at its starting node for the next $4\cdot\sweeprad{\icrit}$ rounds. But, block 1 is a searching block for $\beta$. Therefore, in the next $4\cdot\sweeprad{\icrit}$ rounds of agent $\beta$'s execution, $\beta$ will search all nodes in the $\sweeprad{\icrit}$-neighbourhood of its starting node. As $D \leq \sweeprad{\icrit}$, it follows that $\beta$ will visit $\alpha$'s starting node in some round of block 2 or 3 of $\alpha$'s execution of phase $\icrit$. Thus, rendezvous will occur by the end of agent $\alpha$'s execution of phase $\icrit$.
			
			\item \underline{Case (iii): $8\cdot\sweeprad{\icrit} < \delay \leq 48\cdot\sweeprad{\icrit}$}
			
			Consider the start of block 1 of $\beta$'s execution of phase $\icrit+1$. By the assumed condition on $\delay$, agent $\alpha$'s execution is more than $8\cdot\sweeprad{\icrit}$ rounds ahead and at most $48\cdot\sweeprad{\icrit}$ rounds ahead. Since the block length in phase $\icrit+1$ is $4\cdot\sweeprad{\icrit+1} = 8\cdot\sweeprad{\icrit}$, it follows that $\alpha$'s execution is more than 1 block ahead and at most 6 blocks ahead. In particular, $\alpha$'s execution is in some round of block 2, 3, 4, 5, 6 of phase $\icrit+1$. As blocks 2 through 7 are waiting blocks, it follows that $\alpha$ will be located at its starting node for the next $4\cdot\sweeprad{\icrit+1}$ rounds. But, block 1 is a searching block for $\beta$. Therefore, in the next $4\cdot\sweeprad{\icrit+1}$ rounds of agent $\beta$'s execution, $\beta$ will search all nodes in the $\sweeprad{\icrit+1}$-neighbourhood of its starting node. As $D \leq \sweeprad{\icrit} < \sweeprad{\icrit+1}$, it follows that $\beta$ will visit $\alpha$'s starting node in some round of one of the blocks 2-7 of $\alpha$'s execution of phase $\icrit+1$. Thus, rendezvous will occur by the end of agent $\alpha$'s execution of phase $\icrit+1$.
			
		\end{itemize}
		\item \textbf{Case 2: Phase $\icrit$ is symmetric}\\
		By \Cref{nextasym}, phase $\icrit + 1$ is asymmetric. We consider subcases based on the delay between the agent's start times. 
		\begin{itemize}
			\item \underline{Case (i): $0 \leq \delay \leq 8\cdot\sweeprad{\icrit}$}
			
			There are two subcases, depending on the value of $\colv{\alpha}{\icrit+1}$. 
			
			First, suppose that $\colv{\alpha}{\icrit+1} = \mathrm{red}$, and consider the start of block 10 of $\beta$'s execution of phase $\icrit+1$. By the assumed condition on $\delay$, agent $\alpha$'s execution is between $0$ and $8\cdot\sweeprad{\icrit}$ rounds ahead. Since the block length in phase $\icrit+1$ is $4\cdot\sweeprad{\icrit+1} = 8\cdot\sweeprad{\icrit}$ rounds, it follows that $\alpha$'s execution is at some round of block 10 of phase $\icrit+1$. As blocks 10 and 11 are waiting blocks for agent $\alpha$ in this phase, $\alpha$ will be idle at its starting node for the next $4\cdot\sweeprad{\icrit+1}$ rounds. Since phase $\icrit+1$ is asymmetric, it follows that $\colv{\beta}{\icrit+1} = \mathrm{blue}$, which means block 10 is a searching block for agent $\beta$. Therefore, in the next $4\cdot\sweeprad{\icrit+1}$ rounds of agent $\beta$'s execution, $\beta$ will search all nodes in the $\sweeprad{\icrit+1}$-neighbourhood of its starting node. As $D \leq \sweeprad{\icrit} < \sweeprad{\icrit+1}$, it follows that $\beta$ will visit $\alpha$'s starting node in some round of block 10 or block 11 of $\alpha$'s execution of phase $\icrit+1$. Thus, rendezvous will occur by the end of agent $\alpha$'s execution of phase $\icrit+1$.
			
			Next, suppose that $\colv{\alpha}{\icrit+1} = \mathrm{blue}$, and consider the start of block 8 of $\beta$'s execution of phase $\icrit+1$. By the assumed condition on $\delay$, agent $\alpha$'s execution is between $0$ and $8\cdot\sweeprad{\icrit}$ rounds ahead. Since the block length in phase $\icrit+1$ is $4\cdot\sweeprad{\icrit+1} = 8\cdot\sweeprad{\icrit}$ rounds, it follows that $\alpha$'s execution is at some round of block 8 of phase $\icrit+1$. As blocks 8 and 9 are waiting blocks for agent $\alpha$ in this phase, $\alpha$ will be idle at its starting node for the next $4\cdot\sweeprad{\icrit+1}$ rounds. Since phase $\icrit+1$ is asymmetric, it follows that $\colv{\beta}{\icrit+1} = \mathrm{red}$, which means block 8 is a searching block for agent $\beta$. Therefore, in the next $4\cdot\sweeprad{\icrit+1}$ rounds of agent $\beta$'s execution, $\beta$ will search all nodes in the $\sweeprad{\icrit+1}$-neighbourhood of its starting node. As $D \leq \sweeprad{\icrit} < \sweeprad{\icrit+1}$, it follows that $\beta$ will visit $\alpha$'s starting node in some round of block 8 or block 9 of $\alpha$'s execution of phase $\icrit+1$. Thus, rendezvous will occur by the end of agent $\alpha$'s execution of phase $\icrit+1$.

			\item \underline{Case (ii): $8\cdot\sweeprad{\icrit} < \delay \leq 16\cdot\sweeprad{\icrit}$}
			
			Consider the start of block 1 of $\beta$'s execution of phase $\icrit+1$. By the assumed condition on $\delay$, agent $\alpha$'s execution is more than $8\cdot\sweeprad{\icrit}$ rounds ahead and at most $16\cdot\sweeprad{\icrit}$ rounds ahead. Since the block length in phase $\icrit+1$ is $4\cdot\sweeprad{\icrit+1}=8\cdot\sweeprad{\icrit}$, it follows that $\alpha$'s execution is more than 1 block ahead and at most 2 blocks ahead. In particular, it follows that $\alpha$'s execution is at some round in block 2 of phase $\icrit+1$. As blocks 2 and 3 are waiting blocks, it follows that $\alpha$ will be located at its starting node for the next $4\cdot\sweeprad{\icrit+1}$ rounds. But, block 1 is a searching block for $\beta$. Therefore, in the next $4\cdot\sweeprad{\icrit+1}$ rounds of agent $\beta$'s execution, $\beta$ will search all nodes in the $\sweeprad{\icrit+1}$-neighbourhood of its starting node. As $D \leq \sweeprad{\icrit}<\sweeprad{\icrit+1}$, it follows that $\beta$ will visit $\alpha$'s starting node in some round of block 2 or 3 of $\alpha$'s execution of phase $\icrit+1$. Thus, rendezvous will occur by the end of agent $\alpha$'s execution of phase $\icrit+1$.
			
			\item \underline{Case (iii): $16\cdot\sweeprad{\icrit} < \delay \leq 48\cdot\sweeprad{\icrit}$}
			
			Consider the start of block 1 of $\beta$'s execution of phase $\icrit+2$. By the assumed condition on $\delay$, agent $\alpha$'s execution is more than $16\cdot\sweeprad{\icrit}$ rounds ahead and at most $48\cdot\sweeprad{\icrit}$ rounds ahead. Since the block length in phase $\icrit+2$ is $4\cdot\sweeprad{\icrit+2} = 16\cdot\sweeprad{\icrit}$, it follows that $\alpha$'s execution is more than 1 block ahead and at most 3 blocks ahead. In particular, $\alpha$'s execution is in some round of block 2 or 3 of phase $\icrit+2$. As blocks 2, 3 and 4 are waiting blocks, it follows that $\alpha$ will be located at its starting node for the next $4\cdot\sweeprad{\icrit+1}$ rounds. But, block 1 is a searching block for $\beta$. Therefore, in the next $4\cdot\sweeprad{\icrit+1}$ rounds of agent $\beta$'s execution, $\beta$ will search all nodes in the $\sweeprad{\icrit+1}$-neighbourhood of its starting node. As $D \leq \sweeprad{\icrit} < \sweeprad{\icrit+1}$, it follows that $\beta$ will visit $\alpha$'s starting node in some round of one of the blocks 2,3 or 4 of $\alpha$'s execution of phase $\icrit+1$. Thus, rendezvous will occur by the end of agent $\alpha$'s execution of phase $\icrit+1$.
			
		\end{itemize}
	\end{itemize}
	In all cases, we proved that rendezvous occurs by the end of phase $\icrit+2$ in agent $\alpha$'s execution. The total number of rounds that elapse in phases $0,\ldots,\icrit+2$ is $\sum_{i=0}^{\icrit+2} 44\cdot\sweeprad{i} = 44\sum_{i=0}^{\icrit+2} 2^{i+1} < 44\cdot 2^{\icrit+4} = 704\cdot 2^{\icrit} = 704\cdot 2^{\lceil\log_2(D)\rceil - 1} \leq 704\cdot 2^{\log_2(D)} = 704\cdot D$.
\end{proof}

%--------------------------------------------------
%--------------------------------------------------
\section{Arbitrary lines with known initial distance between agents}\label{sec:knowndist}
%--------------------------------------------------
%--------------------------------------------------

%--------------------------------------------------
%--------------------------------------------------
\subsection{Upper bound}\label{knownDupper}
%--------------------------------------------------
%--------------------------------------------------
In this section, we describe an algorithm called $\rvknownDalg$ that solves rendezvous on lines with arbitrary node labelings when two agents start at arbitrary positions and when the delay between the rounds in which they start executing the algorithm is arbitrary. The algorithm works in time $O(D\log^* \ell)$, where $D$ is the initial distance between the agents, and $\ell$ is the larger label of the two starting nodes.
The agents know $D$, but they do not know the delay between the starting rounds. Also, we note that the agents have no global sense of direction, but each agent can locally choose port 0 from its starting node to represent `right' and port 1 from its starting node to represent `left'. Further, using knowledge of the port number of the edge on which it arrived at a node, an agent is able to choose whether its next move will continue in the same direction or if it will switch directions. Without loss of generality, we may assume that all node labels are strictly greater than one, since the algorithm could be re-written to add one to each label value in its own memory before carrying out any computations involving the labels. This assumption ensures that, for any node label $v$, the value of $\log^*(v)$ is strictly greater than 0. 

%--------------------------------------------------
%--------------------------------------------------
\subsubsection{Algorithm description}
%--------------------------------------------------
%--------------------------------------------------
Algorithm $\rvknownDalg$ proceeds in two stages. In the first stage, the agents assign colours to their starting nodes according to a proper colouring of the nodes that are integer multiples of $D$ away. They each accomplish this by determining the node labels within a sufficiently large neighbourhood and then executing the 3-colouring algorithm $\EarlyStopCV$ from \Cref{CValg} on nodes that are multiples of $D$ away from their starting node. 
The second stage consists of repeated periods, with each period consisting of equal-sized blocks of rounds. In each block, an agent either stays idle at its starting node, or, it spends the block performing a search of nearby nodes in an attempt to find the other agent (and if not successful, returns back to its starting node). Whether or not an agent idles or searches in a particular block of the period depends on the colour it picked for its starting node. The overall idea behind the algorithm's correctness is: the agents are guaranteed to pick different colours for their starting node in the first stage, and so, when both agents are executing the second stage, one agent will search while the other idles, which will result in rendezvous. The full details of the algorithm's execution by an agent $x$ is as follows.

\textbf{Stage 1: Colouring.} Let $v_x$ be the starting node of agent $x$. We identify the node with its label.
Let $r = D\cdot\kappa\log^*(v_x)$, where $\kappa > 1$ is the constant defined in the running time of the algorithm $\EarlyStopCV$ from \Cref{CValg}. Denote by $\mathcal{B}_r$ the $r$-neighbourhood of $v_x$. First, agent $x$ determines $\mathcal{B}_r$ (including all node labels) by moving right $r$ times, then left $2r$ times, then right $r$ times, ending back at its starting node $v_x$. Let $V$ be the subset of nodes in $\mathcal{B}_r$ whose distance from $v_x$ is an integer multiple of $D$. In its local memory, agent $x$ creates a path graph $G_x$ consisting of the nodes in $V$, with two nodes connected by an edge if and only if their distance in $\mathcal{B}_r$ is exactly $D$. This forms a path graph centered at $v_x$ with $\kappa\log^*(v_x)$ nodes in each direction. Next, the agent simulates an execution of the algorithm $\EarlyStopCV$ by the nodes of $G_x$ to assign a colour $c_x \in \{0,1,2\}$ to its starting node $v_x$.

\textbf{Stage 2: Search.} The agent repeatedly executes periods consisting of $8D$ rounds each, partitioned into two equal-sized blocks of $4D$ rounds each. In each of the two blocks, the agent behaves in one of two ways: if a block is designated as a \emph{waiting block}, then the agent stays at its starting node for all $4D$ rounds; otherwise, a block is designated as a \emph{searching block}, in which the agent moves right $D$ times, then left $2D$ times, then right $D$ times. Whether a block is designated as a `waiting' or `searching' block depends on the agent's starting node colour that was determined in the first stage. In particular, if $c_x =0$, then both blocks are waiting blocks; if $c_x=1$, then block 1 is a searching block and block 2 is a waiting block; and, if $c_x=2$, then both blocks are searching blocks.

%--------------------------------------------------
%--------------------------------------------------
\subsubsection{Algorithm analysis}
%--------------------------------------------------
%--------------------------------------------------

In this section, we prove that Algorithm $\rvknownDalg$ solves rendezvous within $O(D\log^*{\ell}$) rounds, where $\ell$ is the larger of the labels of the two starting nodes of the agents.

Consider an arbitrary instance on some line $L$ with an arbitrary labeling of the nodes with positive integers. Suppose that two agents $\alpha$ and $\beta$ execute the algorithm. For each $x \in \{\alpha,\beta\}$, we denote by $v_x$ the label of agent $x$'s starting node, and we denote by $c_x$ the colour assigned to node $v_x$ at the end of Stage 1 in $x$'s execution of the algorithm. To help with the wording of the analysis only, fix a global orientation for $L$ so that $v_\alpha$ appears to the `left' of $v_\beta$ (and recall that the agents have no access to this information). 

First, we argue that after both agents have finished Stage 1 of their executions, they have assigned different colours to their starting nodes.
\begin{lemma}\label{stage1diff}
	In any execution of $\rvknownDalg$ by agents $\alpha$ and $\beta$, in every round after both agents finish their execution of Stage 1, we have $c_\alpha \neq c_\beta$.
\end{lemma}
\begin{proof}
	Without loss of generality, assume that $v_\alpha > v_\beta$. Let $y_0$ be the node in $L$ to the left of $v_\alpha$ at distance exactly $D \cdot \kappa\log^*(v_\alpha)$. For each $i \in \{1,\ldots, 2\kappa\log^*(v_\alpha)+1\}$, let $y_i$ be the node in $L$ at distance $i\cdot D$ to the right of $y_0$. Note that $v_\alpha = y_{\kappa\log^*(v_\alpha)}$ and $v_\beta = y_{\kappa\log^*(v_\alpha)+1}$.
	
	Create a path graph $P$ consisting of $2\kappa\log^*(v_\alpha)+2$ nodes. Label the leftmost node in $P$ with $y_0$, and label the node at distance $i$ from $y_0$ in $P$ using the label $y_i$.
	
	By the definition of $P$, note that $y_{\kappa\log^*(v_\alpha)}$ and $y_{\kappa\log^*(v_\alpha)+1}$ are neighbours in $P$, which implies that $v_\alpha$ and $v_\beta$ are neighbours in $P$. This is important because it implies that, if the nodes of $P$ run the algorithm $\EarlyStopCV$ from \Cref{CValg}, the nodes labeled $v_\alpha$ and $v_\beta$ will choose different colours from the set $\{0,1,2\}$.
	
	The rest of the proof shows that, for $x \in \{\alpha,\beta\}$, the graph $G_x$ built in Stage 1 by agent $x$ is an induced subgraph of $P$. This is sufficient since it implies that having each agent $x$ simulate the algorithm $\EarlyStopCV$ on their local $G_x$ results in the same colour assignment to the node labeled $v_x$ as an execution of $\EarlyStopCV$ on the nodes of $P$, so $v_\alpha$ and $v_\beta$ will be assigned different colours at the end of Stage 1 of Algorithm $\rvknownDalg$. Since the colour assignment is not changed in any round after Stage 1, the result follows.

	First, consider $v_\alpha$. Let $w_0$ be the label of the leftmost node of $G_\alpha$. By the definition of $G_\alpha$, the node labeled $w_0$ is at distance exactly $\kappa\log^*(v_\alpha)$ to the left of $v_\alpha$ in $G_\alpha$, so the node labeled $w_0$ is at distance exactly $D\cdot\kappa\log^*(v_\alpha)$ to the left of $v_\alpha$ in $L$. This proves that $w_0 = y_0 \in P$. For each $j \in \{0,\ldots,2\kappa\log^*(v_\alpha)\}$, define $w_j$ to be the label of the node at distance $j$ to the right of the node labeled $w_0$ in $G_\alpha$. By induction on $j$, we prove that $w_j = y_j \in P$ for all $j \in \{0,\ldots,2\kappa\log^*(v_\alpha)\}$. The base case $w_0 = y_0 \in P$ was proved above. As induction hypothesis, assume that $w_{j-1} = y_{j-1} \in P$ for some $j \in \{1,\ldots,2\kappa\log^*(v_\alpha)\}$. Consider $w_j$, which by definition of $G_\alpha$ is the neighbour to the right of $w_{j-1}$ in $G_\alpha$, and, moreover, is located distance exactly $D$ to the right of $w_{j-1}$ in $L$. By the induction hypothesis, we know that $w_{j-1} = y_{j-1}$, so $w_j$ is located distance exactly $D$ to the right of $y_{j-1}$ in $L$, and so $w_j = y_j$ by the definition of $y_j$. To confirm that $y_j \in P$, we note that $j \leq 2\kappa\log^*(v_\alpha) < 2\kappa\log^*(v_\alpha)+1$, and that the rightmost node in $P$ is $y_{2\kappa\log^*(v_\alpha)+1}$. This concludes the inductive step, and the proof that $G_\alpha$ is an induced subgraph of $P$.
	
	Next, consider $v_\beta$. Let $u_0$ be the label of the leftmost node of $G_\beta$. By the definition of $G_\beta$, the node labeled $u_0$ in $G_\beta$ is at distance exactly $\kappa\log^*(v_\beta)$ to the left of $v_\beta$, so, in $L$, the node labeled $u_0$ is at distance exactly $D\cdot\kappa\log^*(v_\beta)$ to the left of $v_\beta$. However, since $v_\alpha$ is located at distance exactly $D$ to the left of $v_\beta$ in $L$, and $y_0$ is located at distance exactly $D\cdot\kappa\log^*(v_\alpha)$ to the left of $v_\alpha$ in $L$, it follows that $y_0$ is located at distance exactly $D\cdot(1+\kappa\log^*(v_\alpha))$ to the left of $v_\beta$ in $L$. Since $v_\alpha > v_\beta$, we conclude that $d(y_0,v_\beta) = D\cdot(1+\kappa\log^*(v_\alpha)) > D\cdot\kappa\log^*(v_\beta) = d(u_0,v_\beta)$ in $L$, so $y_0$ must be to the left of $u_0$ in $L$. Further, it means that $d(u_0,y_0) = D\cdot(1+\kappa\log^*(v_\alpha)) - D\cdot\kappa\log^*(v_\beta) = D\cdot[1+\kappa\log^*(v_\alpha)-\kappa\log^*(v_\beta)]$ in $L$. This proves that $u_0 = y_{1+\kappa\log^*(v_\alpha)-\kappa\log^*(v_\beta)}$. We confirm that $y_{1+\kappa\log^*(v_\alpha)-\kappa\log^*(v_\beta)} \in P$ by noticing that the subscript $1+\kappa\log^*(v_\alpha)-\kappa\log^*(v_\beta)$ lies in the set $\{1,\ldots,2\kappa\log^*(v_\alpha)+1\}$ since $v_\alpha > v_\beta$. Next, define $h = 1+\kappa\log^*(v_\alpha)-\kappa\log^*(v_\beta)$, and, for each $j \in \{0,\ldots,2\kappa\log^*(v_\beta)\}$, define $u_j$ to be the label of the node at distance $j$ to the right of the node labeled $u_0$ in $G_\beta$. By induction on $j$, we prove that $u_j = y_{j+h} \in P$ for all $j \in \{0,\ldots,2\kappa\log^*(v_\beta)\}$. The base case $u_0 = y_{h} \in P$ was proved above. As induction hypothesis, assume that $u_{j-1} = y_{j-1+h} \in P$ for some $j \in \{1,\ldots,2\kappa\log^*(v_\beta)\}$. Consider $u_j$, which by definition of $G_\beta$ is the neighbour to the right of $u_{j-1}$ in $G_\beta$, and, moreover, is located distance exactly $D$ to the right of $u_{j-1}$ in $L$. By the induction hypothesis, we know that $u_{j-1} = y_{j-1+h}$, so $u_j$ is located distance exactly $D$ to the right of $y_{j-1+h}$ in $L$, and so $u_j = y_{j+h}$ by the definition of $y_{h}$. To confirm that $y_{j+h} \in P$, we note that $j \leq 2\kappa\log^*(v_\beta)$ and $h = 1+\kappa\log^*(v_\alpha)-\kappa\log^*(v_\beta)$, so $j+h \leq 1+\kappa\log^*(v_\alpha)+\kappa\log^*(v_\beta) \leq 2\kappa\log^*(v_\alpha)+1$, where the last inequality is due to $v_\alpha > v_\beta$. As the rightmost node of $P$ is $y_{2\kappa\log^*(v_\alpha)+1}$, it follows that $y_{j+h}$ is at or to the left of the rightmost node in $P$, so $y_{j+h} \in P$. This concludes the inductive step, and the proof that $G_\alpha$ is an induced subgraph of $P$.
\end{proof}

To prove that the algorithm correctly solves rendezvous within $O(D\log^*\ell)$ rounds for arbitrary delay between starting rounds, there are two main cases to consider. If the delay is large, then the late agent is idling for the early agent's entire execution of Stage 1, and rendezvous will occur while the early agent is exploring its $(D\kappa\log^*\ell)$-neighbourhood. Otherwise, the delay is relatively small, so both agents reach Stage 2 quickly, and the block structure of the repeated periods ensures that one agent will search while the other waits, so rendezvous will occur. These arguments are formalized in the following result.

\begin{theorem}
	Algorithm $\rvknownDalg$ solves rendezvous in $O(D\log^*\ell)$ rounds on lines with arbitrary node labelings when two agents start at arbitrary positions at known distance $D$, and when the delay between the rounds in which they start executing the algorithm is arbitrary.
\end{theorem}
\begin{proof}
	The agent that starts executing the algorithm first is called the \emph{early agent}, and the other agent is called the \emph{late agent}. If both agents start executing the algorithm in the same round, then arbitrarily call one of them early and the other one late. Without loss of generality, we assume that $\alpha$ is the early agent. The number of rounds that elapse between the two starting rounds is denoted by $\delay$. 
	
	First, we address the case where the delay between the starting rounds of the agents is large. This situation is relatively easy to analyze, since agent $\beta$ remains idle during agent $\alpha$'s entire execution up until they meet.
	\begin{claim}
		If $\delay > 4D\kappa\log^*{v_\alpha}$, then rendezvous is guaranteed to occur within $4D\kappa\log^*{v_\alpha}$ rounds in agent $\alpha$'s execution.
	\end{claim}
	To prove the claim, we note that the total number of rounds that elapse in agent $\alpha$'s execution of Stage 1 is $4D\kappa\log^*{v_\alpha}$. By the assumption about $\delay$, agent $\beta$ is idle at its starting node for $\alpha$'s entire execution of Stage 1. Finally, note that $\alpha$ explores its entire $(D\kappa\log^*{\alpha})$-neighbourhood during its execution of Stage 1, where $D\kappa\log^*{\alpha} \geq D$, so it follows that agent $\alpha$ will be located at agent $\beta$'s starting node at some round in its execution of Stage 1, which completes the proof of the claim.
	
	So, in the remainder of the proof, we assume that $\delay \leq 4D\kappa\log^*{\alpha}$. We prove that rendezvous occurs during Stage 2 by considering the round $\tau$ in which agent $\beta$ starts Stage 2 of its execution. By \Cref{stage1diff}, we have $c_\alpha \neq c_\beta$ in every round from this time onward. We separately consider cases based on the values of $c_\alpha$ and $c_\beta$.
	\begin{itemize}
		\item {\bf Case 1: $\mathbf{c_\alpha = 0 \textbf{ or } c_\beta = 0}$.}
		Let $x$ be the agent with $c_x = 0$, and let $y$ be the other agent. Since $c_x \neq c_y$, we know that $c_y$ is either 1 or 2. From the algorithm's description, agent $y$ will start a searching block within $8D$ rounds after round $\tau$. Moreover, agent $x$ remains idle for the entirety of Stage 2 of its execution. Since $y$ explores its entire $D$-neighbourhood during a searching block, it follows that $y$ will be located at $x$'s starting node within $12D$ rounds after~$\tau$.
		\item {\bf Case 2: $\mathbf{c_\alpha = 1 \textbf{ or } c_\beta = 1}$.}
		Let $x$ be the agent with $c_x = 1$, and let $y$ be the other agent. Since $c_x \neq c_y$, we know that $c_y$ is either 0 or 2. As the case $c_y=0$ is covered by Case 1 above, we proceed with $c_y=2$. From the algorithm's description, agent $x$ has a waiting block that begins within $8D$ rounds after $\tau$. Suppose that this waiting block starts in some round $t$ of $x$'s execution, then we know that $x$ stays idle at its starting node for the $4D$ rounds after $t$. But round $t$ corresponds to the $i$'th round of some searching block in agent $y$'s execution of Stage 2 (since $y$ only performs searching blocks). In the first $4D-i$ rounds, agent $y$ completes its searching block, and then performs the first $i$ rounds of the next searching block. Since the searching blocks are performed using the same movements each time, it follows that $y$ explores its entire $D$-neighbourhood in the $4D$ rounds after $t$, so will meet $x$ at $x$'s starting node within those rounds. Altogether, since $t$ occurs within $8D$ rounds after $\tau$, and $y$'s searching block takes another $4D$ rounds, it follows that rendezvous occurs within $12D$ rounds after $\tau$.
	\end{itemize}
	In all cases, we proved that rendezvous occurs within $12D$ rounds after $\tau$. But $\tau$ is the round in which agent $\beta$ starts Stage 2, so $\tau \leq 4D\kappa\log^*v_\beta$ in $\beta$'s execution. By our assumption on the delay, $\beta$'s execution starts at most $4D\kappa\log^*v_\alpha$ rounds after the beginning of $\alpha$'s execution. Altogether, this means that rendezvous occurs within $4D\kappa\log^*v_\alpha + 4D\kappa\log^*v_\beta + 12D$ rounds from the start of $\alpha$'s execution. Setting $\ell = \max\{v_\alpha,v_\beta\}$, we get that rendezvous occurs within time $8D\kappa\log^*\ell + 12D \in O(D\log^*\ell)$, as desired.
\end{proof}

%--------------------------------------------------
%--------------------------------------------------
\subsection{Lower bound}\label{lower}
%--------------------------------------------------
%--------------------------------------------------
In this section we prove a $\Omega(D\log^* \ell)$ lower bound for rendezvous time on the infinite line, where $\ell$ is the larger label of the two starting nodes, even assuming that agents start simultaneously, they know the initial distance $D$ between them, and they have a global sense of direction.
We start with some terminology about rendezvous executions.
\begin{definition}
	Consider any labeled infinite line $L$, and any two nodes labeled $v,w$ on $L$ that are at fixed distance $D$. Consider any rendezvous algorithm $\mathcal{A}$, and suppose that two agents start executing $\mathcal{A}$ in the same round: one agent $\alpha_v$ starting at node $v$, and the other agent $\alpha_w$ starting at node $w$. Denote by $\gamma(\mathcal{A},L,v,w)$ the resulting execution until $\alpha_v$ and $\alpha_w$ meet. When $\mathcal{A}$ and $L$ are clear from the context, we will simply write $\gamma(v,w)$ to denote the execution.  Denote by $|\gamma(\mathcal{A},L,v,w)|$ the number of rounds that have elapsed before $\alpha_v$ and $\alpha_w$ meet in the execution.  
\end{definition}

The following definition formalizes the notion of ``behaviour sequence'': an integer sequence that encodes the movements made by an agent in each round of an algorithm's execution.
\begin{definition}
	Consider any execution $\gamma(\mathcal{A},L,v,w)$ by an agent $\alpha_v$ starting at node $v$ and an agent $\alpha_w$ starting at node $w$, with both agents starting simultaneously.  Define the behaviour sequence $\mathcal{B}_v(\mathcal{A},L,v,w)$ as follows: for each $t \in \{1,\ldots,|\gamma(\mathcal{A},L,v,w)|\}$, set the $t$'th element to $0$ if $\alpha_v$ moves left in round $t$ of the execution, to $1$ if $\alpha_v$ stays at its current node in round $t$ of the execution, and to $2$ if $\alpha_v$ moves right in round $t$ of the execution. Similarly, define the sequence $\mathcal{B}_w(\mathcal{A},L,v,w)$ using the moves by agent $\alpha_w$. When $\mathcal{A}$ and $L$ are clear from the context, we will simply write $\mathcal{B}_v(v,w)$ and $\mathcal{B}_w(v,w)$ to denote the two behaviour sequences of $\alpha_v$ and $\alpha_w$, respectively.
\end{definition}
As the agents are anonymous and we only consider deterministic algorithms, note that for any fixed $\mathcal{A},L,v,w$, we have $\gamma(v,w)=\gamma(w,v)$ and $\mathcal{B}_v(v,w) = \mathcal{B}_v(w,v)$ and $\mathcal{B}_w(v,w) = \mathcal{B}_w(w,v)$. Moreover, for a fixed starting node $v$ on a fixed line $L$, the behaviour of an agent running an algorithm $\mathcal{A}$ does not depend on the starting node (or behaviour) of the other agent, until the two agents meet. This implies the following result, i.e., if we look at two executions of $\mathcal{A}$ where one agent $\alpha_v$ starts at the same node $v$ in both executions, then $\alpha_v$'s behaviour in both executions is exactly the same up until rendezvous occurs in the shorter execution.
\begin{proposition}\label{prop:prefixes}
	Consider any labeled infinite line $L$, and any fixed rendezvous algorithm $\mathcal{A}$. Consider any fixed node $v$ in $L$, and let $w_1$ and $w_2$ be two nodes other than $v$. Let $p = \min\{|\gamma(v,w_1)|,|\gamma(v,w_2)|\}$. Then $\mathcal{B}_v(v,w_1)$ and $\mathcal{B}_v(v,w_2)$ have equal prefixes of length $p$.
\end{proposition}
The following proposition states that two agents running a rendezvous algorithm starting at two different nodes cannot have the same behaviour sequence. This follows from the fact that the distance between two agents cannot decrease if they perform the same action in each round (i.e., both move left, both move right, or both don't move).
\begin{proposition}\label{prop:unequal}
	Consider any labeled infinite line $L$, and any fixed rendezvous algorithm $\mathcal{A}$. For any two nodes $x$ and $y$ in $L$, in the execution $\gamma(x,y)$ we have $\mathcal{B}_{x}(x,y) \neq \mathcal{B}_{y}(x,y)$.
\end{proposition}

The remainder of this section is dedicated to proving the $\Omega(D\log^*\ell)$ lower bound for rendezvous on the infinite line when the two agents start at a known distance $D$ apart. We proceed in two steps: first, we prove an $\Omega(\log^* \ell)$ lower bound in the case where $D=1$, and then we prove the general $\Omega(D\log^*\ell)$ lower bound using a reduction from the $D=1$ case. Throughout this section, we will refer to the constant $\kappa$ that was defined in \Cref{CValg} in order to state the running time bound of the  algorithm $\EarlyStopCV$.

%--------------------------------------------------
%--------------------------------------------------
\subsubsection{The \texorpdfstring{{\boldmath$D=1$}}{D = 1} case}\label{D1}
%--------------------------------------------------
%--------------------------------------------------
We prove a $\Omega(\log^* \ell)$ lower bound for rendezvous on the infinite line, where $\ell$ is the larger label of the two starting nodes, in the special case where the two agents start at adjacent nodes. This lower bound applies even to algorithms that start simultaneously and know that the initial distance between the two agents is 1.

The overall idea is to assume that there exists a very fast rendezvous algorithm (i.e., an algorithm that always terminates within $\frac{1}{16\kappa}\log^*(\ell)$ rounds) and prove that this implies the existence of a distributed 3-colouring algorithm for the $\mathcal{LOCAL}$ model whose running time is faster than the lower bound proven by Linial (see \Cref{linial}). This contradiction proves that any rendezvous algorithm must have running time $\Omega(\log^*(\ell))$.

The first step is to reduce distributed colouring in the $\mathcal{LOCAL}$ model to rendezvous. The following result describes how to use the rendezvous algorithm to create the distributed colouring algorithm. The idea is to record the agent's behaviour sequence in the execution of the rendezvous algorithm, and convert the sequence to an integer colour. {\color{black} \Cref{prop:unequal} guarantees that the assigned colours are different.}
\begin{lemma}\label{lem:firstcolouring}
	Consider any rendezvous algorithm $\rvalg$ such that $\rvalg$ always terminates within $\frac{1}{16\kappa}\log^*(\ell)$ rounds, where $\ell$ is the larger label of the two starting nodes. Then there exists a distributed colouring algorithm $\colalg$ such that, for any labeled infinite line $L$, and for any finite subline $P$ of $L$ consisting of nodes whose labels are bounded above by some integer $Y$, algorithm $\colalg$ uses $\lfloor \frac{1}{16\kappa}\log^*(Y) \rfloor+1$ rounds of communication and assigns to each node in $P$ an integer colour from the range $1,\ldots,4^{2\lfloor \frac{1}{16\kappa}\log^*(Y) \rfloor+1}$.
\end{lemma}
\begin{proof}
	Let $u = \lfloor \frac{1}{16\kappa}\log^*(Y) \rfloor$. First, we define the distributed colouring algorithm $\colalg$ for an arbitrary node $x$ in $P$. 
	\begin{enumerate}
		\item First, use $u+1$ rounds of communication to obtain the $(u+1)$-neighbourhood of $x$. For notational convenience, we denote by $x_1,\ldots,x_{2u+3}$ the nodes in the $(u+1)$-neighbourhood of $x$ from left to right (such that $x_{u+2} = x$). 
		\item Next, $x$ locally simulates an execution of the rendezvous algorithm $\rvalg$ with the agents starting at nodes $x_{u+1}$ and $x_{u+2}$. This is possible since the rendezvous algorithm is guaranteed to terminate within $u$ rounds, so the agents will always be located at nodes in the range $x_1,\ldots,x_{2u+2}$. From this simulation, $x$ determines the behaviour sequence $\mathcal{B}_{x_{u+2}}(x_{u+1},x_{u+2})$. 
		\item Next, $x$ locally simulates an execution of the rendezvous algorithm $\rvalg$ with the agents starting at nodes $x_{u+2}$ and $x_{u+3}$. This is possible since the rendezvous algorithm is guaranteed to terminate within $u$ rounds, so the agents will always be located at nodes in the range $x_2,\ldots,x_{2u+3}$. From this simulation, $x$ determines the behaviour sequence $\mathcal{B}_{x_{u+2}}(x_{u+2},x_{u+3})$.
		\item Finally, node $x$ concatenates the two behaviour sequences using the digit 3 to separate them, i.e., creates the quaternary string $\mathcal{B}_{x_{u+2}}(x_{u+1},x_{u+2})\cdot 3 \cdot \mathcal{B}_{x_{u+2}}(x_{u+2},x_{u+3})$ and interprets this as a base-4 integer to obtain its integer colour $c$.
	\end{enumerate}
	First, note that the two behaviour sequences $\mathcal{B}_{x_{u+2}}(x_{u+1},x_{u+2})$ and $\mathcal{B}_{x_{u+2}}(x_{u+2},x_{u+3})$ have at most $u$ digits each, as every execution of $\rvalg$ terminates within $u$ rounds. Therefore, the colour created by $x$ consists of at most $2u+1$ digits from the set $\{0,1,2,3\}$, so $x$ assigns itself an integer colour from the range $1,\ldots,4^{2u+1}$. Further, it is clear from the algorithm's description that it does so using only $u+1$ rounds of communication.
	
	It remains to show that the algorithm $\colalg$ produces a proper colouring. The proof proceeds by contradiction: assume that there are two adjacent nodes $v,w$ in $P$ such that $\colalg$ assigns the same colour $c$ to both nodes. We denote by $v_1,\ldots,v_{2u+3}$ the nodes in the $(u+1)$-neighbourhood of $v$ from left to right (such that $v_{u+2} = v$), and we denote by $w_1,\ldots,w_{2u+3}$ the nodes in the $(u+1)$-neighbourhood of $w$ from left to right (such that $w_{u+2} = w$). Without loss of generality, assume that $w$ is $v$'s right neighbour, so we have $v_{u+2} = w_{u+1}$ and $w_{u+2} = v_{u+3}$.
	
	Writing the colour $c$ as a base-4 integer, there is exactly one position with the digit 3, since $\colalg$ created $c$ by concatenating two strings from the set $\{0,1,2\}^*$ using 3 as a separator. Therefore, we can write $c$ uniquely as the string $\mathcal{B}_{\textrm{left}} \cdot 3 \cdot \mathcal{B}_{\textrm{right}}$. The following facts about $\mathcal{B}_{\textrm{left}}$ and $\mathcal{B}_{\textrm{right}}$ will be used in the remainder of the proof.
	\begin{enumerate}[label=(\Roman*)]
		\item From $v$'s execution of $\colalg$, we get that $\mathcal{B}_{\textrm{left}} = \mathcal{B}_{v_{u+2}}(v_{u+1},v_{u+2})$.\label{vBL}
		\item From $v$'s execution of $\colalg$, we get that $\mathcal{B}_{\textrm{right}} = \mathcal{B}_{v_{u+2}}(v_{u+2},v_{u+3})$.\label{vBR}
		\item From the previous paragraph, we know that $v_{u+2} = w_{u+1}$ and $w_{u+2} = v_{u+3}$, so from \ref{vBR}, we get that $\mathcal{B}_{\textrm{right}} = \mathcal{B}_{w_{u+1}}(w_{u+1},w_{u+2})$.\label{rewriteBR}
		\item From $w$'s execution of $\colalg$, we get that $\mathcal{B}_{\textrm{left}} = \mathcal{B}_{w_{u+2}}(w_{u+1},w_{u+2})$.\label{wBL}
	\end{enumerate}
	
	From \ref{vBL} and \ref{vBR}, $\mathcal{B}_{\textrm{left}} = \mathcal{B}_{v_{u+2}}(v_{u+1},v_{u+2})$ and $\mathcal{B}_{\textrm{right}} = \mathcal{B}_{v_{u+2}}(v_{u+2},v_{u+3})$. So, by applying \Cref{prop:prefixes} to $\mathcal{B}_{\textrm{left}}$ and $\mathcal{B}_{\textrm{right}}$, we conclude that $\mathcal{B}_{\textrm{left}}$ and $\mathcal{B}_{\textrm{right}}$ have equal prefixes of length $\min\{|\mathcal{B}_{\textrm{left}}|,|\mathcal{B}_{\textrm{right}}|\}$.
	
	Consider the execution $\gamma(w_{u+1},w_{u+2})$. By the definition of behaviour sequences, we have that $|\gamma(w_{u+1},w_{u+2})| = |\mathcal{B}_{w_{u+1}}(w_{u+1},w_{u+2})| = |\mathcal{B}_{w_{u+2}}(w_{u+1},w_{u+2})|$. From \ref{rewriteBR} and \ref{wBL}, it follows that $|\mathcal{B}_{\textrm{right}}| = |\mathcal{B}_{\textrm{left}}|$. Together with the fact from the previous paragraph that $\mathcal{B}_{\textrm{right}}$ and $\mathcal{B}_{\textrm{left}}$ have equal prefixes of length $\min\{|\mathcal{B}_{\textrm{left}}|,|\mathcal{B}_{\textrm{right}}|\}$, we conclude that $\mathcal{B}_{\textrm{left}} = \mathcal{B}_{\textrm{right}}$.
	
	Since $\mathcal{B}_{\textrm{left}} = \mathcal{B}_{\textrm{right}}$, from \ref{rewriteBR} and \ref{wBL} we get $\mathcal{B}_{w_{u+1}}(w_{u+1},w_{u+2}) = \mathcal{B}_{w_{u+2}}(w_{u+1},w_{u+2})$, which contradicts \Cref{prop:unequal}. So, the assumption that $\colalg$ assigns the same colour to adjacent nodes $v$ and $w$ was incorrect.
\end{proof}
The second step is to take the algorithm $\colalg$ from \Cref{lem:firstcolouring} and turn it into a 3-colouring algorithm $\threecolalg$ using very few additional rounds. We do this by using the algorithm $\EarlyStopCV$ from \Cref{CValg} to quickly reduce the number of colours down to 3. Combined with the previous lemma, we get the following result that shows how to obtain a very fast distributed 3-colouring algorithm under the assumption that we have a very fast rendezvous algorithm.
\begin{lemma}\label{A3col}
	Consider any rendezvous algorithm $\rvalg$ such that $\rvalg$ always terminates within $\frac{1}{16\kappa}\log^*(\ell)$ rounds, where $\ell$ is the larger label of the two starting nodes. Then there exists a distributed 3-colouring algorithm $\threecolalg$ such that, for any labeled infinite line $L$, and for any finite subline $P$ of $L$ consisting of nodes whose labels are bounded above by some integer $Y$, algorithm $\threecolalg$ uses at most $\left( \frac{1}{4} + \frac{1}{16\kappa}\right) \log^*(Y) + 1 + 3\kappa$ rounds of communication to 3-colour the nodes of $P$.
\end{lemma}
\begin{proof}
	Let $u = \lfloor \frac{1}{16\kappa}\log^*(Y) \rfloor$. By \Cref{lem:firstcolouring}, there is a colouring algorithm $\colalg$ using $u+1$ rounds of communication that assigns to each node in $P$ a colour from the range $1,\ldots,4^{2u+1}$. Taking these colours as the new node labels, we then apply \Cref{CValg} (i.e., $\EarlyStopCV$) to obtain a proper 3-colouring of the nodes in $P$, and we are guaranteed that all nodes terminate within time $\kappa\log^*(4^{2u+1})$. By our choice of $u$, we get that
	\begin{align*}
		\kappa\log^*(4^{2u+1}) & \leq \kappa\log^*(4^{2\left(\frac{1}{16\kappa}\log^*(Y)\right)+1})\\
		& = 
		\kappa\log^*(4^{\left(\frac{1}{8\kappa}\log^*(Y)\right)+1})\\
		& = 
		\kappa\log^*(2^{\left(\frac{1}{4\kappa}\log^*(Y)\right)+2})\\
		& = \kappa(1+\log^*(\log(2^{\left(\frac{1}{4\kappa}\log^*(Y)\right)+2}))) {\textrm{\ \ \ (by definition of $\log^*$)}}\\
		& = 
		\kappa(1+\log^*(\frac{1}{4\kappa}\log^*(Y)+2))\\
		& \leq \kappa(1+\frac{1}{4\kappa}\log^*(Y)+2) {\textrm{\ \ \   (since $\log^*(x) \leq x$ for $x \geq 2$)}}\\
		& = 3\kappa + \frac{1}{4}\log^*(Y)
	\end{align*}

	Therefore, we have created a 3-colouring algorithm whose running time is at most $\left(\frac{1}{16\kappa}\log^*(Y) + 1\right) + \left(3\kappa + \frac{1}{4}\log^*(Y)\right)$ rounds, as desired. 
\end{proof}

Finally, we demonstrate how to use the above result to prove the desired $\Omega(\log^*\ell)$ lower bound for rendezvous. The idea is to construct an infinite line that contains an infinite sequence of finite sublines, each of which is a worst-case instance (according to Linial's lower bound), and obtaining the desired contradiction by observing that the upper bound on the running time of $\threecolalg$ violates the lower bound guaranteed by Linial's result.
\begin{lemma}\label{lem:D1LB}
	Any algorithm that solves the rendezvous task on all labeled infinite lines, where the two agents start at adjacent nodes, uses $\Omega(\log^*\ell)$ rounds in the worst case, where $\ell$ is the larger label of the two starting nodes.
\end{lemma}
\begin{proof}
	Assume the existence of an algorithm $\rvalg$ that always terminates within $\frac{1}{16\kappa}\log^*(\ell)$ rounds, where $\ell$ is the larger label of the two starting nodes. We consider the distributed algorithm $\threecolalg$ whose existence is guaranteed by \Cref{A3col}.
	
	Construct an infinite line $L$ of nodes, pick one node which will act as the `origin' $O$ and label it with 1. Label each node at distance $i$ to the left of $O$ with the integer $2i+1$. Partition the nodes to the right of $O$ into contiguous segments, where the $j$'th segment to the right of $O$ consists of a path of $j$ nodes. The $j$ nodes in the $j$'th segment to the right of $O$ are labeled with positive even integers from the range $(j-1)j+2,...,j(j+1)$ in the following way: first, apply Linial's lower bound (see \Cref{linial}) to find a path graph $P$ on $j$ nodes labeled with integers from $\mathcal{I} = \{(j-1)j+2,(j-1)j+4,...,j(j+1)\}$ for which algorithm $\threecolalg$ requires $\frac{1}{2}\log^*{j} - 1$ rounds in order to 3-colour $P$; then, assign the same labeling to the nodes in the $j$'th segment to the right of $O$. This concludes the construction of the infinite line $L$.
	
	Finally, consider any fixed positive integer $\ell \geq 9$. On one hand, by the construction of $L$, note that:
	\begin{itemize}
		\item the $\lfloor \sqrt{\ell}-1 \rfloor$'th segment to the right of $O$ consists of $\lfloor \sqrt{\ell}-1 \rfloor$ nodes and has labels in the range $\lfloor \sqrt{\ell}-1 \rfloor^2 - \lfloor \sqrt{\ell}-1 \rfloor + 2,...,\lfloor \sqrt{\ell}-1 \rfloor^2 + \lfloor \sqrt{\ell}-1 \rfloor$, which is a subrange of $\ell - 5\sqrt{\ell} + 6,\ldots,\ell-\sqrt{\ell}$, and,
		\item the labeling was chosen such that algorithm $\threecolalg$ uses at least $\frac{1}{2}(\log^*{(\sqrt{\ell}}-1)) - 1 \geq \frac{1}{2}(\log^*{(\log{\ell})}) - 1 = \frac{1}{2}(\log^*(\ell) - 1) - 1 = \frac{1}{2}(\log^*(\ell)) - \frac{3}{2}$ rounds to 3-colour the nodes in this segment.
	\end{itemize}
	On the other hand, since the labels in the segment are bounded above by $\ell-\sqrt{\ell}$, \Cref{A3col} guarantees that the algorithm $\threecolalg$ 3-colours the segment using at most $\left( \frac{1}{4} + \frac{1}{16\kappa}\right) \log^*(\ell-\sqrt{\ell}) + 1 + 3\kappa \leq \left( \frac{1}{4} + \frac{1}{16\kappa}\right) \log^*(\ell) + 1 + 3\kappa$. 
	
	So, we have shown that algorithm $\threecolalg$ uses at most $\left( \frac{1}{4} + \frac{1}{16\kappa}\right) \log^*(\ell) + 1 + 3\kappa$ rounds and at least $\frac{1}{2}(\log^*(\ell)) - \frac{3}{2}$ rounds to 3-colour the nodes of the segment, a contradiction for sufficiently large $\ell$. Therefore, the assumed algorithm $\rvalg$ cannot exist.
\end{proof}

%--------------------------------------------------
%--------------------------------------------------
\subsubsection{The \texorpdfstring{{\boldmath$D > 1$}}{D > 1} case}
%--------------------------------------------------
%--------------------------------------------------
We prove a $\Omega(D\log^*\ell)$ lower bound for rendezvous on the infinite line, where $\ell$ is the larger label of the two starting nodes, in the case where the two agents start at nodes that are distance $D > 1$ apart. This lower bound applies even to algorithms that start simultaneously and know that the initial distance between the two agents is $D$. Hence it shows that the running time of Algorithm $\rvknownDalg$ has optimal order of magnitude among rendezvous algorithms knowing the initial distance between the agents.

The overall idea is to assume that there exists a very fast rendezvous algorithm called $\rvalg$ (that always terminates within $\frac{1}{224\kappa}D\log^*\ell$ rounds) and prove that this implies the existence of a rendezvous algorithm $\rvadjalg$ for the $D=1$ case that always terminates within $\frac{1}{16\kappa}\log^*\ell$ rounds, which we already proved is impossible in \Cref{D1}. This contradiction proves that any rendezvous algorithm for the $D>1$ case must have running time $\Omega(D\log^*\ell)$.

\begin{theorem}\label{th-lower}
	Any algorithm that solves the rendezvous task on all labeled infinite lines, where the two agents start at known distance $D > 1$ apart, uses $\Omega(D\log^*\ell)$ rounds in the worst case, where $\ell$ is the larger label of the two starting nodes.
\end{theorem}
\begin{proof}
	Assume the existence of an algorithm $\rvalg$ that always terminates within $\frac{1}{224\kappa}D\log^*(\ell)$ rounds, where $\ell$ is the larger label of the two starting nodes.

	We will use $\rvalg$ to design an algorithm $\rvadjalg$ that solves the rendezvous task on any infinite line when the agents start at adjacent nodes. The idea is to take any instance where the agents start at adjacent nodes, and ``blow it up'' by a factor of $D$: all node labels are multiplied by a factor of $D$, and $D-1$ ``dummy nodes'' are inserted between each pair of nodes. Each node from the original instance, together with the $D-1$ nodes to its right, form a `segment' in the blown-up instance. Then, the algorithm $\rvalg$ is simulated on the blown-up instance in stages of length $D$, and each simulated stage corresponds to 1 round of algorithm $\rvadjalg$ in the original instance. At the end of every simulated stage, the simulated agent is in some segment that has leftmost node $v$, so the real agent situates itself at node $v$ for the corresponding round in the original instance. Roughly speaking, since $\rvalg$ guarantees rendezvous in the blown-up instance, the two simulated agents will end up in the same segment (with the same leftmost node $v$), so the two real agents will end up at node $v$ in the original instance. Further, since each stage of $D$ simulated rounds corresponds to one round in the original instance, the number of rounds used by $\rvadjalg$ is a factor of $D$ less than the running time of the simulated algorithm $\rvalg$. This gives us a contradiction, as the resulting running time for $\rvadjalg$ is smaller than the lower bound proven in \Cref{lem:D1LB}.
	
	However, the above summary overlooks some complications. 
	\begin{enumerate}
		\item Assigning labels to the nodes in the blown-up instance: the agent in the original instance needs to assign labels to the blown-up instance in a way that is consistent with the original. However, the agent initially only knows the label at its own node, so, it must learn about neighbouring labels before it can assign labels in the blown-up instance.
		\item Guaranteeing rendezvous: two agents running $\rvadjalg$ are each independently simulating $\rvalg$ in their own memory, so they cannot detect when the simulated agents meet. So, although $\rvalg$ guarantees rendezvous, it might be the case that there is never a stage of $D$ simulated rounds after which the two simulated agents end up in the same segment, since the simulated agents might continue moving after the undetected rendezvous and end up in different segments. However, by choosing the segment and stage lengths appropriately, we can guarantee that, after the undetected rendezvous occurs in the simulation, the two simulated agents are either in the same segment or in neighbouring segments.  So, after the simulation is done, the real agents do a `dance' in the original instance in such a way that rendezvous occurs even if the simulated agents didn't end up in the same segment.
	\end{enumerate} 
	
	We now present the fully-detailed version of $\rvadjalg$. Consider an agent $\alpha$ initially located at some node labeled with the integer $i$ on an infinite line $L_1$. The agent maintains in its memory a parallel instance: an infinite line $L_D$ with an agent $\alpha'$ initially located at a node labeled $D\cdot i$. 
	
	The algorithm proceeds in phases, each consisting of seven rounds. Consider an arbitrary phase, denote by $v_x$ the integer label of the node where $\alpha$ is located at the start of the phase, and denote by $v_{x-1}$ and $v_{x+1}$ the integer labels of the nodes to the left and right of $v_x$, respectively. The agent performs the following during the phase:
	\begin{enumerate}[label=(\Roman*)]
		\item In the first four rounds of the phase, the agent $\alpha$ moves left to $v_{x-1}$, then right to $v_x$, then right to $v_{x+1}$, then left to $v_x$.
		\item Using its knowledge of $v_{x-1},v_x,v_{x+1}$, agent $\alpha$ updates any previously unlabeled nodes in its parallel instance $L_D$ as follows:
		\begin{itemize}
			\item the node at distance $D$ to the left of $D\cdot v_x$ is labeled with the integer $D\cdot v_{x-1}$
			\item the node at distance $D$ to the right of $D\cdot v_x$ is labeled with the integer $D\cdot v_{x+1}$
			\item the $D-1$ nodes between $D\cdot v_{x-1}$ and $D\cdot v_{x}$ are labeled, from left to right, using the consecutive integers $D\cdot v_{x-1}+1,\ldots,D\cdot v_{x-1}+D-1$
			\item the $D-1$ nodes between $D\cdot v_{x}$ and $D\cdot v_{x+1}$ are labeled, from left to right, using the consecutive integers $D\cdot v_{x}+1,\ldots,D\cdot v_{x}+D-1$
			\item the $D-1$ nodes to the right of $D\cdot v_{x+1}$ are labeled, from left to right, using the consecutive integers $D\cdot v_{x+1}+1,\ldots,D\cdot v_{x+1}+D-1$
		\end{itemize}
		For the sake of presenting the rest of the algorithm details, we define \emph{segments} as follows: for each $v_y \in \{v_{x-1},v_x,v_{x+1}\}$, we define the segment $S_y$ to be the finite subline of $L_D$ consisting of the $D$ nodes labeled $D\cdot v_y, D\cdot v_{y}+1,\ldots,D\cdot v_{y}+D-1$.
		\item Agent $\alpha$ simulates, in $L_D$, the execution of $\rvalg$ by the agent $\alpha'$ for $D$ (simulated) rounds, starting from where $\alpha'$ was located at the end of the previous phase's simulation (or, if this is the first phase, starting from the node labeled $D\cdot i$).
		\item The behaviour of $\alpha$ in $L_1$ in the final three rounds of the phase depends on the behaviour of $\alpha'$ in $L_D$ during the $D$ simulated rounds. There are several cases:
		\begin{enumerate}[label=(\roman*)]
			\item If, in $L_D$, agent $\alpha'$ was located in segment $S_{x}$ for all $D$ simulated rounds, then $\alpha$ stays at node $v_x$ in $L_1$ for the final three rounds of the phase.
			\item If, in $L_D$, agent $\alpha'$ was located in segments $S_{x}$ and $S_{x-1}$ during the $D$ simulated rounds, and ends the simulation in segment $S_{x-1}$, then $\alpha$ does the following in the final three rounds of the phase: stays at $v_x$ for one round, moves left to $v_{x-1}$, then stays at $v_{x-1}$ for the final round.
			\item If, in $L_D$, agent $\alpha'$ was located in segments $S_{x}$ and $S_{x-1}$ during the $D$ simulated rounds, and ends the simulation in segment $S_{x}$, then $\alpha$ does the following in the final three rounds of the phase: stays at $v_x$ for one round, moves left to $v_{x-1}$, then moves right to $v_{x}$.
			\item If, in $L_D$, agent $\alpha'$ was located in segments $S_{x}$ and $S_{x+1}$ during the $D$ simulated rounds, and ends the simulation in segment $S_{x+1}$, then $\alpha$ does the following in the final three rounds of the phase: moves right to $v_{x+1}$, then stays at $v_{x+1}$ for the final two rounds.
			\item If, in $L_D$, agent $\alpha'$ was located in segments $S_{x}$ and $S_{x+1}$ during the $D$ simulated rounds, and ends the simulation in segment $S_{x}$, then $\alpha$ does the following in the final three rounds of the phase: moves right to $v_{x+1}$, stays at $v_{x+1}$ for one round, then moves left to $v_{x}$.
		\end{enumerate}
	\end{enumerate}
	
	In what follows, we will write `simulation $k$' to refer to the $D$ simulation rounds that occur during step (III) in the $k$'th phase of algorithm $\rvadjalg$. In the remainder of the proof, we consider two agents $\alpha$ and $\beta$ both executing $\rvadjalg$ starting in the same round.
	
	To prove the correctness of $\rvadjalg$, we will often need to argue about the locations of the simulated agents within $L_D$. The fact that the number of simulation rounds per phase and the number of nodes in each segment are both $D$ allows us to narrow down in which segment(s) a simulated agent might be located during the simulation rounds. To this effect, the following three claims provide the facts that we will need throughout the proof.
	
	\begin{claim}\label{OnlyNbrSegs}
		During simulation $k$, if there is a round where a particular simulated agent is located in some segment $S_x$, then the same simulated agent can never be located in $S_{x-2}$ or $S_{x+2}$ during simulation $k$.
	\end{claim}
	To prove the claim, suppose that an agent is located at some node $v$ in $S_x$ in some round of simulation $k$. Note that simulation $k$ has exactly $D$ simulated rounds, which is strictly less than the distance between $v$ and any node in $S_{x-2}$ or $S_{x+2}$. Together with the fact that an agent can traverse at most one edge per round, this completes the proof of the claim.
	
	\begin{claim}\label{nonex+1}
		Suppose that, during simulation $k$, there is a simulated round $\tau$ in which both simulated agents are located at the same node $w$ in some segment $S_x$. If a simulated agent is in $S_{x-1}$ in any round before $\tau$ during simulation $k$, then no simulated agent is located in $S_{x+1}$ in any round after $\tau$ during simulation $k$.
	\end{claim}
	To prove the claim, let $d_1$ denote the distance between $w$ and the rightmost node of segment $S_{x-1}$, and let $d_2$ denote the distance between $w$ and the leftmost node of segment $S_{x+1}$. Since $w$ is in $S_x$, and $S_x$ consists of $D$ nodes, it follows that $d_1+d_2 > D$. Suppose that a simulated agent is in $S_{x-1}$ in some round before $\tau$. By the definition of $d_1$, it takes at least $d_1$ rounds for this agent to reach $w$, so there are at least $d_1$ rounds that elapse before $\tau$ in simulation $k$. Similarly, by the definition of $d_2$, it takes at least $d_2$ rounds for any agent to travel from $w$ to the leftmost node of $S_{x+1}$. It follows that, from the start of simulation $k$, the earliest round in which an agent can be located in $S_{x+1}$ is $d_1+d_2$. However, simulation $k$ consists of $D < d_1+d_2$ rounds, which completes the proof of the claim.
	
	By symmetry (swapping the roles of $S_{x-1}$ and $S_{x+1}$), the proof of \Cref{nonex+1} can be reused to prove the following claim.
	\begin{claim}\label{nonex-1}
		Suppose that, during simulation $k$, there is a simulated round $\tau$ in which both simulated agents are located at the same node in some segment $S_x$. If a simulated agent is in $S_{x+1}$ before $\tau$ during simulation $k$, then no simulated agent is located in $S_{x-1}$ after $\tau$ during simulation $k$.
	\end{claim}

	The following claim establishes the key relationship between the real instance and the simulated instance, i.e., it relates the location of a real agent in $L_1$ to the location of a simulated agent in $L_D$. 
	
	\begin{claim}\label{connectInstances}
		Agent $\alpha$ starts a phase at a node $v_x$ in $L_1$ if and only if simulated agent $\alpha'$ starts simulation $k$ at a node in segment $S_x$ in $L_D$.
	\end{claim}
	The proof proceeds by induction on $k$. For the base case, from the algorithm's description we see that $v_x = i$ and that the simulated agent starts at the node labeled $D\cdot i = D\cdot v_x$ in $L_D$, which is the leftmost node in segment $S_x$. As induction hypothesis, suppose that the claim holds for some $k \geq 1$. In particular, agent $\alpha$ starts phase $k$ at some node $v_x$ in $L_1$, and we assume that the simulated agent $\alpha'$ will start simulation $k$ at some node in segment $S_x$ in $L_D$. By \Cref{OnlyNbrSegs}, the simulated agent $\alpha'$ can only be located in two of $\{S_{x-1},S_x,S_{x+1}\}$ during the simulated rounds. Therefore, in the algorithm's description, the five sub-cases of (IV) are exhaustive. Finally, by inspecting these five cases, we can see that, for some fixed $y \in \{x-1,x,x+1\}$, agent $\alpha$ is located at node $v_y$ in the final round of the phase, and simulated agent $\alpha'$ is located at a node in segment $S_y$ at the end of the simulation. It follows that $\alpha$ starts phase $k+1$ at node $v_y$ in $L_1$, and that $\alpha'$ will start simulation $k+1$ at some node in $S_y$ in $L_D$, which concludes the proof of the claim.
	
	We can now argue that $\rvadjalg$ solves rendezvous in $L_1$. Notice that two agents $\alpha,\beta$ starting at distance 1 apart in $L_1$ will simulate agents $\alpha',\beta'$ located distance $D$ apart in $L_D$, and since $\rvalg$ is a correct algorithm, there must be a phase in the execution of $\rvadjalg$ where the simulated agents $\alpha'$ and $\beta'$ are located at the same node $w$ in some segment $S_y$ of $L_D$. Consider the earliest such phase $k$, suppose that the real agents $\alpha$ and $\beta$ have not met before phase $k$, and consider the earliest simulated round $\tau$ in simulation $k$ such that $\alpha'$ and $\beta'$ are both located at $w$. Note that the two simulated agents $\alpha'$ and $\beta'$ did not start simulation $k$ in the same segment, as this would imply (via \Cref{connectInstances}) that the real agents were located at the same node in $L_1$ at the end of the previous phase. Also, by \Cref{OnlyNbrSegs}, the simulated agents can never be located in segments $S_{y-2}$ or $S_{y+2}$ during simulation $k$. Therefore, without loss of generality, we may assume that simulated agent $\alpha'$ starts simulation $k$ in segment $S_{y-1}$ or in segment $S_y$, and that simulated agent $\beta'$ starts simulation $k$ in segment $S_y$ (only if $\alpha'$ starts in $S_{y-1}$) or in segment $S_{y+1}$. We separately consider all the possible cases below.

	\begin{itemize}
		\item Case 1: Agent $\alpha'$ starts simulation $k$ in segment $S_{y-1}$, agent $\beta'$ starts simulation $k$ in segment $S_{y+1}$
		
		\Cref{OnlyNbrSegs,nonex+1,nonex-1} imply that both $\alpha'$ and $\beta'$ are located in segment $S_y$ at the end of simulation $k$. Since agent $\alpha'$ starts simulation $k$ in $S_{y-1}$ and ends the simulation in $S_y$, we know from the algorithm's description that $\alpha$ behaves as described in (IV)(iv) (with $x=y-1$) for the final three rounds of phase $k$. It follows that $\alpha$ is located at node $v_y$ in $L_1$ in the last round of phase $k$. Also, since agent $\beta'$ starts simulation $k$ in $S_{y+1}$ and ends the simulation in $S_y$, we know from the algorithm's description that $\beta$ behaves as described in (IV)(ii) (with $x=y+1$) for the final three rounds of phase $k$. It follows that $\beta$ is located at node $v_y$ in $L_1$ in the last round of phase $k$. Thus, $\alpha$ and $\beta$ will meet at the same node in the final round of phase $k$, if not earlier.
		
		\item Case 2: Agent $\alpha'$ starts simulation $k$ in segment $S_{y-1}$, agent $\beta'$ starts simulation $k$ in segment $S_y$
		
		First, since $\alpha'$ is at $w$ in $S_y$ in simulated round $\tau$, it follows from \Cref{OnlyNbrSegs} that $\alpha'$ can only be located in $S_{y-1} \cup S_y \cup S_{y+1}$ during simulation $k$. Further, by \Cref{nonex+1}, it follows that $\alpha'$ can never be in $S_{y+1}$. Together with the fact that $\alpha'$ starts in $S_{y-1}$, we know that $\alpha'$ is in $S_y$ and $S_{y-1}$ during simulation $k$. From the description of $\rvadjalg$ (with $x=y-1$), in the final three rounds of phase $k$, agent $\alpha$ behaves according to (IV)(iv) or (IV)(v).
		
		With regards to the behaviour of $\beta$ in the final three rounds of phase $k$, there are several cases to consider based on the behaviour of the simulated agent $\beta'$ during simulation $k$. By \Cref{OnlyNbrSegs}, the following cases are exhaustive.
		\begin{itemize}
			\item Case 2(a): $\beta'$ is located in $S_y$ and $S_{y+1}$ during simulation $k$
			
			As $\alpha'$ starts simulation $k$ in $S_{y-1}$, \Cref{nonex+1} implies that $\beta'$ cannot be located in $S_{y+1}$ after $\tau$, so $\beta'$ cannot be located in $S_{y+1}$ at the end of simulation $k$. From the description of $\rvadjalg$ (with $x=y$), in the final three rounds of phase $k$, agent $\beta$ behaves according to (IV)(v), so is located at $v_y$ in the final round of phase $k$. Moreover, the above argument also implies that $\beta'$ must have been located in $S_{y+1}$ before $\tau$, so, by \Cref{nonex-1}, simulated agent $\alpha'$ cannot be located in $S_{y-1}$ after $\tau$, so $\alpha'$ cannot be located in $S_{y-1}$ at the end of simulation $k$. From the description of $\rvadjalg$ (with $x=y-1$), in the final three rounds of phase $k$, agent $\alpha$ behaves according to (IV)(iv), so is located at $v_y$ in the final round of phase $k$. Thus, $\alpha$ and $\beta$ will meet at the same node in the final round of phase $k$, if not earlier.
			
			\item Case 2(b): $\beta'$ is only located in $S_y$ during simulation $k$
			
			From the description of $\rvadjalg$ (with $x=y$), in the final three rounds of phase $k$, agent $\beta$ behaves according to (IV)(i), so is located at $v_y$ in the final three rounds of phase $k$. We proved above that $\alpha$ behaves according to (IV)(iv) or (IV)(v) (with $x = y-1$), and in either case, will be located at $v_y$ in the second-to-last round of phase $k$. Thus, $\alpha$ and $\beta$ will meet at the same node in the sixth round of phase $k$, if not earlier.
			
			\item Case 2(c): $\beta'$ is located in $S_y$ and $S_{y-1}$ during simulation $k$
			
			From the description of $\rvadjalg$ (with $x=y$), in the final three rounds of phase $k$, agent $\beta$ behaves according to (IV)(ii) or (IV)(iii) (with $x=y$), so, in either case, will be located at $v_{y}$ in the fifth round of phase $k$. We already proved above that $\alpha$ behaves according to (IV)(iv) or (IV)(v) (with $x = y-1$), and in either case, will be located at $v_{y}$ in the fifth round of phase $k$. Thus, $\alpha$ and $\beta$ will meet at the same node in the fifth round of phase $k$, if not earlier.
		\end{itemize}
		
		\item Case 3: Agent $\alpha'$ starts simulation $k$ in segment $S_{y}$, agent $\beta'$ starts simulation $k$ in segment $S_{y+1}$
		
		First, since $\beta'$ is at $w$ in $S_y$ in simulated round $\tau$, it follows from \Cref{OnlyNbrSegs} that $\beta'$ can only be located in $S_{y-1} \cup S_y \cup S_{y+1}$ during simulation $k$. Further, by \Cref{nonex-1}, it follows that $\beta'$ can never be in $S_{y-1}$. Together with the fact that $\beta'$ starts in $S_{y+1}$, we know that $\beta'$ is in $S_y$ and $S_{y+1}$ during simulation $k$. From the description of $\rvadjalg$ (with $x=y+1$), in the final three rounds of phase $k$, agent $\beta$ behaves according to (IV)(ii) or (IV)(iii).
		
		With regards to the behaviour of $\alpha$ in the final three rounds of phase $k$, there are several cases to consider based on the behaviour of the simulated agent $\alpha'$ during simulation $k$. By \Cref{OnlyNbrSegs}, the following cases are exhaustive.
		\begin{itemize}
			\item Case 3(a): $\alpha'$ is located in $S_y$ and $S_{y-1}$ during simulation $k$
			
			As $\beta'$ starts simulation $k$ in $S_{y+1}$, \Cref{nonex-1} implies that $\alpha'$ cannot be located in $S_{y-1}$ after $\tau$, so $\alpha'$ cannot be located in $S_{y-1}$ at the end of simulation $k$. From the description of $\rvadjalg$ (with $x=y$), in the final three rounds of phase $k$, agent $\alpha$ behaves according to (IV)(iii), so is located at $v_y$ in the final round of phase $k$. Moreover, the above argument also implies that $\alpha'$ must have been located in $S_{y-1}$ before $\tau$, so, by \Cref{nonex+1}, simulated agent $\beta'$ cannot be located in $S_{y+1}$ after $\tau$, so $\beta'$ cannot be located in $S_{y+1}$ at the end of simulation $k$. From the description of $\rvadjalg$ (with $x=y+1$), in the final three rounds of phase $k$, agent $\beta$ behaves according to (IV)(ii), so is located at $v_y$ in the final round of phase $k$. Thus, $\alpha$ and $\beta$ will meet at the same node in the final round of phase $k$, if not earlier.
			
			\item Case 3(b): $\alpha'$ is only located in $S_y$ during simulation $k$
			
			From the description of $\rvadjalg$ (with $x=y$), in the final three rounds of phase $k$, agent $\alpha$ behaves according to (IV)(i), so is located at $v_y$ in the final three rounds of phase $k$. We proved above that $\beta$ behaves according to (IV)(ii) or (IV)(iii) (with $x = y+1$), and in either case, will be located at $v_y$ in the second-to-last round of phase $k$. Thus, $\alpha$ and $\beta$ will meet at the same node in the sixth round of phase $k$, if not earlier.
			
			\item Case 3(c): $\alpha'$ is located in $S_y$ and $S_{y+1}$ during simulation $k$
			
			From the description of $\rvadjalg$ (with $x=y$), in the final three rounds of phase $k$, agent $\alpha$ behaves according to (IV)(iv) or (IV)(v) (with $x=y$), so, in either case, will be located at $v_{y+1}$ in the fifth round of phase $k$. We proved above that $\beta$ behaves according to (IV)(ii) or (IV)(iii) (with $x = y+1$), and in either case, will be located at $v_{y+1}$ in the fifth round of phase $k$. Thus, $\alpha$ and $\beta$ will meet at the same node in the fifth round of phase $k$, if not earlier.
		\end{itemize}

	\end{itemize}
	In all cases, we proved that rendezvous will occur by the end of phase $k$, which concludes the proof of correctness of $\rvadjalg$.
	
	Finally, we consider the running time of $\rvadjalg$. By our assumption on the running time of $\rvalg$, the number of simulated rounds needed before $\alpha'$ and $\beta'$ are located at the same node in $L_D$ is at most $\frac{1}{224\kappa}D\log^*(D\ell)$, since $D\ell$ is the larger of the two labels at the starting nodes of $\alpha'$ and $\beta'$ in $L_D$. We proved above that $\alpha$ and $\beta$ meet in $L_1$ by the end of the first phase in which $\alpha'$ and $\beta'$ are located at the same node in $L_D$ (if not earlier). Since each phase contains $D$ simulation rounds, it follows that every execution of $\rvadjalg$ lasts for at most $\lceil \frac{1}{224\kappa}\log^*(D\ell)\rceil \leq \frac{1}{224\kappa}\log^*(D\ell) + 1$ phases. Since each phase of $\rvadjalg$ consists of 7 rounds, it follows that every execution of $\rvadjalg$ lasts for at most $\frac{1}{32\kappa}\log^*(D\ell) + 7$ rounds. Finally, for sufficiently large $\ell$, 
	
	\begin{align*}
		\frac{1}{32\kappa}\log^*(D\ell) + 7 & \leq \frac{1}{32\kappa}\log^*(\ell^2) + 7\\
		& = \frac{1}{32\kappa}(1+\log^*(\log(\ell^2))) + 7\\
		& = \frac{1}{32\kappa}\log^*(2\log(\ell)) + \frac{1}{32\kappa} + 7\\
		& \leq \frac{1}{32\kappa}\log^*(\ell) + \frac{1}{32\kappa} + 7\\
		& \leq \frac{1}{16\kappa}\log^*(\ell)
	\end{align*}
	We conclude that $\rvadjalg$ is a rendezvous algorithm for the case $D=1$ that uses at most $\frac{1}{16\kappa}\log^*(\ell)$ rounds. However, this contradicts the proof of \Cref{lem:D1LB}, which demonstrates that no such algorithm can exist. Therefore, the assumed algorithm $\rvalg$ cannot exist.
\end{proof}

%--------------------------------------------------
%--------------------------------------------------
\section{Arbitrary lines with unknown initial distance between agents}\label{sec:unknowndist}
%--------------------------------------------------
%--------------------------------------------------
In this section, we describe an algorithm called $\rvunknownDalg$ that solves rendezvous on lines with arbitrary node labelings in time $O(D^2(\log^*\ell)^3)$ (where $D$ is the initial distance between the agents and $\ell$ is the larger label of the two starting nodes)  when two agents start at arbitrary positions and when the delay between the rounds in which they start executing the algorithm is arbitrary. The agents do not know the initial distance $D$ between them, and they do not know the delay between the starting rounds. Also, we note that the agents have no global sense of direction, but each agent can locally choose port 0 from its starting node to represent `right' and port 1 from its starting node to represent `left'. Further, using knowledge of the port number of the edge on which it arrived at a node, an agent is able to choose whether its next move will continue in the same direction or if it will switch directions. Without loss of generality, we may assume that all node labels are strictly greater than one, since the algorithm could be re-written to add one to each label value in its own memory before carrying out any computations involving the labels. This assumption ensures that, for any node label $v$, the value of $\log^*(v)$ is strictly greater than 0.

%--------------------------------------------------
%--------------------------------------------------
\subsection{Algorithm description}
%--------------------------------------------------
%--------------------------------------------------
As seen in \Cref{knownDupper}, if the initial distance $D$ between the agent is known, then we have an algorithm $\rvknownDalg$ that will solve rendezvous in $O(D\log^*\ell)$ rounds. We wish to extend that algorithm for the case of unknown distance by repeatedly running $\rvknownDalg$ with guessed values for $D$. To get an optimal algorithm, i.e., with running time $O(D\log^*\ell)$, a natural attempt would be to proceed by doubling the guess until it exceeds $D$, so that the searching range of an agent includes the starting node of the other agent. However, this approach will not work in our case, because our algorithm $\rvknownDalg$ requires the exact value of $D$ to guarantee rendezvous. More specifically, the colouring stage using guess $g$ only guarantees that nodes at distance exactly $g$ are assigned different colours. So, instead, our algorithm $\rvunknownDalg$ increments the guessed value by 1 so that the guess is guaranteed to eventually be equal to $D$, which results in a running time quadratic in $D$ instead of linear in $D$.

At a high level, our algorithm $\rvunknownDalg$ consists of phases, where each phase is an attempt to solve rendezvous using a value $g$ which is a guess for the value $D$. The first phase sets $g=1$. Each phase has three stages. In the first stage, the agent waits at its starting node for a fixed number of rounds. In the second and third stages, the agent executes a modified version of the algorithm $\rvknownDalg$ from \Cref{knownDupper} using the value $g$ instead of $D$. At the end of each phase, if rendezvous has not yet occurred, the agent increments its guess $g$ and proceeds to the next phase. {\color{black} A major complication is that an adversary can choose the wake-up times of the agents so that the phases do not align well, e.g., the agents are using the same guess $g$ but are at different parts of the phase, or, they are in different phases and not using the same guess $g$. This means we have to very carefully design the phases and algorithm analysis to account for arbitrary delays between the wake-up times.}

The detailed description of the algorithm is as follows.
Consider an agent $x$ whose execution starts at a node labeled $v_x$. We now describe an arbitrary phase in the algorithm's execution. Let $g \geq 1$, let $d = 1+\floor*{\log_2 g}$, and recall that $\kappa > 1$ is an integer constant defined in the running time of the algorithm $\EarlyStopCV$ from \Cref{CValg}. The $g$-th phase executed by agent $x$, denoted by $\Pg{g}{x}$, consists of executing the following three stages.

\textbf{Stage 0: Wait.} Stay at the node $v_x$ for $36\cdot 2^d \cdot \klog{v_x}$ rounds.

\textbf{Stage 1: Colouring.} Let $r = g\cdot\kappa\log^*(v_x)$. Denote by $\mathcal{B}_r$ the $r$-neighbourhood of $v_x$, and let $V_g$ be the subset of nodes in $\mathcal{B}_r$ whose distance from $v_x$ is an integer multiple of $g$. First, agent $x$ determines $\mathcal{B}_r$ (including all node labels) by moving right $r$ times, then left $2r$ times, then right $r$ times, ending back at its starting node $v_x$.  Then, in its local memory, agent $x$ creates a path graph $G_x$ consisting of the nodes in $V_g$, with two nodes connected by an edge if and only if their distance in $\mathcal{B}_r$ is exactly $g$. This forms a path graph centered at $v_x$ with $\kappa\log^*(v_x)$ nodes in each direction. The agent simulates an execution of the algorithm $\EarlyStopCV$ from \Cref{CValg} by the nodes of $G_x$ to obtain a colour $c_x \in \{0,1,2\}$. Let $CV_x$ be the 2-bit binary representation of $c_x$. Transform $CV_x$ into an 8-bit binary string $CV_x'$ by replacing each 0 in $CV_x$ with 0011, and replacing each 1 in $CV_x$ with 1100. Finally, create a 9-bit string $S$ by appending a 1 to $CV_x'$.

\textbf{Stage 2: Search.} This stage consists of performing $\klog{v_x}$ periods of $|S| = 9$ blocks each. Block $i$ of a period is designated as a \emph{waiting} block if the $i$'th bit of $S$ is 0, else it is designated as a \emph{searching} block. A waiting block consists of $4\cdot (2^d)$ consecutive rounds during which the agent stays at its starting node. A searching block consists of $4\cdot (2^d)$ consecutive rounds: the agent first moves right $2^d$ times, then left $2\cdot 2^d$ times, then right $2^d$ times.
%--------------------------------------------------
%--------------------------------------------------
\subsection{Algorithm analysis}
%--------------------------------------------------
%--------------------------------------------------
In this section, we prove that Algorithm $\rvunknownDalg$ solves rendezvous within $O(D^2(\log^*{\ell})^3)$ rounds, where $\ell$ is the larger of the labels of the two starting nodes of the agents.

Consider an arbitrary instance on some line $L$ with an arbitrary labeling of the nodes with positive integers. Consider any execution of $\rvunknownDalg$ by agents $\alpha$ and $\beta$ starting at nodes labeled $v_\alpha$ and $v_\beta$, respectively. Denote by $\delay$ the number of rounds that elapse between the starting times of the two agents. Without loss of generality, if $\delay > 0$, then assume that $\alpha$'s execution starts first.

At a very high level, there are two main scenarios in which our proposed algorithm must correctly solve rendezvous. In the first scenario, the executions of the algorithm by $\alpha$ and $\beta$ are `highly synchronized', i.e., their executions start at the same time (or roughly at the same time), and, the length of any particular phase $\Pg{g}{x}$ is the same in both executions (i.e., because $\log^*{v_\alpha} = \log^*{v_\beta}$). In this scenario, the main challenge is to break symmetry. Once the phase with $g = D$ is reached, stages 1 and 2 of the algorithm ensure that rendezvous will occur, for roughly the same reason that our algorithm $\rvknownDalg$ works: stage 1 ensures that the two starting nodes of the agents are assigned different colours, and stage 2 ensures that there will be a time interval when one agent is performing a waiting block at its starting node while the other agent is performing a searching block (during which rendezvous will occur). In the second scenario, the executions of the algorithm $\alpha$ and $\beta$ become `desynchronized' at some point, i.e., either there is a significant delay between their starting times, or, their phase lengths are different. In this scenario, the main challenge is to solve rendezvous despite arbitrary asymmetry that we cannot control, i.e., we can think of an adversary that chooses the delay between starting times and chooses the labels for the starting nodes of the agents. Stage 0 plays a central role in ensuring that rendezvous will occur in this situation: we will find a time interval when one agent is executing a significantly long waiting period in stage 0 of a phase (either due to a large value of $g$, or, a large value of $\log^*(v_x)$), while the other agent is performing a search of its own $D$-neighbourhood in stage 1 or 2 of some phase with $g \geq D$. 

The analysis proceeds according to the following roadmap. In \Cref{basic}, we introduce terminology and prove some basic facts that will be used throughout the rest of the analysis. Then, we separately consider two cases: $\log^*{v_\alpha} = \log^*{v_\beta}$ (\Cref{samelogstar}) and $\log^*{v_\alpha} \neq \log^*{v_\beta}$ (\Cref{difflogstar}). Within each of these cases, we will consider subcases based on whether the delay between the agents' start times is small or large. In the $\log^*{v_\alpha} \neq \log^*{v_\beta}$ case, there are further subcases to consider based on which of the two values is bigger, and different approaches are needed depending on the ratio between $\log^*{v_\alpha}$ and $\log^*{v_\beta}$.

\subsubsection{Terminology and basic facts}\label{basic}
Consider an agent $x \in \{\alpha,\beta\}$. For the analysis, we group together various phases of the execution based on the value of $d$ used in the phase. In particular, an \emph{epoch} $\ep{j}{x}$ is defined as the time interval consisting of all phases $\Pg{g}{x}$ executed by agent $x$ such that $1+\floor*{\log_2 g}=j$. For example, $\ep{1}{x}= \Pg{1}{x}$, $\ep{2}{x} = \Pg{2}{x} \cup \Pg{3}{x}$, and $\ep{3}{x} = \Pg{4}{x}\cup\Pg{5}{x}\cup\Pg{6}{x}\cup\Pg{7}{x}$. More generally, for each $j \geq 1$, the first phase of epoch $\ep{j}{x}$ is $\Pg{2^{j-1}}{x}$, and the last phase of epoch $\ep{j}{x}$ is $\Pg{2^j-1}{x}$. Further, the number of phases in epoch $\ep{j}{x}$ is $2^{j-1}$. 

Many of our arguments will depend on calculations involving the length of certain phases and epochs. We denote by $|\Pg{g}{x}|$ the number of rounds in phase $\Pg{g}{x}$, and we denote by $|\ep{j}{x}|$ the number of rounds in epoch $\ep{j}{x}$. Also, we will often want to consider the number of rounds that have elapsed so far in an execution. In particular, for an agent $x$, we will consider the round in which $x$ starts a particular epoch, and we will set out to find a bound on the number of rounds that have elapsed up to that point in time. We denote by $\first{M}{x}$ the number of rounds that elapse during the first $M$ epochs of agent $x$'s execution, i.e., $\first{M}{x} := \sum\limits_{j=1}^{M} |\ep{j}{x}|$.

\subsubsection*{Facts about the lengths of phases and epochs}
In this section, we compute the lengths of phases and epochs, and provide some useful relationships between them.  We start by computing the exact lengths of each phase and each epoch.

\begin{fact}\label{fact:onephase}
	Consider an agent $x$ whose execution starts at a node labeled $v_x$. For any $g \geq 1$ and any $j \geq 1$ such that $\Pg{g}{x}$ is in epoch $\ep{j}{x}$, we have $|\Pg{g}{x}| = (72 \cdot 2^{j}+4g)\cdot\klog{v_x}$.
\end{fact}
\begin{proof}
	Stage 0 consists of $36\cdot 2^{j}\cdot \klog{v_x}$ rounds. Stage 1 consists of $4  \cdot g \cdot \klog{v_x}$ rounds. Stage 2 consists of $9 \cdot 4 \cdot 2^{j}\cdot \klog{v_x} = 36\cdot 2^{j}\cdot \klog{v_x}$ rounds. Adding up these expressions gives the desired result.
\end{proof}

\begin{fact}\label{fact:epochlength}
	Consider an agent $x$ whose execution starts at a node labeled $v_x$. For any $j \geq 1$, we have $|\ep{j}{x}| = 2^{j-1}\left[ 75\cdot 2^{j}  -2\right]\cdot\klog{v_x}$.
\end{fact}
\begin{proof}
	For $j \geq 1$, each epoch $\ep{j}{x}$ consists of $2^{j-1}$ phases $\Pg{2^{j-1}}{x},\ldots,\Pg{2^j-1}{x}$. For each $g \in \{2^{j-1},\ldots,2^j-1\}$, each phase $\Pg{g}{x}$ has length $(72 \cdot 2^{j}+4g)\cdot\klog{v_x}$. Thus, the length of epoch $\ep{j}{x}$ is given by
	\begin{align*}
		& \sum_{g=2^{j-1}}^{2^j-1} \left[ (72 \cdot 2^{j}+4g)\cdot \klog{v_x}\right] \\
		=\ & 4\klog{v_x}\left[ \left[\sum_{g=2^{j-1}}^{2^j-1} 18 \cdot 2^{j}\right] + \left[\sum_{g=2^{j-1}}^{2^j-1} g \right] \right]\\
		=\ & 4\klog{v_x}\left[ \left[\left(2^{j-1}\right)\left(18 \cdot 2^{j}\right)\right] + \left[\sum_{g=1}^{2^{j-1}}(g+2^{j-1} - 1) \right] \right]\\
		=\ & 4\klog{v_x}\left[ \left[\left(2^{j-1}\right)\left(18 \cdot 2^{j}\right)\right] 
		+ \left[\left[\sum_{g=1}^{2^{j-1}}g\right] + \left[\sum_{g=1}^{2^{j-1}}(2^{j-1} - 1)\right] \right]
		\right]\\
		=\ & 4\klog{v_x}\left[ \left[\left(2^{j-1}\right)\left(18 \cdot 2^{j}\right)\right] 
		+ \left[\left[\frac{2^{j-1}(2^{j-1}+1)}{2}\right] + \left[(2^{j-1})(2^{j-1} - 1)\right] \right]
		\right]\\
		=\ & 4\klog{v_x}\left[ \left[\left(2^{j-1}\right)\left(18 \cdot 2^{j}\right)\right] 
		+ \left[\frac{2^{j-1}(2^{j-1}+1) + 2^j(2^{j-1}-1)}{2} \right]
		\right]\\
		=\ & \klog{v_x}\left[ \left[\left(2^{j+1}\right)\left(18 \cdot 2^{j}\right)\right] 
		+ \left[{2^{j}(2^{j-1}+1) + 2^{j+1}(2^{j-1}-1)} \right]
		\right]\\
		=\ & \klog{v_x}\left[ 18\cdot 2^{2j+1} + 2^{2j-1} + 2^{j} + 2^{2j} - 2^{j+1}
		\right]\\
		=\ & \klog{v_x}\cdot 2^{j-1}\left[ 72\cdot 2^{j} + 2^{j} + 2^{1} + 2\cdot 2^{j} - 2^{2}
		\right]\\
		=\ & \klog{v_x}\cdot 2^{j-1}\left[ 75\cdot 2^{j}  -2 
		\right]\\
	\end{align*}
\end{proof}

Using the length of each individual epoch, we compute the number of rounds that elapse in the first $M \geq 1$ epochs of an execution. Once established, this gives an expression that can be used to derive an upper bound on the running time of the algorithm's execution.
\begin{fact}\label{fact:epochsum}
	For any agent $x$ and any integer $M \geq 1$, the value of $\first{M}{x}$ is equal to $2\klog{v_x}\left[25\cdot 2^{2M} - 2^{M} - 24\right]$.
\end{fact}
\begin{proof}
	By \Cref{fact:epochlength}, we have that $|\ep{j}{x}| = 2^{j-1}\left[ 75\cdot 2^{j}  -2\right]\cdot\klog{v_x}$ for each $j \in \{1,\ldots,M\}$. So,
	\begin{align*}
		\first{M}{x} & = \sum_{j=1}^{M} |\ep{j}{x}|\\
		& = \sum_{j=1}^{M} \left[2^{j-1}\left[ 75\cdot 2^{j}  -2\right]\cdot\klog{v_x}\right]\\
		& = \klog{v_x}\sum_{j=1}^{M} \left[2^{j-1}\left[ 75\cdot 2^{j}  -2\right]\right]\\
		& = \klog{v_x}\left[150\sum_{j=1}^{M} 2^{2j-2} - \sum_{j=1}^{M} 2^j\right]\\
		& = \klog{v_x}\left[150\sum_{j=1}^{M} 4^{j-1} - \sum_{j=1}^{M} 2^j\right]\\
		& = \klog{v_x}\left[150\left[\frac{4^M-1}{3}\right] - \left[\frac{2(2^M-1)}{1}\right]\right]\\
		& = \klog{v_x}\left[50\cdot 2^{2M} - 50 - 2^{M+1}+2\right]\\
		& = 2\klog{v_x}\left[25\cdot 2^{2M} - 2^{M} - 24\right]
	\end{align*}
\end{proof}

\begin{corollary}\label{cor:runningtime}
	If rendezvous occurs by the end of epoch $\ep{M}{x}$ in agent $x$'s execution, then the total number of rounds that elapse in $x$'s execution of the algorithm is $O(2^{2M}\log^*(v_x))$.
\end{corollary}

The following facts provide useful relationships between the lengths of epochs and phases. The first fact shows that each epoch is at least twice as large as the previous one.
\begin{fact}\label{fact:epochincreasing}
	For any agent $x$ and any $j \geq 1$, we have $|\ep{j}{x}| < \frac{1}{2}|\ep{j+1}{x}|$. This implies that $\first{M}{x} = \sum_{j=1}^{M} |\ep{j}{x}| < |\ep{M+1}{x}|$.
\end{fact}

The following fact states that the first phase in an epoch has length bounded above by the average phase length in that epoch. This is because the first phase in an epoch is the shortest.
\begin{fact}\label{fact:firstphasefraction}
	For any agent $x$ and any $j \geq 1$, the first phase $\Pg{2^{j-1}}{x}$ of epoch $\ep{j}{x}$ has length $|\Pg{2^{j-1}}{x}| \leq \frac{1}{2^{j-1}}|\ep{j}{x}|$.
\end{fact}

Epoch $\ep{1}{x}$ consists of one phase, and epoch $\ep{2}{x}$ consists of two phases, and the final phase of an epoch is always the longest. This means that the majority of the rounds in each of these epochs are in the final phase. However, the following fact demonstrates that this is no longer true starting from epoch $\ep{3}{x}$ onwards, i.e., the final phase is strictly contained in the latter half of the epoch.
\begin{fact}\label{fact:finalphasefraction}
	For any agent $x$ and any $j \geq 3$, the last phase $\Pg{2^{j}-1}{x}$ of epoch $\ep{j}{x}$ has length $|\Pg{2^{j}-1}{x}| < \frac{1}{2}|\ep{j}{x}|$.
\end{fact}
\begin{proof}
	By \Cref{fact:onephase}, 
	\begin{align*}
		|\Pg{2^{j}-1}{x}| & = (72 \cdot 2^{j}+4(2^{j}-1))\cdot\klog{v_x}\\
		& = (76 \cdot 2^{j} - 4)\cdot\klog{v_x}\\
		& = \klog{v_x}\cdot 2^1(38 \cdot 2^{j} - 2)\\
		& \leq \klog{v_x}\cdot 2^{j-2}(38 \cdot 2^{j} - 2)\\
		& = \frac{1}{2}\klog{v_x}\cdot 2^{j-1}(38 \cdot 2^{j} - 2)\\
		& < \frac{1}{2}\klog{v_x}\cdot 2^{j-1}(75 \cdot 2^{j} - 2)\\
		& = \frac{1}{2}|\ep{j}{x}|
	\end{align*}
\end{proof}

\subsubsection*{Facts about searching and rendezvous}
In this section, we provide some facts that will be crucial in our proofs that rendezvous occurs. Our approach to guaranteeing rendezvous can be summarized as follows: for some round $\tau$, we establish that one agent $y$ will wait at its starting node for some $Y$ consecutive rounds starting at $\tau$, and, the other agent $x$ will visit every node in the $D$-neighbourhood of $x$'s starting node within some $X \leq Y$ consecutive rounds starting at $\tau$. We can think of agent $x$ performing a `search' for agent $y$, and by establishing the above condition, we guarantee that agent $x$ will `find' agent $y$ at $y$'s starting node within the $X$ rounds following $\tau$.

The following technical fact will enable us to find an explicit upper bound to use as $X$ in the approach described above. In particular, for an arbitrary round $\tau$ in a phase $\Pg{g}{x}$, it gives an upper bound on how many rounds will elapse before $x$ has explored all nodes within the $g$-neighbourhood of $x$'s starting node.

\begin{fact}\label{willdosearch}
	Consider an agent $x$ with starting node $v_x$. Consider any fixed $j \geq 1$, and suppose that agent $x$'s execution is at some round $\tau$ in some phase $\Pg{g}{x}$ with $g \leq 2^j-1$ (i.e., its execution is in an epoch with index at most $j$). If $\tau$ is before the final block of Stage 2 of phase $\Pg{2^{j} - 1}{x}$, then agent $x$ visits all nodes in the $g$-neighbourhood of $v_x$ within the $44\klog{v_x}2^{j}$ rounds following $\tau$.
\end{fact}
\begin{proof}
	Suppose that $x$'s execution has not yet reached the final block of Stage 2 of phase $\Pg{2^{j} - 1}{x}$. Since $\Pg{2^{j} - 1}{x}$ is the final phase of epoch $\ep{j}{x}$, this means that $x$'s execution is in some epoch $\ep{{j'}}{x}$ with ${j'} \leq j$. In what follows, we will use the fact that $g \leq 2^{j'}$, which is true since $g = 2^{\log_2 g} \leq 2^{1+\floor*{\log_2 g}} = 2^{j'}$.
	
	First, suppose agent $x$ is in Stage 0 of phase $\Pg{g}{x}$. The length of Stage 0 is $36\klog{v_x}2^{j'}$ rounds, so within this many rounds, agent $x$ will start Stage 1 of phase $\Pg{g}{x}$. In Stage 1, agent $x$ visits all nodes in the $g$-neighbourhood of $v_x$ using $4\klog{v_x}g \leq 4\klog{v_x}2^{j'}$ rounds. So, we have shown that $x$ will visit all nodes in the $g$-neighbourhood of $v_x$ within the next $36\klog{v_x}2^{{j'}}+4\klog{v_x}2^{{j'}}=40\klog{v_x}2^{{j'}}<44\klog{v_x}2^{{j'}} \leq 44\klog{v_x}2^{j}$ rounds, as desired.
	
	Next, suppose agent $x$ is in Stage 1 of phase $\Pg{g}{x}$. The length of Stage 1 is $4\klog{v_x}g \leq 4\klog{v_x}2^{j'}$ rounds, so within this many rounds, agent $x$ will start Stage 2 of phase $\Pg{g}{x}$. Since the first three bits of $S$ contain at least one 1, it follows that agent $x$ will perform a searching block within the first three blocks of Stage 2. In such a searching block, agent $x$ will visit all nodes in the $(2^{j'})$-neighbourhood of $v_x$, which contains the $g$-neighbourhood of $v_x$ since $g \leq 2^{j'}$. As each block has length exactly $4\cdot 2^{j'}$, it follows that agent $x$ will visit all nodes in the $g$-neighbourhood of $v_x$ within $12\cdot 2^{j'}$ rounds. So, we have shown that $x$ will visit all nodes in the $g$-neighbourhood of $v_x$ within the next $4\klog{v_x}2^{j'} + 12\cdot 2^{j'} \leq 16\klog{v_x}2^{{j'}}<44\klog{v_x}2^{{j'}} \leq 44\klog{v_x}2^{j}$ rounds, as desired.
	
	Next suppose that agent $x$ is in Stage 2 of phase $\Pg{g}{x}$. There are two sub-cases to consider:
	\begin{itemize}
		\item Case 1: Agent $x$ is executing a block before the final block of Stage 2 of phase $\Pg{g}{x}$.\\
		By the construction of $S$, the bits equal to 0 occur in pairs, and each such pair is immediately followed by a 1. It follows that, from any round in Stage 2 other than inside the final block, a searching block will begin within 3 blocks. In such a block, agent $x$ visits all nodes in the $(2^{j'})$-neighbourhood of $v_x$, which includes the $g$-neighbourhood of $v_x$ since $g \leq 2^{j'}$. Due to the block length, this means that agent $x$ will visit all nodes in the $g$-neighbourhood of $v_x$ within $16\cdot 2^{j'}$ rounds. So, we have shown that $x$ will visit all nodes in the $g$-neighbourhood of $v_x$ within the next $16\cdot 2^{j'} <44\klog{v_x}2^{{j'}} \leq 44\klog{v_x}2^{j}$ rounds, as desired.
		
		\item Case 2: Agent $x$ is executing the final block of phase $\Pg{g}{x}$ with $g < 2^{j}-1$.\\
		The final block of phase $\Pg{g}{x}$ consists of $4\cdot 2^{j'}$ rounds, and then phase $\Pg{g+1}{x}$ begins. Note that $\Pg{g+1}{x}$ belongs to an epoch $E_{j''}$ with $j'' \in \{{j'},{j'}+1\}$, but, from the assumption $g < 2^{j}-1$, we know that $j'' \leq j$ (in other words, since $\Pg{g}{x}$ is before the final phase of epoch $\ep{j}{x}$, it follows that $\Pg{g+1}{x}$ is in an epoch before $\ep{j+1}{x}$). The length of Stage 0 in phase $\Pg{g+1}{x}$ is $36\klog{v_x}2^{j''}$ rounds, so within this many rounds, agent $x$ will start Stage 1 of phase $\Pg{g+1}{x}$. In Stage 1, agent $x$ visits all nodes in the $(g+1)$-neighbourhood of $v_x$ using $4\klog{v_x}(g+1) = 4\klog{v_x}2^{\log_2(g+1)} \leq 4\klog{v_x}2^{1+\floor*{\log_2(g+1)}}= 4\klog{v_x}2^{j''}$ rounds. Since the $(g+1)$-neighbourhood of $v_x$ includes its $g$-neighbourhood, $x$ will visit all nodes in the $g$-neighbourhood of $v_x$ in the next $4\cdot 2^{{j'}}+ 36\klog{v_x}2^{j''}+4\klog{v_x}2^{j''}\leq 44\klog{v_x}2^{j''} \leq 44\klog{v_x}2^{j}$ rounds, as desired.
	\end{itemize}
\end{proof}

In early phases of the algorithm's execution, it is often the case that rendezvous will not occur, since an agent is never searching far enough for the other agent. However, our algorithm has the desirable property that, from some epoch onward, all searches will have radius at least $D$, so they will include the starting node of the other agent. This epoch will play a critical role in the remainder of the analysis, which motivates the following definition and facts.

\begin{definition}
	Denote by $\dcrit$ the integer $1+\floor*{\log_2 D}$. For any agent $x$, the epoch $\ep{\dcrit}{x}$ is called the \emph{critical epoch}.
\end{definition}

\begin{fact}\label{fact:containsD}
	For any agent $x$, the epoch $\ep{\dcrit}{x}$ contains the phase $\Pg{D}{x}$.
\end{fact}

\begin{fact}\label{fact:ggeqD}
	Consider any agent $x$ and any integer $j \geq \dcrit+1$. For any phase $\Pg{g}{x}$ contained in epoch $\ep{j}{x}$, we have $2^j > 2^{\dcrit} \geq D$ and $g \geq D$.
\end{fact}
\begin{proof}
	Consider any $j \geq \dcrit+1$, and note that $2^{j-1} \geq 2^{\dcrit} = 2^{1+\floor*{\log_2 D}} \geq 2^{\log_2 D} \geq D$. Next, suppose that $\Pg{g}{x}$ is contained in epoch $\ep{j}{x}$. By the definition of $\ep{j}{x}$, we have that $g \geq 2^{j-1}$, and it follows from above that $g \geq D$.
\end{proof}

Combining \Cref{willdosearch,fact:ggeqD}, we get the following useful lemma, which will be used to prove that rendezvous occurs by pairing it with an argument that the other agent $y$ will wait at its own starting node for at least the next $44\klog{v_x}2^{j}$ rounds.
\begin{lemma}\label{lemma:willdoDsearch}
	Consider an agent $x$ with starting node $v_x$. For any $j \geq \dcrit+1$, suppose that agent $x$'s execution is in epoch $\ep{j}{x}$ but has not yet reached the final block of Stage 2 of the final phase of epoch $\ep{j}{x}$. Then agent $x$ visits all nodes in the $D$-neighbourhood of $v_x$ within the next $44\klog{v_x}2^{j}$ rounds.
\end{lemma}

We conclude this section with a proof that the algorithm solves rendezvous in the case where the delay between agent starting times is very large. Essentially, if agent $\beta$ has not started its execution before agent $\alpha$ executes a phase with guess $g \geq D$, then rendezvous will occur when $\alpha$ visits $v_\beta$ in that phase.

\begin{lemma}\label{verylargedelay}
	Suppose that agent $\beta$ does not start its execution before agent $\alpha$ completes the epoch $\ep{\dcrit}{\alpha}$, i.e., $\delay \geq \first{\dcrit}{\alpha}$. Then rendezvous occurs by the end of $\alpha$'s execution of epoch $\ep{\dcrit}{\alpha}$, and the number of rounds that elapse in $\alpha$'s execution before rendezvous is $O(D^2\log^*(v_\alpha))$.
\end{lemma}
\begin{proof}
	By \Cref{fact:containsD}, epoch $\ep{\dcrit}{\alpha}$ contains the phase $\Pg{D}{\alpha}$. Stage 1 of $\Pg{D}{\alpha}$ consists of $\alpha$ visiting all nodes in the $D$-neighbourhood of $v_\alpha$. If agent $\beta$ has not started its execution while $\alpha$ executes this phase, then $\beta$ is located at its starting node $v_\beta$, which is at distance $D$ from $v_\alpha$, so rendezvous occurs (if not earlier). By \Cref{cor:runningtime}, the total number of rounds that elapse in $\alpha$'s execution of the algorithm is $O(2^{\dcrit}\log^*(v_\alpha)) = O(D^2\log^*(v_\alpha))$.
\end{proof}

\subsubsection{Algorithm analysis in the case where \texorpdfstring{$\log^*(v_\alpha) = \log^*(v_\beta)$}{log*(vα) = log*(vβ)}}\label{samelogstar}

In this section, we analyze the algorithm in the case where $\log^*(v_\alpha) = \log^*(v_\beta)$. Under this assumption, we observe from \Cref{fact:onephase,fact:epochlength} that the lengths of phases and epochs are the same across both agents' executions, i.e., for all $g \geq 1$ and $j \geq 1$, we can exploit the fact that $|\Pg{g}{\alpha}|=|\Pg{g}{\beta}|$ and $|\ep{j}{\alpha}| = |\ep{j}{\beta}|$.

\Cref{verylargedelay} handles the case where the delay between the starting times of the agents is very large. The next case we consider is when the delay between starting times is very small, i.e., bounded above by the length of one block in Stage 2 of the $D$-th phase. In this case, both agents will be in Stage 2 of $D$-th phase at roughly the same time, and the distinct colours assigned to the two starting nodes in Stage 1 guarantee that one agent will perform a searching block while the other agent is performing waiting blocks (so rendezvous occurs). 

\begin{lemma}\label{smalldelay}
	Suppose that $\delay \leq 4\cdot 2^{\dcrit}$. Then rendezvous occurs by the end of $\alpha$'s execution of phase $\Pg{D}{\alpha}$, i.e., by the end of epoch $\ep{\dcrit}{\alpha}$.
\end{lemma}
\begin{proof}
	Consider $\alpha$'s execution of Stage 2 of phase $\Pg{D}{\alpha}$. From the description of the algorithm and \Cref{fact:containsD}, recall that the length of each block of Stage 2 in phase $\Pg{D}{\alpha}$ is $4\cdot 2^{\dcrit}$ rounds. So, by the assumption on delay, when agent $\alpha$ starts a block $i \geq 2$ in Stage 2 of phase $\Pg{D}{\alpha}$, agent $\beta$ is also starting block $i$, or, is in block $i-1$ of Stage 2 of phase $\Pg{D}{\beta}$.
	
	Since both $\alpha$ and $\beta$ are in a phase where $g=D$, and distance between their starting positions is exactly $D$, their value of $CV$ assigned in Stage 1 by the Cole-Vishkin algorithm differs.
	
	\begin{itemize}
		\item Case 1: The $CV$ values differ at the first bit\\
		Then bits 1-4 of $S$ are 0011 for one agent, and 1100 for the other. 
		\begin{itemize}
			\item Case (a): Bits 1-4 of $S$ for $\alpha$ are 0011.\\
			Consider the start of block 4 of agent $\alpha$'s execution of Stage 2 in phase $\Pg{D}{\alpha}$. In the next $4\cdot 2^{\dcrit}$ rounds, agent $\alpha$ visits all nodes in the $(2^{\dcrit})$-neighbourhood of $v_\alpha$ (which contains the entire $D$-neighbourhood of $v_\alpha$). Meanwhile, as bits 3-4 at $\beta$ are both 0, agent $\beta$ is waiting at its starting node $v_\beta$, so rendezvous occurs during block 4 of agent $\alpha$'s execution of Stage 2 in phase $\Pg{D}{\alpha}$. 
			\item Case (b): Bits 1-4 of $S$ for $\alpha$ are 1100.\\
			Consider the start of block 2 of agent $\alpha$'s execution of Stage 2 in phase $\Pg{D}{\alpha}$. In the next $4\cdot 2^{\dcrit}$ rounds, agent $\alpha$ visits all nodes in the $(2^{\dcrit})$-neighbourhood of $v_\alpha$ (which contains the entire $D$-neighbourhood of $v_\alpha$). Meanwhile, as bits 1-2 at $\beta$ are both 0, agent $\beta$ is waiting at its starting node $v_\beta$, so rendezvous occurs during block 2 of agent $\alpha$'s execution of Stage 2 in phase $\Pg{D}{\alpha}$. 
		\end{itemize}
		\item Case 2: The $CV$ values differ at the second bit\\
		Then bits 5-8 of $S$ are 0011 for one agent, and 1100 for the other.
		\begin{itemize}
			\item Case (a): Bits 5-8 of $S$ for $\alpha$ are 0011.\\
			Consider the start of block 8 of agent $\alpha$'s execution of Stage 2 in phase $\Pg{D}{\alpha}$. In the next $4\cdot 2^{\dcrit}$ rounds, agent $\alpha$ visits all nodes in the $(2^{\dcrit})$-neighbourhood of $v_\alpha$ (which contains the entire $D$-neighbourhood of $v_\alpha$). Meanwhile, as bits 7-8 at $\beta$ are both 0, agent $\beta$ is waiting at its starting node $v_\beta$, so rendezvous occurs during block 8 of agent $\alpha$'s execution of Stage 2 in phase $\Pg{D}{\alpha}$. 
			\item Case (b): Bits 5-8 of $S$ for $\alpha$ are 1100.\\
			Consider the start of block 6 of agent $\alpha$'s execution of Stage 2 in phase $\Pg{D}{\alpha}$. In the next $4\cdot 2^{\dcrit}$ rounds, agent $\alpha$ visits all nodes in the $(2^{\dcrit})$-neighbourhood of $v_\alpha$ (which contains the entire $D$-neighbourhood of $v_\alpha$). Meanwhile, as bits 5-6 at $\beta$ are both 0, agent $\beta$ is waiting at its starting node, so rendezvous occurs during block 6 of agent $\alpha$'s execution of Stage 2 in phase $\Pg{D}{\alpha}$. 
		\end{itemize}
	\end{itemize}
\end{proof}

Finally, we consider the remaining case: we assume that agent $\beta$ starts its execution before agent $\alpha$ completes its execution of epoch $\ep{\dcrit}{\alpha}$, but, we assume that the delay is not very small, i.e., we assume that the delay is strictly larger than the length of one block in Stage 2 of phase $D$. In particular, in what follows, we assume that $4\cdot 2^{\dcrit} < \delay < \first{\dcrit}{\alpha}$. At a high level, the analysis starts by considering the moment that agent $\beta$'s execution has reached the start of phase $\Pg{D}{\beta}$, so that all future searches by $\beta$ will include $\alpha$'s starting node. Then, we consider the earliest round, after that moment, in which $\alpha$ starts a new epoch. The lower bound on the delay guarantees that, when $\alpha$ starts this new epoch, agent $\beta$'s execution is far enough behind such that $\beta$'s next search happens relatively soon compared to the length of $\alpha$'s waiting period at the start of the new epoch, and this guarantees that rendezvous occurs. The upper bound on the delay guarantees that this new epoch in $\alpha$'s execution has a relatively small index, which will lead to the desired upper bound on the running time. 

We start by proving a lower bound on the delay between the agents. To do so, we consider how `far ahead' agent $\alpha$'s execution is at the moment agent $\beta$ starts phase $\Pg{D}{\alpha}$.

\begin{lemma}\label{lemma:hasdelay}
	Suppose that when agent $\beta$ starts phase $\Pg{D}{\beta}$, agent $\alpha$'s execution is in a phase in epoch $\ep{\dcrit+ m}{\alpha}$ for some integer $m \geq 0$. Then, $\delay > 4\cdot (2^{\dcrit+m})$.
\end{lemma}
\begin{proof}
	If $m=0$, then the result follows directly from the assumption that $\delay > 4\cdot 2^{\dcrit}$. The remainder of the proof handles the cases where $m \geq 1$.
	
	By \Cref{fact:containsD}, phase $\Pg{D}{\beta}$ is in epoch $\ep{\dcrit}{\beta}$. This means that, when $\beta$ starts phase $\Pg{D}{\beta}$, it has not started Stage 2 of the final phase of epoch $\ep{\dcrit}{\beta}$. Moreover, it has not started Stage 2 of the final phase of any epoch $\ep{\dcrit+k}{\beta}$ for any $k \geq 1$.
	
	Suppose that, when $\beta$ starts phase $\Pg{D}{\beta}$, agent $\alpha$'s execution is in a phase in epoch $\ep{\dcrit+m}{\alpha}$ for some integer $m \geq 1$. From above, agent $\beta$ has not started Stage 2 of the final phase of epoch $\ep{\dcrit+m-1}{\beta}$. In particular, this means that the delay between the starting times of the agents is at least equal to the number of rounds in Stage 2 of the final phase of epoch $\ep{\dcrit+m-1}{\beta}$. Stage 2 of this phase consists of at least 9 blocks of length $4\cdot (2^{\dcrit+m-1})$ rounds. But $9\cdot4\cdot(2^{\dcrit+m-1}) > 2\cdot4\cdot(2^{\dcrit+m-1}) = 4\cdot (2^{\dcrit+m})$, as desired.
\end{proof}

Using the previous result, we can argue that, when $\alpha$ starts a new epoch, agent $\beta$'s execution is not in the final block of the epoch that $\alpha$ just finished. Then, we can apply \Cref{lemma:willdoDsearch} to get an upper bound on how soon $\beta$ will visit all nodes in the $D$-neighbourhood of its starting node, and show that $\alpha$'s waiting period at the start of the new epoch is at least as large as that bound (which guarantees that rendezvous occurs).
\begin{lemma}\label{AlphaAlsoD}
	Consider an arbitrary integer $m \geq 0$, and suppose that when agent $\beta$ starts phase $\Pg{D}{\beta}$, agent $\alpha$'s execution is in some phase in epoch $\ep{\dcrit+m}{\alpha}$. Then rendezvous occurs by the end of $\alpha$'s execution of epoch $\ep{\dcrit+m+1}{\alpha}$.
\end{lemma}
\begin{proof}
	Let $F = 2^{\dcrit+m}-1$. By the definition of epochs, phase $\Pg{F}{\alpha}$ is the final phase of epoch $\ep{\dcrit+m}{\alpha}$. So, it follows from the description of the algorithm, that each block in Stage 2 of phase $\Pg{F}{\alpha}$ consists of $4\cdot 2^{\dcrit+m}$ rounds.
	
	We consider the moment that $\alpha$ reaches the end of phase $\Pg{F}{\alpha}$ in its execution, i.e., the moment that $\alpha$ reaches the end of epoch $\ep{\dcrit+m}{\alpha}$ in its execution. By assumption, this is after $\beta$ starts its execution of phase $\Pg{D}{\beta}$. So, $\beta$ is executing a phase $\Pg{g}{\beta}$ with $g \geq D$ when $\alpha$ reaches the end of epoch $\ep{\dcrit+m}{\alpha}$. Further, by \Cref{lemma:hasdelay}, we know that $\delay > 4\cdot (2^{\dcrit+m})$. So, when $\alpha$ reaches the end of epoch $\ep{\dcrit+m}{\alpha}$, agent $\beta$ has not yet reached the final block of Stage 2 in its execution of phase $\Pg{F}{\beta}$ (as each block has length exactly $4\cdot 2^{\dcrit+m}$ rounds). By \Cref{willdosearch} with $j = \dcrit+m$, agent $\beta$ will visit all nodes in the $g$-neighbourhood of $v_\beta$ (which contains the $D$-neighbourhood of $v_\beta$) within the next $44\klog{v_\beta}2^{\dcrit+m}$ rounds from when $\alpha$ reaches the end of epoch $\ep{\dcrit+m}{\alpha}$.
	
	However, the end of epoch $\ep{\dcrit+m}{\alpha}$ in $\alpha$'s execution is actually the start of epoch $\ep{\dcrit+m+1}{\alpha}$. In the first phase of this epoch, agent $\alpha$ waits at its starting node $v_\alpha$ during Stage 0, which consists of $32\cdot 2^{\dcrit+m+1}\cdot \klog{v_\alpha} = 64 \klog{v_\alpha}2^{\dcrit+m} > 44\klog{v_\alpha}2^{\dcrit+m} = 44\klog{v_\beta}2^{\dcrit+m}$ rounds. In particular, this means that $\alpha$ will be located at its starting node $v_\alpha$ when $\beta$ visits all nodes in the $D$-neighbourhood of $v_\beta$. It follows that rendezvous occurs in or before epoch $\ep{\dcrit+m+1}{\alpha}$ in $\alpha$'s execution.
\end{proof}

Finally, by determining an upper bound on the number of rounds that elapse until $\beta$'s execution reaches the start of phase $\Pg{D}{\beta}$, we can determine that $\alpha$'s execution does not get `too far ahead', i.e., $\alpha$ has not yet reached the end of epoch $\ep{\dcrit+1}{\alpha}$. Then, using \Cref{AlphaAlsoD}, we conclude that rendezvous will occur shortly after $\alpha$ starts executing a new epoch, i.e., shortly after $\alpha$ starts epoch $\ep{\dcrit+2}{\alpha}$.
\begin{lemma}\label{mediumdelay}
	Suppose that $4\cdot 2^{\dcrit} < \delay < \first{\dcrit}{\alpha}$. Then rendezvous occurs by the end of epoch $\ep{\dcrit+2}{\alpha}$ in $\alpha$'s execution.
\end{lemma}

\begin{proof}
	We determine an upper bound on the number of rounds that elapse until $\beta$ starts executing phase $\Pg{D}{\beta}$ in its execution. By \Cref{fact:containsD}, phase $\Pg{D}{\beta}$ belongs to epoch $\ep{\dcrit}{\beta}$, so it follows that an upper bound on the number of rounds that elapse before $\beta$ starts executing phase $\Pg{D}{\beta}$ is the number of rounds in the first $\dcrit$ epochs, i.e., $\first{\dcrit}{\beta}$. By \Cref{fact:epochincreasing}, this value is strictly less than $|\ep{\dcrit + 1}{\beta}|$. Thus, we have shown that when $\beta$ starts phase $\Pg{D}{\beta}$, strictly fewer than $|\ep{\dcrit + 1}{\beta}|$ rounds have elapsed since the start of $\beta$'s execution.
	
	From the assumption that $\delay < \first{\dcrit}{\alpha}$, we know that agent $\beta$ starts its execution before agent $\alpha$ completes its execution of epoch $\ep{\dcrit}{\alpha}$. It follows that $\alpha$ has not started epoch $\ep{\dcrit+1}{\alpha}$ when $\beta$ starts its execution. By the upper bound in the previous paragraph, strictly fewer than $|\ep{\dcrit + 1}{\beta}|$ rounds elapse from the start of $\beta$'s execution until $\beta$ starts phase $\Pg{D}{\beta}$. Since $|\ep{\dcrit + 1}{\beta}|=|\ep{\dcrit + 1}{\alpha}|$, it follows that when $\beta$ starts phase $\Pg{D}{\beta}$, agent $\alpha$'s execution has not yet reached the end of epoch $\ep{\dcrit+1}{\alpha}$. Moreover, when $\beta$ starts phase $\Pg{D}{\beta}$, agent $\alpha$'s execution is already past the start of phase $\Pg{D}{\alpha}$ (this is because we assumed that $\alpha$'s execution does not start before $\beta$'s, and the phases/epochs in both executions have the same length). It follows that $\alpha$'s execution is in some epoch $\ep{\dcrit+m}{\alpha}$ for some $m \geq 0$.
	
	Therefore, we have shown that when $\beta$ starts phase $\Pg{D}{\beta}$, agent $\alpha$'s execution is in some phase in an epoch $\ep{\dcrit+m}{\alpha}$ where $0 \leq m \leq 1$. By \Cref{AlphaAlsoD}, it follows that rendezvous occurs by the end of $\alpha$'s execution of epoch $\ep{\dcrit+2}{\alpha}$.
\end{proof}

Combining \Cref{verylargedelay,smalldelay,mediumdelay}, we get that rendezvous occurs by the end of epoch $\ep{\dcrit+2}{\alpha}$ in agent $\alpha$'s execution. Applying \Cref{cor:runningtime} completes the analysis of the algorithm in the case where $\log^*{\alpha} = \log^*{\beta}$.
\begin{lemma}\label{lem:analysissamelogstar}
	Rendezvous occurs within $O(D^2\log^*(v_\alpha))=O(D^2\log^*(v_\beta)) = O(D^2\log^*\ell) $ rounds in agent $\alpha$'s execution.
\end{lemma}
\begin{proof}
	Let $M = \dcrit+2$. By \Cref{verylargedelay,smalldelay,mediumdelay}, we get that rendezvous occurs by the end of epoch $\ep{M}{\alpha}$ in agent $\alpha$'s execution. By \Cref{cor:runningtime}, the number of rounds that elapse in agent $\alpha$'s execution by the end of epoch $\ep{M}{\alpha}$ is $O(2^{2M}\log^*(v_\alpha)) = O(2^{2M}\log^*(v_\beta))$, and note that $2^{2M} = 2^{ 2\dcrit+4} = 2^{2 + 2\floor*{\log_2 D} + 4} = 64\cdot (2^{\floor*{\log_2 D}})^2 \leq 64\cdot D^2$.
\end{proof}

\subsubsection{Algorithm analysis in the case where \texorpdfstring{$\log^*(v_\alpha) \neq \log^*(v_\beta)$}{log*(vα) ≠ log*(vβ)}}\label{difflogstar}

In this section, we analyze the algorithm in the case where $\log^*(v_\alpha) \neq \log^*(v_\beta)$. \Cref{verylargedelay} handles the case where the delay between the starting times of the agents is very large. So, in the remainder of the analysis, we assume that agent $\beta$'s execution starts before agent $\alpha$ completes epoch $\ep{\dcrit}{\alpha}$, i.e., $\delay < \first{\dcrit}{\alpha}$.

We observe from \Cref{fact:onephase,fact:epochlength} that the lengths of phases and epochs are not the same across both agents' executions, i.e., for all $g \geq 1$ and $j \geq 1$, we know that  $|\Pg{g}{\alpha}|\neq|\Pg{g}{\beta}|$ and $|\ep{j}{\alpha}| \neq |\ep{j}{\beta}|$. As a result, the analysis must be carried out more carefully. Once again, the overall idea is to ensure that there is a time interval during which one agent (``the searcher'') is visiting all nodes in the $D$-neighbourhood of its starting node, while the other agent (``the sitter'') remains stationary at its starting node during the entire time interval. While in stage 0 of a phase, an agent is attempting to be the sitter, and while in stages 1 and 2, an agent is attempting to be the searcher. However, as seen in the description of the algorithm, the lengths of these stages depend on several values: the current phase index, the current epoch index, and the value of $\log^*{v_x}$ where $v_x$ is the label of the starting node of agent $x$. As a result, there are several different factors that go into determining whether a sitter is waiting long enough at its starting node for the searcher to find it:
\begin{itemize}
	\item The delay between the starting times of the agents affects how `far ahead' one agent is, i.e., agent $\alpha$ can be executing a phase or epoch with a much larger index when agent $\beta$ starts its execution.
	\item The lengths of phases and epochs affects how `far ahead' one agent can get, since an agent with shorter phases and epochs will move through the phases and epochs more quickly. There are two possible situations to worry about: if agent $\alpha$ starts its execution much earlier and has very short phases/epochs, then its execution gets `very far ahead'; however, if agent $\beta$ has very short phases/epochs, then its execution `catches up and passes $\alpha$' and is eventually executing phases/epochs with higher index than $\alpha$.
	\item The ratio between $\log^*{v_\alpha}$ and $\log^*{v_\beta}$ acts as an additional scaling factor. For example, even if an agent is executing a phase with much smaller phase number than the other agent, its waiting period in stage 0 can be very large due to the size of $\log^*(v_x)$, and this can make up for a large difference between phase or epoch indices.
\end{itemize}
The challenge in the analysis is to account for these factors and demonstrate that there is always an interval of time where one agent is a sitter for long enough while the other agent searches. In fact, agents $\alpha$ and $\beta$ will not have fixed roles, i.e., depending on the specific combination of the above factors, agent $\alpha$ will sometimes be the searcher and sometimes be the sitter when rendezvous occurs. As a result, the analysis needs to be broken down into several subcases based on which agent's phases are shorter, and the specific ratio between $\log^*{v_\alpha}$ and $\log^*{v_\beta}$. To this end, the remainder of the analysis is broken down into two main sections: the first considers the case where $\log^*{v_\alpha} < \log^*{v_\beta}$, and the second considers the case where $\log^*{v_\alpha} > \log^*{v_\beta}$. Within each section, the analysis is broken down based on the ratio $\mu$ between $\log^*{v_\alpha}$ and $\log^*{v_\beta}$.

\subsubsection*{The case where $\log^*(v_\alpha) < \log^*(v_\beta)$}

Define $\mu > 1$ such that $\mu\cdot\log^{*}{v_\alpha} = \log^*{v_\beta}$.
We first consider the case where $\mu$ is very large, i.e., the value of $\log^*{v_\beta}$ is much larger than the value of $\log^*{v_\alpha}$. In this case, we notice that even in the first phase of $\beta$'s execution, its waiting period in Stage 0 is extremely long, i.e., it contains enough rounds to wait for $\alpha$ to complete its first $\dcrit+1$ epochs. So, $\beta$ will be waiting at its starting node long enough for $\alpha$ to visit all nodes in the $D$-neighbourhood of $v_\alpha$ at least once, so rendezvous will occur.

\begin{lemma}\label{betabigger:largemu}
Suppose that $\mu > 12D^2$. Rendezvous occurs by the end of agent $\alpha$'s execution of epoch $\ep{\dcrit+1}{\alpha}$.
\end{lemma}
\begin{proof}
Consider the start of $\beta$'s execution. Stage 0 of phase $\Pg{1}{\beta}$ is a waiting period of length $72\klog{v_\beta} = 72\mu\klog{v_\alpha} > 864D^2\klog{v_\alpha} > 800D^2\klog{v_\alpha} = 2\klog{v_\alpha}\cdot 25\cdot 2^{4+2\log_2 D} \geq 2\klog{v_\alpha}\cdot 25\cdot 2^{4+2\floor*{\log_2 D}} = 2\klog{v_\alpha}\cdot 25\cdot 2^{2(\dcrit+1)} \geq \first{\dcrit+1}{\alpha}$ (where the last inequality uses \Cref{fact:epochsum} with $M=\dcrit+1$). This means that, during $\beta$'s waiting period in stage 0 of phase $\Pg{1}{\beta}$, there is enough time for $\alpha$ to execute at least its first $\dcrit+1$ epochs. By our assumption about $\delay$, note that $\alpha$'s execution has not yet reached the end of epoch $\ep{\dcrit}{\alpha}$. Therefore, agent $\alpha$ will execute at least one phase of epoch $\ep{\dcrit+1}{\alpha}$ while $\beta$ is executing stage 0 of phase $\Pg{1}{\beta}$. But in stage 1 of each phase of epoch $\ep{\dcrit+1}{\alpha}$, agent $\alpha$ visits all nodes in the $D$-neighbourhood of $v_\alpha$, which means that $\alpha$ will visit $v_\beta$ while agent $\beta$ is waiting there in stage 0 of its first phase. So, rendezvous occurs by the end of epoch $\ep{\dcrit+1}{\alpha}$ in $\alpha$'s execution, as desired.
\end{proof}

Next, we consider the case where $\mu$ is strictly larger than 4, but not larger than $12D^2$. In this case, we consider the round in which $\beta$ starts epoch $\ep{\dcrit+2}{\beta}$, and consider the number of rounds it takes $\beta$ to reach this round in its execution. We can then figure out how far along $\alpha$ is in its execution at this time, and, roughly speaking, show that $\alpha$ is at most $\log\mu-1$ epochs `ahead', i.e., it's searching and waiting periods are larger by a factor of at most $\mu/2$. But the fact that $\log^*(v_\beta) = \mu\cdot\log^*(v_\alpha)$ means that $\beta$'s waiting period in stage 0 has an additional $\mu$ factor built into its length, so this waiting period at the start of epoch $\ep{\dcrit+2}{\beta}$ is long enough to ensure that it includes $\alpha$'s next search, so rendezvous will occur. This idea is formalized in the following result.

\begin{lemma}\label{betabigger:mubigger4}
Suppose that $4 < \mu \leq 12D^2$. Rendezvous occurs by the end of agent $\alpha$'s execution of epoch $\ep{\dcrit+\log\mu+1}{\alpha}$.
\end{lemma}
\begin{proof}
Consider the round in which agent $\beta$ starts epoch $\ep{\dcrit+2}{\beta}$. The number of rounds that occur in $\beta$'s execution up to that round is calculated by
\begin{align*}
	\first{\dcrit+1}{\alpha}& = 2\klog{v_\beta}\left[25\cdot 2^{2(\dcrit+1)} - 2^{(\dcrit+1)} - 24\right]\\
	& = 2\mu\klog{v_\alpha}\left[25\cdot 2^{2(\dcrit+1)} - 2^{(\dcrit+1)} - 24\right]\\\\
	& = 2\klog{v_\alpha}\left[25\cdot 2^{2(\dcrit+1+\frac{1}{2}\log\mu)} - 2^{(\dcrit+1+\frac{1}{2}\log\mu)}\sqrt{\mu} - 24\mu\right]\\
	& \leq 2\klog{v_\alpha}\left[25\cdot 2^{2(\dcrit+1+\frac{1}{2}\log\mu)} - 2^{(\dcrit+1+\frac{1}{2}\log\mu)} - 24\right]\\
	& < 2\klog{v_\alpha}\left[25\cdot 2^{2(\dcrit+\log\mu)} - 2^{(\dcrit+\log\mu)} - 24\right]\hspace{3mm}\textrm{\color{gray}(as $1 < \frac{1}{2}\log\mu$)}\\
	& \leq \first{\dcrit+\log\mu}{\alpha}
\end{align*}

To obtain an upper bound on the number of rounds that have elapsed in $\alpha$'s execution when $\beta$ starts epoch $\ep{\dcrit+2}{\beta}$, we can take the number of rounds that have elapsed in $\beta$'s execution and add to it the delay between start times, i.e., $\first{\dcrit+1}{\beta} + \delay$. Using the above bound and our assumption on the delay between starting times, we get
\begin{align*}
	& \first{\dcrit+1}{\beta} + \delay \\
	<\ & \first{\dcrit+\log\mu}{\alpha} + \first{\dcrit}{\alpha}\\
	<\ & \first{\dcrit+\log\mu}{\alpha}  + \frac{1}{2}|\ep{\dcrit+2}{\alpha}| \hspace{3mm}\textrm{\color{gray}(using \Cref{fact:epochincreasing})}\\
	<\ & \first{\dcrit+\log\mu}{\alpha}  + \frac{1}{2}|\ep{\dcrit+\log\mu + 1}{\alpha}| \hspace{3mm}\textrm{\color{gray}(since $2 < \log\mu+1$)}
\end{align*}
In particular, this means that when $\beta$ starts epoch $\ep{\dcrit+2}{\beta}$, agent $\alpha$'s execution has not yet started the second half of epoch $\ep{\dcrit+{\log\mu}+1}{\alpha}$. By \Cref{fact:finalphasefraction}, the final phase of epoch $\ep{\dcrit+{\log\mu}+1}{\alpha}$ is less than half of the length of the epoch, so it follows that $\alpha$ has not yet started executing the final phase of epoch $\ep{\dcrit+{\log\mu}+1}{\alpha}$. Moreover, when $\beta$ is starting epoch $\ep{\dcrit+2}{\beta}$, agent $\alpha$ is in an epoch $\ep{j}{\alpha}$ with $j \geq \dcrit+2$ since $\alpha$'s phases are shorter and $\alpha$ did not start its execution after $\beta$. By \Cref{lemma:willdoDsearch}, agent $\alpha$ visits all nodes in the $D$-neighbourhood of $v_\alpha$ within the next $44\klog{v_\alpha}2^{\dcrit+{\log\mu}+1}$ rounds.

Finally, note that $\beta$ starts its execution of epoch $\ep{\dcrit+2}{\beta}$ with a waiting period consisting of $36\klog{v_\beta}2^{\dcrit+2} = 72\klog{v_\beta}2^{\dcrit+1} = 72\mu\klog{v_\alpha}2^{\dcrit+1} = 72\klog{v_\alpha}2^{\dcrit+{\log\mu}+1}$ rounds. Thus, $\beta$ is at its starting node $v_\beta$ when $\alpha$ visits $v_\beta$, so rendezvous occurs by the end of epoch $\ep{\dcrit+\log\mu+1}{\alpha}$ in $\alpha$'s execution.
\end{proof}

Next, we consider the case where $\mu$ is strictly larger than $3$, but not larger than $4$. The argument is similar in structure to the previous case: consider the round in which $\beta$ starts epoch $\ep{\dcrit+1}{\beta}$, and consider the number of rounds it takes $\beta$ to reach this round in its execution. We can then figure out how far along $\alpha$ is in its execution at this time, and, roughly speaking, show that $\alpha$ is at most one epoch 'ahead', i.e., it's searching and waiting periods are larger by a factor of at most 2. But the fact that $\log^*(v_\beta) = \mu\cdot\log^*(v_\alpha)$ with $\mu > 3$ means that $\beta$'s waiting period in stage 0 has an additional $\mu>3$ factor built into its length, so this waiting period at the start of epoch $\ep{\dcrit+1}{\beta}$ is long enough to ensure that it includes $\alpha$'s next search, so rendezvous will occur. This idea is formalized in the following result.

\begin{lemma}\label{betabigger:mubigger3}
Suppose that $3 < \mu \leq 4$. Rendezvous occurs by the end of agent $\alpha$'s execution of epoch $\ep{\dcrit+2}{\alpha}$.
\end{lemma}
\begin{proof}
Consider the round in which agent $\beta$ starts epoch $\ep{\dcrit+1}{\beta}$. The number of rounds that occur in $\beta$'s execution up to that round is calculated by
\begin{align*}
	& \first{\dcrit}{\beta} \\
	=\ & 2\klog{v_\beta}\left[25\cdot 2^{2\dcrit} - 2^{\dcrit} - 24\right]\\
	=\ & 2\kappa[\mu\log^*(v_\alpha)]\left[25\cdot 2^{2\dcrit} - 2^{\dcrit} - 24\right]\\
	=\ & 2\kappa[\log^*(v_\alpha)+(\mu-1)\log^*(v_\alpha)]\left[25\cdot 2^{2\dcrit} - 2^{\dcrit} - 24\right]\\
	=\ & 2\klog{v_\alpha}\left[25\cdot 2^{2\dcrit} - 2^{\dcrit} - 24\right] +2(\mu-1)\klog{v_\alpha}\left[25\cdot 2^{2\dcrit} - 2^{\dcrit} - 24\right]\\
	=\ & \first{\dcrit}{\alpha} + 2(\mu-1)\klog{v_\alpha}\left[25\cdot 2^{2\dcrit} - 2^{\dcrit} - 24\right]
\end{align*}

To obtain an upper bound on the number of rounds that have elapsed in $\alpha$'s execution when $\beta$ starts epoch $\ep{\dcrit+1}{\beta}$, we can take the number of rounds that have elapsed in $\beta$'s execution and add to it the delay between start times, i.e., $\first{\dcrit}{\beta}+ \delay$. Using the above equality and our assumption on the delay between starting times, we get

\begin{align*}
	& \first{\dcrit}{\beta}+\delay\\
	<\ & \first{\dcrit}{\beta}+\first{\dcrit}{\alpha}\\
	=\ & \first{\dcrit}{\beta} + 2\klog{v_\alpha}\left[25\cdot 2^{2\dcrit} - 2^{\dcrit} - 24\right]\\
	=\ & \first{\dcrit}{\alpha} + 2(\mu-1)\klog{v_\alpha}\left[25\cdot 2^{2\dcrit} - 2^{\dcrit} - 24\right]\\
	& +2\klog{v_\alpha}\left[25\cdot 2^{2\dcrit} - 2^{\dcrit} - 24\right]\\
	=\ & \first{\dcrit}{\alpha} +2\mu\klog{v_\alpha}\left[25\cdot 2^{2\dcrit} - 2^{\dcrit} - 24\right]\\
	\leq\ & \first{\dcrit}{\alpha} + 2(4)\klog{v_\alpha}\left[25\cdot 2^{2\dcrit} - 2^{\dcrit} - 24\right]\\
	=\ & \first{\dcrit}{\alpha} + 2\klog{v_\alpha}\left[50\cdot 2^{2\dcrit+1} - 2^{\dcrit+2} - 96\right]\\
	<\ & \first{\dcrit}{\alpha} + 2\klog{v_\alpha}\left[50\cdot 2^{2\dcrit+1} - 2^{\dcrit+1} - 96\right]\\
	<\ & \first{\dcrit}{\alpha} + 2|\ep{\dcrit+1}{\alpha}|
	\hspace{3mm}\textrm{\color{gray}(using \Cref{fact:epochlength} with $j=\dcrit+1$)}\\
	<\ & \first{\dcrit}{\alpha} + |\ep{\dcrit+1}{\alpha}| + \frac{1}{2}|\ep{\dcrit+2}{\alpha}|
	\hspace{3mm}\textrm{\color{gray}(using \Cref{fact:epochincreasing} with $j=\dcrit+1$)}\\
\end{align*}
We conclude that $\alpha$'s execution has not yet started the second half of epoch $\ep{\dcrit+2}{\alpha}$. By \Cref{fact:finalphasefraction}, the final phase of epoch $\ep{\dcrit+2}{\alpha}$ is less than half of the length of the epoch, so it follows that $\alpha$ has not yet started executing the final phase of epoch $\ep{\dcrit+2}{\alpha}$. Moreover, when $\beta$ is starting epoch $\ep{\dcrit+1}{\beta}$, agent $\alpha$ is in an epoch $\ep{j}{\alpha}$ with $j \geq \dcrit+1$ since $\alpha$'s phases are shorter and $\alpha$ did not start its execution after $\beta$. By \Cref{lemma:willdoDsearch}, agent $\alpha$ visits all nodes in the $D$-neighbourhood of $v_\alpha$ within the next $44\klog{v_\alpha}2^{\dcrit+2} = 88\klog{v_\alpha}2^{\dcrit+1}$ rounds.

Finally, note that $\beta$ starts its execution of epoch $\ep{\dcrit+1}{\beta}$ by execution stage 0 of the first phase of the epoch, i.e., a waiting period consisting of $36\klog{v_\beta}2^{\dcrit+1} = 36\mu\klog{v_\alpha}2^{\dcrit+1} > 36(3)\klog{v_\alpha}2^{\dcrit+1} = 108\klog{v_\alpha}2^{\dcrit+1}$ rounds. Therefore, $\beta$ will be waiting at its starting node $v_\beta$ when $\alpha$ visits $v_\beta$, so rendezvous will occur by the end of epoch $\ep{\dcrit+2}{\alpha}$ in $\alpha$'s execution, as desired.
\end{proof}

Next, we consider the remaining case, i.e., where $\mu$ is at most $3$. At a high level, we consider the round in which agent $\beta$ starts epoch $\ep{\dcrit+1}{\beta}$, and prove that $\alpha$'s execution has not yet reached the end of epoch $\ep{\dcrit+1}{\alpha}$. In the case where $\mu \geq 11/9$, agent $\beta$'s waiting period is long enough to contain a search by $\alpha$, so rendezvous occurs with $\alpha$ as `searcher' and $\beta$ as `sitter'. However, when $\mu < 11/9$, this is not true. So, in this case, we instead look at when $\alpha$ starts epoch $\ep{\dcrit+2}{\alpha}$, and notice that its waiting period at the start of this epoch is long enough to include the next search by $\beta$, so rendezvous occurs with $\beta$ as `searcher' and $\alpha$ as `sitter'.  This idea is formalized in the following result.

\begin{lemma}\label{betabigger:musmaller3}
Suppose that $1 < \mu \leq 3$. Rendezvous occurs by the end of agent $\alpha$'s execution of epoch $\ep{\dcrit+2}{\alpha}$.
\end{lemma}
\begin{proof}
Consider the round in which agent $\beta$ starts epoch $\ep{\dcrit+1}{\beta}$. The number of rounds that have elapsed in $\beta$'s execution is 
\begin{align*}
	& \first{\dcrit}{\beta}\\
	& = 2\klog{v_\beta}\left[25\cdot 2^{2\dcrit} - 2^{\dcrit} - 24\right]\\
	& = 2\kappa[\mu\log^*(v_\alpha)]\left[25\cdot 2^{2\dcrit} - 2^{\dcrit} - 24\right]\\
	& = 2\kappa[\log^*(v_\alpha)+(\mu-1)\log^*(v_\alpha)]\left[25\cdot 2^{2\dcrit} - 2^{\dcrit} - 24\right]\\
	& = 2\klog{v_\alpha}\left[25\cdot 2^{2\dcrit} - 2^{\dcrit} - 24\right]+2(\mu-1)\klog{v_\alpha}\left[25\cdot 2^{2\dcrit} - 2^{\dcrit} - 24\right]\\
	& = \first{\dcrit}{\alpha} + 2(\mu-1)\klog{v_\alpha}\left[25\cdot 2^{2\dcrit} - 2^{\dcrit} - 24\right]
\end{align*}

To obtain an upper bound on the number of rounds that have elapsed in $\alpha$'s execution when $\beta$ starts epoch $\ep{\dcrit+1}{\beta}$, we can take the number of rounds that have elapsed in $\beta$'s execution and add to it the delay between start times, i.e., $\first{\dcrit}{\beta}+ \delay$. Using the above equality and our assumption on the delay between starting times, we get

\begin{align*}
	& \first{\dcrit}{\beta}+\delay\\
	<\ & \first{\dcrit}{\beta}+\first{\dcrit}{\alpha}\\
	=\ & \first{\dcrit}{\beta} + 2\klog{v_\alpha}\left[25\cdot 2^{2\dcrit} - 2^{\dcrit} - 24\right]\\
	=\ & \first{\dcrit}{\alpha} + 2(\mu-1)\klog{v_\alpha}\left[25\cdot 2^{2\dcrit} - 2^{\dcrit} - 24\right]\\ 
	& +2\klog{v_\alpha}\left[25\cdot 2^{2\dcrit} - 2^{\dcrit} - 24\right]\\
	=\ & \first{\dcrit}{\alpha} +2\mu\klog{v_\alpha}\left[25\cdot 2^{2\dcrit} - 2^{\dcrit} - 24\right]\\
	\leq\ & \first{\dcrit}{\alpha} + 2(3)\klog{v_\alpha}\left[25\cdot 2^{2\dcrit} - 2^{\dcrit} - 24\right]\\
	=\ & \first{\dcrit}{\alpha} + \klog{v_\alpha}\left[75\cdot 2^{2\dcrit+1} - 3\cdot 2^{\dcrit+1} - 144\right]\\
	=\ & \first{\dcrit}{\alpha} + \klog{v_\alpha}\left[75\cdot 2^{2\dcrit+1} - 2^{\dcrit+1} - 144\right] - 2\cdot2^{\dcrit+1}\klog{v_\alpha}\\
	\leq\ & \first{\dcrit}{\alpha} + \klog{v_\alpha}\left[75\cdot 2^{2\dcrit+1} - 2^{\dcrit+1} - 144\right] - 4\cdot2^{\dcrit+1}\\
	&\textrm{\color{gray}(since $\kappa\geq 2$ and $\log^*(v_\alpha) \geq 1$)}\\
	<\ & \first{\dcrit}{\alpha} + |\ep{\dcrit+1}{\alpha}| - 4\cdot2^{\dcrit+1}
	\hspace{3mm}\textrm{\color{gray}(using \Cref{fact:epochlength} with $j=\dcrit+1$)}\\
\end{align*}

In every phase of $\ep{\dcrit+1}{\alpha}$ in $\alpha$'s execution, each block in stage 2 has length exactly $4\cdot2^{\dcrit+1}$. So, the above inequality proves that, when $\beta$ starts its execution of epoch $\ep{\dcrit+1}{\beta}$, agent $\alpha$'s execution has not yet reached the final block of the final phase of epoch $\ep{\dcrit+1}{\alpha}$. Moreover, when $\beta$ is starting epoch $\ep{\dcrit+1}{\beta}$, agent $\alpha$ is in an epoch $\ep{j}{\alpha}$ with $j \geq \dcrit+1$ since $\alpha$'s phases are shorter and $\alpha$ did not start its execution after $\beta$.

The remainder of the proof considers two subcases based on the value of $\mu$.

\begin{itemize}
	\item Case (a): $\mu \geq 11/9$\\
	By \Cref{lemma:willdoDsearch}, agent $\alpha$ visits all nodes in the $D$-neighbourhood of $v_\alpha$ within the next $44\klog{v_\alpha}2^{\dcrit+1}$ rounds. However, $\beta$ starts its execution of epoch $\ep{\dcrit+1}{\beta}$ with a waiting period whose length is $36\klog{v_\beta}2^{\dcrit+1} = 36\mu\klog{v_\alpha}2^{\dcrit+1} \geq 36(11/9)\klog{v_\alpha}2^{\dcrit+1} = 44\klog{v_\alpha}2^{\dcrit+1}$.\\Therefore, $\beta$ will be waiting at its starting node $v_\beta$ when $\alpha$ visits $v_\beta$, so rendezvous will occur by the end of epoch $\ep{\dcrit+2}{\alpha}$ in $\alpha$'s execution, as desired.
	\item Case (b): $1 < \mu < 11/9$\\
	Consider the round in which $\alpha$ starts epoch $\ep{\dcrit+2}{\alpha}$. First, note that agent $\beta$'s execution has passed the start of epoch $\ep{\dcrit+1}{\beta}$ at this time: indeed, we showed above that, when agent $\beta$ started epoch $\ep{\dcrit+1}{\beta}$, agent $\alpha$'s execution had not yet reached the end of epoch $\ep{\dcrit+1}{\alpha}$.
	
	Next, we set out to show that, at this time, agent $\beta$'s execution has not yet reached the final block of the final phase of epoch $\ep{\dcrit+1}{\beta}$. The number of rounds that elapse in agent $\alpha$'s execution up until it starts epoch $\ep{\dcrit+2}{\alpha}$ is
	\begin{align*}
		\first{\dcrit+1}{\alpha} =\ &    2\klog{v_\alpha}\left[25\cdot 2^{2\dcrit+2} - 2^{\dcrit+1} - 24\right]\\
		\leq\ &   2\kappa(\log^*(v_\beta)-1)\left[25\cdot 2^{2\dcrit+2} - 2^{\dcrit+1} - 24\right]\\
		=\ &  2\klog{v_\beta}\left[25\cdot 2^{2\dcrit+2} - 2^{\dcrit+1} - 24\right]\\
		& - 2\kappa\left[25\cdot 2^{2\dcrit+2} - 2^{\dcrit+1} - 24\right]\\
		=\ &  \first{\dcrit+1}{\beta} - 2\kappa\left[25\cdot 2^{2\dcrit+2} - 2^{\dcrit+1} - 24\right]\\
		=\ &  \first{\dcrit+1}{\beta} - 2\kappa\left[50\cdot 2^{2\dcrit+1} - 2^{\dcrit+1} - 24\right]\\
		\leq\ &  \first{\dcrit+1}{\beta} - 2\kappa\left[50\cdot 2^{2\dcrit+1} - 2^{2\dcrit+1} - 24\cdot 2^{2\dcrit+1}\right]\\
		=\ &  \first{\dcrit+1}{\beta} - 2\kappa\left[25\cdot 2^{2\dcrit+1}\right]\\
		<\ &  \first{\dcrit+1}{\beta} - 4\cdot 2^{\dcrit+1}\\
	\end{align*}
	
	Since $\alpha$'s execution starts before or at the same time as agent $\beta$'s, it follows that $\first{\dcrit+1}{\alpha}$ is an upper bound on the number of rounds that have elapsed in agent $\beta$'s execution. This means that $\first{\dcrit+1}{\beta} - 4\cdot 2^{\dcrit+1}$ is an upper bound on the number of rounds that have elapsed in agent $\beta$'s execution at the moment when $\alpha$ starts epoch $\ep{\dcrit+2}{\alpha}$. Since each block in every phase of epoch $\ep{\dcrit+1}{\beta}$ has length exactly $4\cdot 2^{\dcrit+1}$, the bound $\first{\dcrit+1}{\beta} - 4\cdot 2^{\dcrit+1}$ tells us that $\beta$'s execution has not yet reached the final block of the final phase of epoch $\ep{\dcrit+1}{\beta}$. 
	
	From the above, we have proved that agent $\beta$'s execution has passed the start of epoch $\ep{\dcrit+1}{\beta}$, but has not yet reached the final block of the final phase of epoch $\ep{\dcrit+1}{\beta}$. Thus, by \Cref{lemma:willdoDsearch}, agent $\beta$ visits all nodes in the $D$-neighbourhood of $v_\beta$ within $44\klog{v_\beta}2^{\dcrit+1} = 44\mu\klog{v_\alpha}2^{\dcrit+1}  < 44(11/9)\klog{v_\alpha}2^{\dcrit+1} < 55\klog{v_\alpha}2^{\dcrit+1}$ rounds. However, agent $\alpha$ starts its execution of epoch $\ep{\dcrit+2}{\alpha}$ with a waiting period consisting of $36\klog{v_\alpha}2^{\dcrit+2} = 72\klog{v_\alpha}2^{\dcrit+1}$ rounds. Therefore, $\alpha$ will be waiting at its starting node $v_\alpha$ when $\beta$ visits all nodes in the $D$-neighbourhood of $v_\beta$, so rendezvous will occur by the end of agent $\alpha$'s execution of epoch $\ep{\dcrit+2}{\alpha}$.
\end{itemize}

In both cases, rendezvous occurs before the end of epoch $\ep{\dcrit+2}{\alpha}$ in $\alpha$'s execution, as desired.
\end{proof}

Finally, to cover all possible values of $\mu$, we combine \Cref{betabigger:largemu,betabigger:mubigger4,betabigger:mubigger3,betabigger:musmaller3} to conclude the analysis in the case where $\log^*(v_\beta) > \log^*(v_\alpha)$.

\begin{lemma}\label{logstarbetabigger}
Rendezvous will occur within $O(\min\{D^2(\log^*(v_\beta))^2,D^4\log^*(v_\beta)\})$ rounds in agent $\alpha$'s execution.
\end{lemma}
\begin{proof}
When $4 < \mu \leq 12D^2$, \Cref{betabigger:mubigger4} proves that rendezvous occurs by the end of agent $\alpha$'s execution of epoch $\ep{\dcrit+\log\mu+1}{\alpha}$. By \Cref{cor:runningtime}, the total number of rounds that elapse in $\alpha$'s execution of the algorithm is $O(2^{2\dcrit+2\log\mu+2}\log^*(v_\alpha)) = O(D^2\mu^2\log^*(v_\alpha)) = O(D^2\mu\log^*(v_\beta))$. Since $\mu=\frac{log^*(v_\beta)}{\log^*(v_\alpha)}$ and $\mu \leq 12D^2$, we conclude that the order $O(D^2\mu\log^*(v_\beta))$ is included in $O(\min\{D^2(\log^*(v_\beta))^2,D^4\log^*(v_\beta)\})$.

For all other values of $\mu$, \Cref{betabigger:largemu,betabigger:mubigger3,betabigger:musmaller3} imply that rendezvous occurs by the end of agent $\alpha$'s execution of epoch $\ep{\dcrit+2}{\alpha}$. By \Cref{cor:runningtime}, the number of rounds that elapse in $\alpha$'s execution of the algorithm is $O(2^{2\dcrit+4}\log^*(v_\alpha)) = O(D^2\log^*(v_\alpha)) \subseteq O(D^2\log^*(v_\beta)) \subseteq O(\min\{D^2(\log^*(v_\beta))^2,D^4\log^*(v_\beta)\})$.
\end{proof}

\subsubsection*{The case where $\log^*(v_\alpha) > \log^*(v_\beta)$}

Define $\mu > 1$ such that $\mu\cdot\log^{*}{v_\beta} = \log^*{v_\alpha}$. The overall structure of the following analysis is similar to the analysis in the previous section, but is not symmetric in the details.

The most significant difference is the fact that, although $\alpha$ starts its execution before or at the same time as $\beta$ does, agent $\alpha$'s execution moves through the phases/epochs more `slowly' since the phase/epoch lengths are larger. So, there will be a point in time in the execution where $\beta$'s execution `overtakes' $\alpha$'s execution with respect to the index of the epoch they are executing. After this point occurs, we are in a similar situation as the previous section (i.e., one agent is simultaneously `ahead' and has shorter phases), and we can use similar arguments to analyze the algorithm. The next result, which we informally refer to as the Crossover Lemma, gives a bound on how long it takes for $\beta$'s execution to `overtake' $\alpha$'s execution.

\begin{lemma}\label{lem:crossover}
When agent $\alpha$'s execution starts epoch $\ep{\dcrit+\log\log^*(v_\alpha)+1}{\alpha}$, agent $\beta$'s execution has already reached the start of epoch $\ep{\dcrit+\log\log^*(v_\alpha)+1}{\beta}$. In particular, this implies that $\beta$'s execution has already reached the start of epoch $\ep{\dcrit+1}{\beta}$.
\end{lemma}
\begin{proof}
Consider the start of epoch $\ep{\dcrit+\log\log^*(v_\alpha)+1}{\alpha}$ in $\alpha$'s execution. The number of rounds that have elapsed in $\alpha$'s execution is
\begin{align*}
	&\first{\dcrit+\log\log^*(v_\alpha)}{\alpha} \\
	=\ & 2\klog{v_\alpha}\left[25\cdot 2^{2(\dcrit+\log\log^*(v_\alpha))} - 2^{(\dcrit+\log\log^*(v_\alpha))} - 24\right]\\
	\geq\ &  2\kappa (\log^*(v_\beta)+1)\left[25\cdot 2^{2(\dcrit+\log\log^*(v_\alpha))} - 2^{(\dcrit+\log\log^*(v_\alpha))} - 24\right]\\
	=\ & \first{\dcrit+\log\log^*(v_\alpha)}{\beta} + 2\kappa\left[25\cdot 2^{2(\dcrit+\log\log^*(v_\alpha))} - 2^{(\dcrit+\log\log^*(v_\alpha))} - 24\right]\\
	\geq\ & \first{\dcrit+\log\log^*(v_\alpha)}{\beta}\\
	& + 2\kappa\left[25\cdot 2^{2\dcrit+\log\log^*(v_\alpha)} - 2^{(\dcrit+\log\log^*(v_\alpha))} - 24\cdot 2^{\log\log^*(v_\alpha)}\right]\\
	=\ & \first{\dcrit+\log\log^*(v_\alpha)}{\beta} + 2\klog{v_\alpha}\left[25\cdot 2^{2\dcrit} - 2^{\dcrit} - 24\right]\\
	=\ & \first{\dcrit+\log\log^*(v_\alpha)}{\beta} + \first{\dcrit}{\alpha}\\
\end{align*}

To calculate the number of rounds that have elapsed in $\beta$'s execution, we take the number of rounds that have elapsed in $\alpha$'s execution and subtract the delay between the starting times. However, by assumption, the delay is bounded above by $\first{\dcrit}{\alpha}$. Combining this fact with the above inequality, we get that the number of rounds that have elapsed in $\beta$'s execution is
\begin{align*}
	&	\first{\dcrit+\log\log^*(v_\alpha)}{\alpha} - \delay\\
	\geq\ & [\first{\dcrit+\log\log^*(v_\alpha)}{\beta} + \first{\dcrit}{\alpha}] - [\first{\dcrit}{\alpha}]\\
	=\ & \first{\dcrit+\log\log^*(v_\alpha)}{\beta}
\end{align*}
In particular, we have shown that $\beta$'s execution has already completed its first $\dcrit + \log\log^*(v_\alpha)$ phases, which implies that $\beta$'s execution has already reached the start of epoch $\ep{\dcrit + \log\log^*(v_\alpha)+1}{\beta}$.
\end{proof}

We now proceed to prove that rendezvous occurs, and consider various cases based on the value of $\mu$. We first consider the case where $\mu$ is very large, i.e., the value of $\log^*{v_\alpha}$ is much larger than the value of $\log^*{v_\beta}$. We assume that $\mu > 4D$, and consider two subcases depending on which agent reaches its $(\dcrit+1)$'th epoch first.

In the first subcase, we assume that $\alpha$'s execution reaches the start of epoch $\ep{\dcrit+1}{\alpha}$ at the same time or before agent $\beta$'s execution reaches the start of epoch $\ep{\dcrit+1}{\beta}$. We prove that $\alpha$'s waiting period in stage 0 at the start of epoch $\ep{\dcrit+1}{\alpha}$ is very large, i.e., large enough for agent $\beta$ to execute the first $\dcrit+2$ epochs in its execution. Therefore, while $\alpha$ is at its starting node in this waiting period, agent $\beta$ will execute its entire epoch $\ep{\dcrit+2}{\beta}$, and this includes a visit to all nodes in the $D$-neighbourhood of $v_\beta$, which guarantees rendezvous.

\begin{lemma}\label{alphabigger:largemu1}
Suppose that $\mu > 4D$, and suppose that $\alpha$'s execution reaches the start of epoch $\ep{\dcrit+1}{\alpha}$ at the same time or before agent $\beta$'s execution reaches the start of epoch $\ep{\dcrit+1}{\beta}$. Then, rendezvous occurs by the end of agent $\alpha$'s execution of epoch $\ep{\dcrit+1}{\alpha}$.
\end{lemma}
\begin{proof}
We will use the fact that $\mu > 4D = 4\cdot 2^{\log_2 D} \geq 4\cdot 2^{\floor*{\log_2 D}} = 2^{2+\floor*{\log_2 D}} = 2^{\dcrit+1}$.

Consider the round $\tau$ in which $\alpha$'s execution is at the start of epoch $\ep{\dcrit+1}{\alpha}$. From round $\tau$, the number of rounds until $\beta$ reaches the start of epoch $\ep{\dcrit+2}{\beta}$ is bounded above by $\first{\dcrit+1}{\beta} = 2\klog{v_\beta}[25\cdot2^{2\dcrit+2} - 2^{\dcrit+1} - 24] \leq 25\klog{v_\beta} 2^{2\dcrit+2}$. This includes an execution of the entire epoch $\ep{\dcrit+1}{\beta}$ (since $\beta$ has not yet started epoch $\ep{\dcrit+1}{\beta}$ at round $\tau$), so $\beta$ searches its $D$-neighbourhood at least once during these rounds. 

Agent $\alpha$ starts its execution of $\ep{\dcrit+1}{\alpha}$ with a waiting period of length $36\klog{v_\alpha}2^{\dcrit+1} = 36\mu\klog{v_\beta}2^{\dcrit+1} > 36(2^{\dcrit+1})\klog{v_\beta}2^{\dcrit+1} = 36\klog{v_\beta}2^{2\dcrit+2}$ rounds.

Since  $36\klog{v_\beta}2^{2\dcrit+2} > 25\klog{v_\beta}2^{2\dcrit+2}$, it follows that while $\beta$ visits all nodes in the $D$-neighbourhood of $v_\beta$, agent $\alpha$ will be located at $v_\alpha$ during stage 0 of the first phase of epoch $\ep{\dcrit+1}{\alpha}$, so rendezvous will occur by the end of epoch $\ep{\dcrit+1}{\alpha}$ in $\alpha$'s execution.
\end{proof}

In the second subcase, we assume that $\beta$'s execution reaches the start of epoch $\ep{\dcrit+1}{\beta}$ before agent $\alpha$'s execution reaches the start of epoch $\ep{\dcrit+1}{\alpha}$. We consider the round in which agent $\alpha$ starts epoch $\ep{\dcrit+2}{\alpha}$, and prove that agent $\beta$'s execution is in an epoch $\ep{j}{\beta}$ with $\dcrit+1 \leq j < \dcrit+\log\mu+1$. This means that $\beta$ has at least passed the critical epoch (so its searches will include the entire $D$-neighbourhood of $v_\beta$), but it also means that $\beta$ is not `too far ahead' in its execution, i.e., the lengths of its phases are at most a $\mu/2$ factor larger than $\alpha$'s. However, the fact that $\log^*(v_\alpha)$ is a $\mu$ factor larger than $\log^*(v_\beta)$ makes up for this difference, i.e., $\alpha$'s waiting period in stage 0 at the start of epoch $\ep{\dcrit+2}{\alpha}$ is large enough to include a search by $\beta$, which guarantees rendezvous.

\begin{lemma}\label{alphabigger:largemu2}
Suppose that $\mu > 4D$, and suppose that $\beta$'s execution reaches the start of epoch $\ep{\dcrit+1}{\beta}$ before agent $\alpha$'s execution reaches the start of epoch $\ep{\dcrit+1}{\alpha}$. Then, rendezvous occurs by the end of agent $\alpha$'s execution of epoch $\ep{\dcrit+2}{\alpha}$.
\end{lemma}
\begin{proof}
We will use the fact that $1 < \frac{1}{2}\log\mu$, which is true since $\mu > 4D \geq 4$.

Consider the round in which agent $\alpha$ starts epoch $\ep{\dcrit+2}{\alpha}$. The number of rounds that occur in $\alpha$'s execution up to that round is calculated by
\begin{align*}
	\first{\dcrit+1}{\alpha} & = 2\klog{v_\alpha}\left[25\cdot 2^{2(\dcrit+1)} - 2^{(\dcrit+1)} - 24\right]\\
	& = 2\mu\klog{v_\beta}\left[25\cdot 2^{2(\dcrit+1)} - 2^{(\dcrit+1)} - 24\right]\\
	& = 2\klog{v_\beta}\left[25\cdot 2^{2(\dcrit+1+\frac{1}{2}\log\mu)} - 2^{(\dcrit+1+\frac{1}{2}\log\mu)}\sqrt{\mu} - 24\mu\right]\\
	& \leq 2\klog{v_\beta}\left[25\cdot 2^{2(\dcrit+1+\frac{1}{2}\log\mu)} - 2^{(\dcrit+1+\frac{1}{2}\log\mu)} - 24\right]\\
	& < 2\klog{v_\beta}\left[25\cdot 2^{2(\dcrit+\log\mu)} - 2^{(\dcrit+\log\mu)} - 24\right]\hspace{3mm}\textrm{\color{gray}(as $1 < \frac{1}{2}\log\mu$)}\\
	& \leq \first{\dcrit+\log\mu}{\beta}
\end{align*}

However, the number of rounds that have elapsed in $\alpha$'s execution is an upper bound on the number of rounds that have elapsed in $\beta$'s execution. In particular, this means that at most $\first{\dcrit+\log\mu}{\beta}$ rounds have elapsed in $\beta$'s execution, so it has not yet started epoch $\ep{\dcrit+\log\mu+1}{\beta}$ when $\alpha$ starts epoch $\ep{\dcrit+2}{\alpha}$. Namely, $\beta$ is executing some epoch $\ep{j'}{\beta}$ with $j' < \dcrit+\log\mu+1$. Moreover, we can conclude that $j' \geq \dcrit+1$ since we assumed that $\beta$ started epoch $\ep{\dcrit+1}{\beta}$ before $\alpha$ started epoch $\ep{\dcrit+1}{\alpha}$. By \Cref{lemma:willdoDsearch}, agent $\beta$ visits all nodes in the $D$-neighbourhood of $v_\beta$ within the next $44\klog{v_\beta}2^{\dcrit+{\log\mu}+1}$ rounds.

Finally, note that $\alpha$ starts its execution of epoch $\ep{\dcrit+2}{\alpha}$ with a waiting period consisting of $36\klog{v_\alpha}2^{\dcrit+2} = 72\klog{v_\alpha}2^{\dcrit+1} = 72\mu\klog{v_\beta}2^{\dcrit+1} = 72\klog{v_\beta}2^{\dcrit+{\log\mu}+1}$ rounds. Therefore, while $\beta$ visits the $D$-neighbourhood of $v_\beta$, agent $\alpha$ will be located at $v_\alpha$ during stage 0 of the first phase of epoch $\ep{\dcrit+2}{\alpha}$, so rendezvous will occur by the end of epoch $\ep{\dcrit+2}{\alpha}$ in $\alpha$'s execution.
\end{proof}

Next, we consider the case where $\mu$ is strictly larger than 4, but not larger than $8D$. We consider the round in which agent $\alpha$ starts epoch $\ep{\dcrit+\log\log^*(v_\alpha)+2}{\alpha}$. First, using the Crossover Lemma, we can prove that $\beta$'s execution has at least passed the critical epoch $\ep{\dcrit}{\beta}$, so all future searches will include the entire $D$-neighbourhood of $v_\beta$. We also prove that agent $\beta$'s execution is fewer than $\log\mu$ phases `ahead' of $\alpha$'s execution, so the lengths of its phases are at most a $\mu/2$ factor larger than $\alpha$'s. However, the fact that $\log^*(v_\alpha)$ is a $\mu$ factor larger than $\log^*(v_\beta)$ makes up for this difference, i.e., $\alpha$'s waiting period in stage 0 at the start of epoch $\ep{\dcrit+\log\log^*(v_\alpha)+2}{\alpha}$ is large enough to include a search by $\beta$, which guarantees rendezvous.

\begin{lemma}\label{alphabigger:mubigger4}
Suppose that $4 < \mu \leq 4D$. Rendezvous occurs by the end of agent $\alpha$'s execution of epoch $\ep{\dcrit+\log\log^*(v_\alpha)+2}{\alpha}$.
\end{lemma}
\begin{proof}
Let $M=\dcrit+\log\log^*(v_\alpha)$. Consider the round in which agent $\alpha$ starts epoch $\ep{M+2}{\alpha}$. The number of rounds that occur in $\alpha$'s execution up to that round is calculated by
\begin{align*}
	\first{M+1}{\alpha} & = 2\klog{v_\alpha}\left[25\cdot 2^{2(M+1)} - 2^{(M+1)} - 24\right]\\
	& = 2\mu\klog{v_\beta}\left[25\cdot 2^{2(M+1)} - 2^{(M+1)} - 24\right]\\\\
	& = 2\klog{v_\beta}\left[25\cdot 2^{2(M+1+\frac{1}{2}\log\mu)} - 2^{(M+1+\frac{1}{2}\log\mu)}\sqrt{\mu} - 24\mu\right]\\
	& \leq 2\klog{v_\beta}\left[25\cdot 2^{2(M+1+\frac{1}{2}\log\mu)} - 2^{(M+1+\frac{1}{2}\log\mu)} - 24\right]\\
	& < 2\klog{v_\beta}\left[25\cdot 2^{2(M+\log\mu)} - 2^{(M+\log\mu)} - 24\right]\hspace{3mm}\textrm{\color{gray}(since $1 < \frac{1}{2}\log\mu$)}\\
	& \leq \first{M+\log\mu}{\beta}
\end{align*}

However, the number of rounds that have elapsed in $\alpha$'s execution is an upper bound on the number of rounds that have elapsed in $\beta$'s execution. In particular, this means that at most $\first{M+\log\mu}{\beta}$ rounds have elapsed in $\beta$'s execution, so it has not yet started epoch $\ep{M+\log\mu+1}{\beta}$ when $\alpha$ starts epoch $\ep{M+2}{\alpha}$. Namely, $\beta$ is executing some epoch $\ep{j}{\beta}$ with $j < M+\log\mu+1$. Moreover, by \Cref{lem:crossover}, when agent $\alpha$ started epoch $\ep{M+1}{\alpha}$, agent $\beta$ had already started epoch $\ep{\dcrit+1}{\beta}$, so we can conclude that $j \geq \dcrit+1$. By \Cref{lemma:willdoDsearch}, agent $\beta$ visits all nodes in the $D$-neighbourhood of $v_\beta$ within the next $44\klog{v_\beta}2^{M+{\log\mu}+1}$ rounds.

Finally, note that $\alpha$ starts its execution of epoch $\ep{M+2}{\alpha}$ with a waiting period consisting of $36\klog{v_\alpha}2^{M+2} = 72\klog{v_\alpha}2^{M+1} = 72\mu\klog{v_\beta}2^{M+1} = 72\klog{v_\beta}2^{M+{\log\mu}+1}$ rounds. Therefore, while $\beta$ visits the $D$-neighbourhood of $v_\beta$, agent $\alpha$ will be located at $v_\alpha$ during stage 0 of the first phase of epoch $\ep{M+2}{\alpha}$, so rendezvous will occur by the end of epoch $\ep{M+2}{\alpha} = \ep{\dcrit+\log\log^*(v_\alpha)+2}{\alpha}$ in $\alpha$'s execution.
\end{proof}

Next, we consider the case where $\mu$ is strictly larger than 3, but not larger than $4$. We consider the round in which agent $\alpha$ starts epoch $\ep{\dcrit+\log\log^*(v_\alpha)+1}{\alpha}$. First, using the Crossover Lemma, we can prove that $\beta$'s execution has at least passed the critical epoch $\ep{\dcrit}{\beta}$, so all future searches will include the entire $D$-neighbourhood of $v_\beta$. We also prove that agent $\beta$'s execution is fewer than $2$ phases `ahead' of $\alpha$'s execution, so the lengths of its phases are at most a factor of 2 larger than $\alpha$'s. However, the fact that $\log^*(v_\alpha)$ is at least a factor of 3 larger than $\log^*(v_\beta)$ makes up for this difference, i.e., $\alpha$'s waiting period in stage 0 at the start of epoch $\ep{\dcrit+\log\log^*(v_\alpha)+1}{\alpha}$ is large enough to include a search by $\beta$, which guarantees rendezvous.

\begin{lemma}\label{alphabigger:mubigger3}
Suppose that $3 < \mu \leq 4$. Rendezvous occurs by the end of agent $\alpha$'s execution of epoch $\ep{\dcrit+\log\log^*(v_\alpha)+1}{\alpha}$.
\end{lemma}
\begin{proof}
Let $M=\dcrit+\log\log^*(v_\alpha)$, and consider the round in which agent $\alpha$ starts epoch $\ep{M+1}{\alpha}$. The number of rounds that have elapsed in $\alpha$'s execution is 
\begin{align*}
	& \first{M}{\alpha}\\
	 =\ & 2\klog{v_\alpha}\left[25\cdot 2^{2M} - 2^{M} - 24\right]\\
	=\ & 2\kappa\left[\mu\log^*(v_\beta)\right]\left[25\cdot 2^{2M} - 2^{M} - 24\right]\\
	=\ & 2\kappa\left[\log^*(v_\beta) + (\mu-1)\log^*(v_\beta)\right]\left[25\cdot 2^{2M} - 2^{M} - 24\right]\\
	=\ & 2\kappa\log^*(v_\beta)\left[25\cdot 2^{2M} - 2^{M} - 24\right]
	+ 2\kappa(\mu-1)\log^*(v_\beta)\left[25\cdot 2^{2M} - 2^{M} - 24\right]\\
	=\ &\first{M}{\beta} + 2\kappa(\mu-1)\log^*(v_\beta)\left[25\cdot 2^{2M} - 2^{M} - 24\right]\\
	<\ & \first{M}{\beta}+ 2\kappa(4)\log^*(v_\beta)\left[25\cdot 2^{2M} - 2^{M} - 24\right]\\
	=\ & \first{M}{\beta}+ 2\kappa\log^*(v_\beta)\left[50\cdot 2^{2M+1} - 2^{M+2} - 96\right]\\
	<\ & \first{M}{\beta}+ 2\kappa\log^*(v_\beta)\left[50\cdot 2^{2M+1} - 2^{M}\right]\\
	=\ & \first{M}{\beta}+ 2\kappa\log^*(v_\beta)\left[50\cdot 2^{2(M+1)-1} - 2^{(M+1)-1}\right]\\
	=\ & \first{M}{\beta}+ 2\left[|\ep{M+1}{\beta}| - 25\kappa\log^*(v_\beta)2^{2M+1}\right]\hspace{3mm}\textrm{\color{gray}(using \Cref{fact:epochlength} with $j=M+1$)}\\
	=\ & \first{M}{\beta}+ |\ep{M+1}{\beta}| + |\ep{M+1}{\beta}| - 25\kappa\log^*(v_\beta)2^{2M+2}\\
	<\ &\first{M}{\beta}+ |\ep{M+1}{\beta}| + |\ep{M+2}{\beta}| - 4\cdot2^{M+2}\\
\end{align*}

First, by \Cref{lem:crossover}, note that agent $\beta$'s execution is currently in an epoch $\ep{M'}{\beta}$ with $M' \geq \dcrit+1$. Next, since $\beta$ starts at the same time as $\alpha$ or later, the number of rounds that have elapsed in $\alpha$'s execution gives an upper bound on the number of rounds that have elapsed in $\beta$'s execution. In particular, using the above bound, agent $\beta$ has executed strictly fewer than $\first{M}{\beta}+ |\ep{M+1}{\beta}| + |\ep{M+2}{\beta}| - 4\cdot2^{M+2}$ rounds. But this means that $\beta$'s execution has not yet reached the last block of stage 2 of the final phase of epoch $\ep{M+2}{\beta}$, since the length of that block is $4\cdot2^{M+2}$ rounds. By \Cref{willdosearch}, agent $\beta$ visits all nodes in the $D$-neighbourhood of $v_\beta$ within the next $44\klog{v_\beta}2^{M+2}$ rounds. However, note that $\alpha$ starts its execution of epoch $\ep{M+1}{\alpha}$ with a waiting period consisting of $36\klog{v_\alpha}2^{M+1} = 18\mu\klog{v_\beta}2^{M+2} \geq 18(3)\klog{v_\beta}2^{M+2} = 54\klog{v_\beta}2^{M+2}$ rounds. Therefore, while $\beta$ visits all nodes in the $D$-neighbourhood of $v_\beta$, agent $\alpha$ will be located at $v_\alpha$ during stage 0 of the first phase of epoch $\ep{M+1}{\alpha}$, so rendezvous will occur by the end of epoch $\ep{M+1}{\alpha} = \ep{\dcrit+\log\log^*(v_\alpha)+1}{\alpha}$ in $\alpha$'s execution.
\end{proof}

Next, we consider the remaining case, i.e., where $\mu$ is at most $3$. At a high level, we consider the round in which agent $\alpha$ starts epoch $\ep{\dcrit+\log\log^*(v_\alpha)+1}{\alpha}$, and prove that $\beta$'s execution has not yet reached the end of epoch $\ep{\dcrit+\log\log^*(v_\alpha)+1}{\beta}$. In the case where $\mu \geq 11/9$, agent $\alpha$'s waiting period is long enough to contain a search by $\beta$, so rendezvous occurs with $\beta$ as `searcher' and $\alpha$ as `sitter'. However, when $\mu < 11/9$, this is not true. So, in this case, we instead look at when $\beta$ starts epoch $\ep{\dcrit+\log\log^*(v_\alpha)+2}{\alpha}$, and notice that its waiting period at the start of this epoch is long enough to include the next search by $\alpha$, so rendezvous occurs with $\alpha$ as `searcher' and $\beta$ as `sitter'.  This idea is formalized in the following result.

\begin{lemma}\label{alphabigger:musmaller3}
Suppose that $1 < \mu \leq 3$. Rendezvous occurs by the end of agent $\alpha$'s execution of epoch $\ep{\dcrit+\log\log^*(v_\alpha)+1}{\alpha}$.
\end{lemma}
\begin{proof}
Let $M=\dcrit+\log\log^*(v_\alpha)$, and consider the round in which agent $\alpha$ starts epoch $\ep{M+1}{\alpha}$. The number of rounds that have elapsed in $\alpha$'s execution is 
\begin{align*}
	&\first{M}{\alpha}\\
	 =\ & 2\klog{v_\alpha}\left[25\cdot 2^{2M} - 2^{M} - 24\right]\\
	=\ & 2\kappa\left[\mu\log^*(v_\beta)\right]\left[25\cdot 2^{2M} - 2^{M} - 24\right]\\
	=\ & 2\kappa\left[\log^*(v_\beta) + (\mu-1)\log^*(v_\beta)\right]\left[25\cdot 2^{2M} - 2^{M} - 24\right]\\
	=\ & 2\kappa\log^*(v_\beta)\left[25\cdot 2^{2M} - 2^{M} - 24\right] + 2\kappa(\mu-1)\log^*(v_\beta)\left[25\cdot 2^{2M} - 2^{M} - 24\right]\\
	=\ &\first{M}{\beta}+ 2\kappa(\mu-1)\log^*(v_\beta)\left[25\cdot 2^{2M} - 2^{M} - 24\right]\\
	\leq\ & \first{M}{\beta}+ 2\kappa(2)\log^*(v_\beta)\left[25\cdot 2^{2M} - 2^{M} - 24\right]\\
	=\ &\first{M}{\beta}+ \kappa\log^*(v_\beta)\left[50\cdot 2^{2M+1} - 2^{M+2} - 96\right]\\
	<\ & \first{M}{\beta}+ \kappa\log^*(v_\beta)\left[50\cdot 2^{2M+1} - 2^{M}\right]\\
	=\ & \first{M}{\beta}+ \kappa\log^*(v_\beta)\left[50\cdot 2^{2(M+1)-1} - 2^{(M+1)-1}\right]\\
	=\ &\first{M}{\beta}+ |\ep{M+1}{\beta}| - 25\kappa\log^*(v_\beta)2^{2M+1}\hspace{3mm}\textrm{\color{gray}(using \Cref{fact:epochlength} with $j=M+1$)}\\
	<\ & \first{M}{\beta}+ |\ep{M+1}{\beta}| - 4\cdot2^{M+1}\\
\end{align*}

First, by \Cref{lem:crossover}, note that agent $\beta$'s execution is currently in an epoch $\ep{M'}{\beta}$ with $M' \geq \dcrit+1$. Next, since $\beta$ starts at the same time as $\alpha$ or later, the number of rounds that have elapsed in $\alpha$'s execution gives an upper bound on the number of rounds that have elapsed in $\beta$'s execution. In particular, agent $\beta$ has executed strictly fewer than $\first{M}{\beta}+ |\ep{M+1}{\beta}| - 4\cdot2^{M+1}$ rounds. But this means that $\beta$'s execution has not yet reached the final block of stage 2 of the final phase of epoch $\ep{M+1}{\beta}$, since the length of that block is $4\cdot2^{M+1}$ rounds. 

We consider two subcases based on the value of $\mu$.

\begin{itemize}
	\item Case (a): $\mu \geq 11/9$\\
	By \Cref{willdosearch}, agent $\beta$ visits all nodes in the $D$-neighbourhood of $v_\beta$ within the next $44\klog{v_\beta}2^{M+1}$ rounds. However, note that $\alpha$ starts its execution of epoch $\ep{M+1}{\alpha}$ with a waiting period consisting of $36\klog{v_\alpha}2^{M+1} = 36\mu\klog{v_\beta}2^{M+1} \geq 36(11/9)\klog{v_\beta}2^{M+1} = 44\klog{v_\beta}2^{M+1}$ rounds. Therefore, while $\beta$ visits all nodes in the $D$-neighbourhood of $v_\beta$, agent $\alpha$ will be located at $v_\alpha$ during stage 0 of the first phase of epoch $\ep{M+1}{\alpha}$, so rendezvous will occur by the end of epoch $\ep{M+1}{\alpha}$ in $\alpha$'s execution.
	\item Case (b): $1 < \mu < 11/9$\\
	Consider the round in which $\beta$ starts epoch $\ep{M+2}{\beta}$. First, note that agent $\alpha$'s execution has passed the start of epoch $\ep{M+1}{\alpha}$ at this time: indeed, we showed above that, when agent $\alpha$ started epoch $\ep{M+1}{\alpha}$, agent $\beta$'s execution had not yet reached the end of epoch $\ep{M+1}{\beta}$. Moreover, since $M \geq \dcrit$, it follows that agent $\alpha$'s execution has passed the start of epoch $\ep{\dcrit+1}{\alpha}$ at this time.
	
	Next, we set out to show that, at this time, agent $\alpha$'s execution has not yet reached the final block of the final phase of epoch $\ep{M+1}{\alpha}$. The number of rounds that elapse in agent $\beta$'s execution up until it starts epoch $\ep{M+2}{\beta}$ is
	\begin{align*}
		& \first{M+1}{\beta}\\
		=\ & 2\klog{v_\beta}\left[25\cdot 2^{2M+2} - 2^{M+1} - 24\right]\\
		\leq\ &  2\kappa(\log^*(v_\alpha)-1)\left[25\cdot 2^{2M+2} - 2^{M+1} - 24\right]\\
		=\ & 2\klog{v_\alpha}\left[25\cdot 2^{2M+2} - 2^{M+1} - 24\right] - 2\kappa\left[25\cdot 2^{2M+2} - 2^{M+1} - 24\right]\\
		=\ & \first{M+1}{\alpha} - 2\kappa\left[25\cdot 2^{2M+2} - 2^{M+1} - 24\right]\\
		=\ & \first{M+1}{\alpha} - 2\kappa\left[25\cdot 2^{2\dcrit+2\log\log^*(v_\alpha)+2} - 2^{\dcrit+\log\log^*(v_\alpha)+1} - 24\right]\\
		\leq\ & \first{M+1}{\alpha}\\
		& - 2\kappa\left[25\cdot 2^{2\dcrit+2\log\log^*(v_\alpha)+2} - 2^{\dcrit+\log\log^*(v_\alpha)+1} - 24\cdot 2^{\log\log^*(v_\alpha)}\right]\\
		=\ & \first{M+1}{\alpha} - 2\klog{v_\alpha}\left[25\cdot 2^{2\dcrit+\log\log^*(v_\alpha)+2} - 2^{\dcrit+1} - 24\right]\\
		\leq\ & \first{M+1}{\alpha} - 2\klog{v_\alpha}\left[25\cdot 2^{2\dcrit+2} - 2^{\dcrit+1} - 24\right]\\
		=\ & \first{M+1}{\alpha} - 2\klog{v_\alpha}\left[25\cdot 2^{2\dcrit} - 2^{\dcrit} - 24\right]\\ & - 2\klog{v_\alpha}\left[75\cdot 2^{2\dcrit} - 2^{\dcrit}\right]\\
		=\ & \first{M+1}{\alpha} - \first{\dcrit}{\alpha} - 2\klog{v_\alpha}\left[75\cdot 2^{2\dcrit} - 2^{\dcrit}\right]\\
		\leq\ & \first{M+1}{\alpha} - \delay - 2\klog{v_\alpha}\left[75\cdot 2^{2\dcrit} - 2^{\dcrit}\right]\\
		\leq \ & \first{M+1}{\alpha} - \delay  - 2\klog{v_\alpha}\left[74\cdot 2^{2\dcrit}\right]\\
		= \ & \first{M+1}{\alpha} - \delay - 2\kappa\left[74\cdot 2^{2\dcrit+\log\log^*(v_\alpha)}\right]\\
		\leq \ & \first{M+1}{\alpha} - \delay  - 2\kappa\left[74\cdot 2^{\dcrit+\log\log^*(v_\alpha)+1}\right]\\
		= \ & \first{M+1}{\alpha} - \delay  - 2\kappa\left[74\cdot 2^{M+1}\right]\\
		< \ & \first{M+1}{\alpha} - \delay  - 4\cdot 2^{M+1}
	\end{align*}
	
	The number of rounds that have elapsed in $\alpha$'s execution can be obtained by adding $\delay$ to the number of rounds that have elapsed in $\beta$'s execution. So, by the above bound, we have shown that fewer than $\first{M+1}{\alpha}  - 4\cdot 2^{M+1}$ have elapsed in $\alpha$'s execution at the moment when $\beta$ starts epoch $\ep{M+2}{\beta}$. Since each block in every phase of epoch $\ep{M+1}{\alpha}$ has length exactly $4\cdot 2^{M+1}$, the bound $\first{M+1}{\alpha} - 4\cdot 2^{M+1}$ tells us that $\alpha$'s execution has not yet reached the final block of stage 2 of the final phase of epoch $\ep{M+1}{\alpha}$. 
	
	From the above, we have proved that agent $\alpha$'s execution has passed the start of epoch $\ep{\dcrit+1}{\alpha}$, but has not yet reached the final block of the final phase of epoch $\ep{M+1}{\alpha}$. Thus, \Cref{lemma:willdoDsearch} tells us that agent $\alpha$ visits all nodes in the $D$-neighbourhood of $v_\alpha$ within the next $44\klog{v_\alpha}2^{M+1} = 44\mu\klog{v_\beta}2^{M+1}  < 44(11/9)\klog{v_\beta}2^{M+1} < 55\klog{v_\beta}2^{M+1}$ rounds. This will occur by the end of $\alpha$'s execution of epoch $\ep{M+1}{\alpha}$, since the final block of stage 2 of the final phase of epoch $\ep{M+1}{\alpha}$ is a searching block.
	
	However, agent $\beta$ starts its execution of epoch $\ep{M+2}{\beta}$ with a waiting period consisting of $36\klog{v_\beta}2^{M+2} = 72\klog{v_\beta}2^{M+1} > 55\klog{v_\beta}2^{M+1}$ rounds. Therefore, $\beta$ will be waiting at its starting node $v_\beta$ when $\alpha$ visits all nodes in the $D$-neighbourhood of $v_\alpha$, so rendezvous will occur by the end of agent $\alpha$'s execution of epoch $\ep{M+1}{\alpha}$.
\end{itemize}

In both cases, rendezvous occurs by the end of epoch $\ep{M+1}{\alpha} = \ep{\dcrit+\log\log^*(v_\alpha)+1}{\alpha}$ in $\alpha$'s execution. 
\end{proof}

Finally, to cover all possible values of $\mu$, we combine \Cref{alphabigger:largemu1,alphabigger:largemu2,alphabigger:mubigger4,alphabigger:mubigger3,alphabigger:musmaller3} to conclude the analysis in the case where $\log^*(v_\alpha) > \log^*(v_\beta)$.

\begin{lemma}\label{logstaralphabigger}
Rendezvous occurs within $O(D^2(\log^*(v_\alpha))^3)$ rounds in agent $\alpha$'s execution.
\end{lemma}
\begin{proof}
When $\mu > 4D$, \Cref{alphabigger:largemu1,alphabigger:largemu2} imply that rendezvous occurs by the end of epoch $\ep{\dcrit+2}{\alpha}$ in $\alpha$'s execution. By \Cref{cor:runningtime}, the total number of rounds that elapse in $\alpha$'s execution of the algorithm is $O(2^{2\dcrit+4}\log^*(v_\alpha)) = O(D^2\log^*(v_\alpha))$.

For all other values of $\mu$, \Cref{alphabigger:mubigger4,alphabigger:mubigger3,alphabigger:musmaller3} imply that rendezvous occurs by the end of agent $\alpha$'s execution of epoch $\ep{\dcrit+\log\log^*(v_\alpha)+2}{\alpha}$. By \Cref{cor:runningtime}, the total number of rounds that elapse in $\alpha$'s execution of the algorithm is 
$$O(2^{2\dcrit+2\log\log^*(v_\alpha)+4}\log^*(v_\alpha)) = O(D^2(\log^*(v_\alpha))^3).$$
\end{proof}

The analysis of the algorithm is now complete. We combine \Cref{verylargedelay,lem:analysissamelogstar,logstarbetabigger,logstaralphabigger} to prove that rendezvous is guaranteed to occur within $O(D^2(\log^*\ell)^3)$ rounds, where $\ell$ is the larger label of the agents' starting nodes.
\begin{theorem}\label{totalanalysis}
Algorithm $\rvunknownDalg$ solves rendezvous in $O(D^2(\log^*\ell)^3)$ rounds on lines with arbitrary node labelings when two agents start at arbitrary positions at unknown distance $D$, and when the delay between the rounds in which they start executing the algorithm is arbitrary.
\end{theorem}
\begin{proof}
In the case where the delay between the agents is very large, i.e., agent $\beta$ does not start its execution before agent $\alpha$ reaches the end of epoch $\ep{\dcrit}{\alpha}$ in its execution, \Cref{verylargedelay} tells us that rendezvous occurs within $O(D^2\log^*(v_\alpha))$ rounds in $\alpha$'s execution of the algorithm. When $v_\alpha \geq v_\beta$, this gives the upper bound $O(D^2\log^*\ell)$. Otherwise, when $v_\beta > v_\alpha$, we get an upper bound of $O(D^2\log^*(v_\alpha)) \subseteq O(D^2\log^*(v_\beta)) = O(D^2\log^*\ell)$.

In all other cases, we assume that agent $\beta$ starts its execution before agent $\alpha$ reaches the end of epoch $\ep{\dcrit}{\alpha}$ in its execution. In the case where $\log^*(v_\alpha) = \log^*(v_\beta)$, \Cref{lem:analysissamelogstar} tells us that rendezvous occurs within $O(D^2\log^*\ell)$ rounds. When $\log^*(v_\alpha) < \log^*(v_\beta)$, \Cref{logstarbetabigger} tells us that rendezvous occurs within $$O(\min\{D^2(\log^*(v_\beta))^2,D^4\log^*(v_\beta)\}) = O(D^2(\log^*\ell)^2)$$ 
rounds. When $\log^*(v_\alpha) > \log^*(v_\beta)$, \Cref{logstaralphabigger} tells us that rendezvous occurs within $O(D^2(\log^*(v_\alpha))^3) = O(D^2(\log^*\ell)^3)$ rounds.
\end{proof}

%--------------------------------------------------
%--------------------------------------------------
\section{The Line Colouring Algorithm}\label{sec:linecolouring}
%--------------------------------------------------
%--------------------------------------------------

As announced in \Cref{CValg} in \Cref{sec:tools}, in this section we design and analyze a deterministic distributed algorithm \EarlyStopCV that, in the $\mathcal{LOCAL}$ model, properly 3-colours any infinite labeled line such that the execution of this algorithm at any node $x$ with initial label $\mathrm{ID}_x$ terminates in time $O(\log^*(\mathrm{ID}_x))$. It is assumed that the initial node labels form a proper colouring of the line. We describe the algorithm in \Cref{sec:linecolouringalg,sec:stringops,sec:linecolouringpseudo}. In \Cref{sec:linecolouringcomplexity}, we prove that each node $x$ with initial label $\mathrm{ID}_x$ terminates within time $\log^*(\mathrm{ID}_x) + 59$. In \Cref{sec:linecolouringcorrectness}, we prove that the algorithm is correct. In \Cref{sec:linecolouringextend}, we discuss how the algorithm can be applied to cycles and finite paths.

%--------------------------------------------------
%--------------------------------------------------
\subsection{Algorithm description}\label{sec:linecolouringalg}
%--------------------------------------------------
%--------------------------------------------------
Our algorithm builds upon a solution to a problem posed in Homework Exercises 1 from the Principles of Distributed Computing course at ETH Z\"{u}rich \cite{DisCo}.

At a very high level, we want to execute the Cole and Vishkin algorithm \cite{CV}, but with some modifications. We introduce four special colours, denoted by $\lmin$, $\lmax$, $\mathit{Pdone}$, $\mathit{Cdone}$, that are defined in such a way that they are not equal to any colour that could be chosen using the Cole and Vishkin strategy (e.g., they could be defined as negative integers). These colours will be chosen by a node in certain situations where it needs to choose a colour that is guaranteed to be different from its neighbour's, but using the Cole and Vishkin strategy might not work. Another significant modification is that each node will actually choose two colours: its ``final colour'' that it will output upon terminating the algorithm, but also an intermediate ``Phase 1 colour''. These roughly correspond to the two parts of the Cole and Vishkin strategy: in Phase 1, each node picks a colour that is guaranteed to be different from its neighbours' chosen colour (but selects it from a `large' range of colours), and then in Phase 2, each node picks a final colour from the set $\{0,1,2\}$. However, since we want to allow different nodes to execute different phases of the algorithm at the same time (since we want to allow them to terminate at different times), each node must maintain and advertise its ``Phase 1'' and ``final'' colours separately, so that another node that is still performing Phase 1 is basing its decisions on its neighbours' Phase 1 colours (and not their final colours). Finally, we note that the original Cole and Vishkin strategy is described for a directed tree, i.e., where each node has at most one parent and perhaps some children, whereas our algorithm must work in an undirected infinite line, so we introduce a pre-processing step (Phase 0) to set up parent/child relationships.

The first part of the algorithm partitions the undirected infinite line into directed sub-lines that will perform the rest of the algorithm in parallel. In round 1, each node shares its ID with both neighbours. In round 2 (also referred to as Phase 0), each node compares its own ID with those that it received from its neighbours in round 1. If a node determines that it is a \emph{local minimum} (i.e., its ID is less than the ID's of its neighbours), then it picks the special colour $\lmin$ as its Phase 1 colour and proceeds directly to Phase 2 without performing Phase 1. If a node determines that it is a \emph{local maximum} (i.e., its ID is greater than the ID's of its neighbours), then it picks the special colour $\lmax$ as its Phase 1 colour and proceeds directly to Phase 2 without performing Phase 1. All other nodes, i.e., those that are not a local minimum or a local maximum, pick their own ID as their Phase 1 colour and continue to perform Phase 1. Further, each such node also chooses one \emph{parent} neighbour (the neighbour with smaller ID) and one \emph{child} neighbour (the neighbour with larger ID). In particular, the nodes that perform Phase 1 are each part of a directed sub-line that is bordered by local maxima/minima, so the sub-lines can all perform Phase 1 in parallel without worrying that the chosen colours will conflict with colours chosen in other sub-lines.

The second part of the algorithm, referred to as Phase 1, essentially implements the Cole and Vishkin strategy within each directed sub-line. The idea is that each node follows the Cole and Vishkin algorithm using its own Phase 1 colour and the Phase 1 colour of its parent until it has selected a Phase 1 colour in the range $\{0,\ldots,\maxCol\}$ that is different from its parent's Phase 1 colour, and then the node will proceed to Phase 2 where it will pick a final colour in the range $\{0,1,2\}$. However, this idea must be modified to avoid three potential pitfalls:
\begin{enumerate}
	\item Suppose that, in a round $t$, a node $v$ has a parent or child that stopped performing Phase 1 in an earlier round (and possibly has finished Phase 2 and has stopped executing the algorithm). If node $v$ still has a large colour in round $t$, i.e., not in the range $\{0,\ldots,\maxCol\}$, but continues to use the Cole and Vishkin strategy, then there is no guarantee that the new colour it chooses will still be different than those of its neighbours. This is not an issue for the original Cole and Vishkin algorithm, since it was designed in such a way that all nodes execute the algorithm the exact same number of times (using \emph{a priori} knowledge of the network size). To deal with this scenario, node $v$ first looks at the Phase 1 colours that were most recently advertised by its parent and child, and sees if either of them is in the range $\{0,\ldots,\maxCol\}$. If node $v$ sees that its parent's Phase 1 colour is in the range $\{0,\ldots,\maxCol\}$, then it knows that its parent is not performing Phase 1 in this round, so instead of following the Cole and Vishkin strategy, it immediately adopts the special colour \textit{Pdone} as its Phase 1 colour and moves on to Phase 2 of the algorithm. On the other hand, if node $v$ sees that its child's colour is in the range $\{0,\ldots,\maxCol\}$, then it knows that its child is not performing Phase 1 in this round, so instead of following the Cole and Vishkin strategy, it immediately adopts the special colour \textit{Cdone} as its Phase 1 colour and moves on to Phase 2 of the algorithm. By adopting the special colours in this way, node $v$'s Phase 1 colour is guaranteed to be different than any non-negative integer colour that was previously adopted by its neighbours.
	\item Suppose that a node $v$'s parent has a Phase 1 colour equal to a special colour (i.e., one of $\lmin,\lmax,\mathit{Pdone},\mathit{Cdone}$). The Cole and Vishkin strategy is not designed to work with such special colours, so, a node $v$ with such a parent will pretend that its parent has colour 0 instead. By doing this, it will choose some non-negative integer as dictated by the Cole and Vishkin strategy, and this integer is guaranteed to be different from all special colours, so $v$'s chosen Phase 1 colour will be different from its parent's.
	\item Suppose that a node $v$ has a parent with an extremely large integer colour. The Cole and Vishkin strategy will make sure that $v$ chooses a new Phase 1 colour that is different than the Phase 1 colour chosen by its parent, however, it is not guaranteed that this new colour is significantly smaller than the colour $v$ started with. In particular, we want to guarantee that node $v$ terminates Phase 1 within $\log^*(\mathrm{ID}_v)$ rounds, so we want node $v$'s newly-chosen Phase 1 colour to be bounded above by a logarithmic function of its \emph{own} colour in every round (and never depend on its parent's much larger colour). To ensure this, we apply a suffix-free encoding to the binary representation of each node's colour {\bf before} applying the Cole and Vishkin strategy. Doing this guarantees that the smallest index where two binary representations of colours differ is bounded above by the length of the binary representation of the \emph{smaller} colour.
\end{enumerate} 
When a node is ready to proceed to Phase 2, it has chosen a Phase 1 colour from the set $\{0,\ldots,\maxCol\} \cup \{\lmin,\lmax\} \cup \{\mathit{Pdone},\mathit{Cdone}\}$, and its chosen Phase 1 colour is guaranteed to be different from the Phase 1 colours of its two neighbours.

The third part of the algorithm, referred to as Phase 2, uses a round-robin strategy over the 56 possible Phase 1 colours, which guarantees that any two neighbouring nodes pick their final colour in different rounds. In particular, when executing each round of Phase 2, a node calculates the current \emph{token} value, which is defined as the current round number modulo 56. We assume that all nodes start the algorithm at the same time, so the current token value is the same at all nodes in each round. In each round, each node that is performing Phase 2 compares the token value to its own Phase 1 colour. If a node $v$'s Phase 1 colour is in $\{0,\ldots,51\}$ and is equal to the current token value, then it proceeds to choose its final colour (as described below). Otherwise, if a node $v$'s Phase 1 colour is one of the special colours, it waits until the current token value is equal to a value that is dedicated to that special colour (i.e., 52 for $\lmin$, 53 for $\lmax$, 54 for $\mathit{Pdone}$, 55 for $\mathit{Cdone}$), then chooses its final colour in that round. To choose its final colour, node $v$ chooses the smallest colour from $\{0,1,2\}$ that was never previously advertised as a final colour by its neighbours. Then, $v$ immediately sends out a message to its neighbours to advertise the final colour that it chose, then $v$ terminates. Since two neighbouring nodes are guaranteed to have different Phase 1 colours, they will choose their final colour in different rounds, so the later of the two nodes always avoids the colour chosen by the earlier node, and there is always a colour from $\{0,1,2\}$ available since each node only has two neighbours.

%--------------------------------------------------
%--------------------------------------------------
\subsection{String operations}\label{sec:stringops}
%--------------------------------------------------
%--------------------------------------------------
For any two binary strings $S_1,S_2$, denote by $S_1 \cdot S_2$ the concatenation of string $S_1$ followed by string $S_2$. For any positive integer $i$, the function $\textsc{BinaryRep}(i)$ returns the binary string consisting of the base-2 representation of $i$. Conversely, for any binary string $S$, the function $\textsc{IntVal}(S)$ returns the integer value when $S$ is interpreted as a base-2 integer. For any binary string $S$ of length $\ell \geq 1$, the string is a concatenation of bits, i.e., $S = s_{\ell-1}\cdots s_0$, and we will write $S[i]$ to denote the bit $s_i$. The notation $|S|$ denotes the length of $S$.

For any binary string $S$ of length $\ell \geq 1$, we define a function $\textsc{EncodeSF}(S)$ that returns the binary string obtained by replacing each 0 in $S$ with 01, replacing each 1 in $S$ with 10, then prepending 00 to the result. More formally, $\textsc{EncodeSF}(S)$ returns a string $S'$ of length $2\ell+2$ such that $S'[2\ell+1]=S'[2\ell]=0$, and, for each $i \in \{0,\ldots,\ell-1\}$, $S'[2i+1]=S[i]$ and $S'[2i]=1-S[i]$. For example, $\textsc{EncodeSF}(101)=00100110$. The function $\textsc{EncodeSF}$ is an encoding method with two important properties: an encoded string is uniquely decodable, and,  no encoded string is a suffix of another encoded string. These properties are formalized in the following results.

\begin{proposition}\label{SFDecodable}
	For any two non-empty binary strings $S,T$, we have that $\textsc{EncodeSF}(S) = \textsc{EncodeSF}(T)$ if and only if $S=T$.
\end{proposition}

\begin{proposition}\label{SuffixFree}
	Consider any non-empty binary strings $S_1,S_2$. Let $S_1' = \textsc{EncodeSF}(S_1)$ and let $S_2'=\textsc{EncodeSF}(S_2)$. If $S_1 \neq S_2$, then there exists an $i \leq \min\{|S_1'|-1,|S_2'|-1\} = \min\{2|S_1|+1,2|S_2|+1\}$ such that $S_1'[i] \neq S_2'[i]$.
\end{proposition}
\begin{proof}
	Let $\ell_1 = |S_1|$ and let $\ell_2 = |S_2|$. Without loss of generality, assume that $\ell_1 \leq \ell_2$. Suppose that $S_1 \neq S_2$. We consider two cases.
	
	In the first case, suppose that $\ell_1 = \ell_2$. Since $S_1 \neq S_2$, there exists an $j \in \{0,\ldots,\ell_1-1\}$ such that $S_1[j] = 1-S_2[j]$. From the definition of $\textsc{EncodeSF}$, we see that $S_1'[2j+1] = S_1[j] = 1-S_2[j] = 1-S_2'[2j+1]$, which proves that $S_1'[2j+1] \neq S_2'[2j+1]$. Finally, from the choice of $j$, note that $2j+1 \leq 2(\ell_1-1)+1 = 2\ell_1-1 = 2\ell_2-1$. Setting $i=2j+1$, we see that $S_1'[i] \neq S_2'[i]$, and, $i \leq \min\{2\ell_1-1,2\ell_2-1\} = \min\{2|S_1|-1,2|S_2|-1\} < \min\{2|S_1|+1,2|S_2|+1\} = \min\{|S_1'|-1,|S_2'|-1\}$, as desired.
	
	In the second case, suppose that $\ell_1 < \ell_2$. From the definition of $\textsc{EncodeSF}(S_2)$, we see that $S_2'[2\ell_1+1]=S_2[\ell_1]$ and $S_2'[2\ell_1]=1-S_2[\ell_1]$. In particular, either $S_2'[2\ell_1+1]$ or $S_2'[2\ell_1]$ is equal to 1. However, from the definition of $\textsc{EncodeSF}(S_1)$, we see that $S_1'[2\ell_1+1]=S_1'[2\ell_1]=0$. Thus, there is an $i \in \{2\ell_1,2\ell_1+1\}$ such that $S_1'[i] \neq S_2'[i]$. Finally, note that $i \leq 2\ell_1+1 = \min\{2\ell_1+1,2\ell_2+1\}=\min\{2|S_1|+1,2|S_2|+1\} = \min\{|S_1'|-1,|S_2'|-1\}$, as desired.
\end{proof}

%--------------------------------------------------
%--------------------------------------------------
\subsection{Algorithm pseudocode}\label{sec:linecolouringpseudo}
%--------------------------------------------------
%--------------------------------------------------
In what follows, we assume that all nodes start the algorithm simultaneously starting in round 1, and each node has access to a variable called $\mathtt{clockVal}$ which accurately contains the current round number. Each node $v$ has a label denoted by $\mathrm{ID}_v$. We assume that each node can distinguish between its two neighbours, and it does so by referring to them as using fixed labels $A$ and $B$. There are no additional assumptions about orientation, i.e., each node independently chooses which of its neighbours it references as $A$ or $B$. Consider two neighbouring nodes $v,w$: if $w$ is the node that $v$ references using $A$, then we say that $w$ is $v$'s \emph{$A$-neighbour}, and otherwise $w$ is $v$'s \emph{$B$-neighbour}.

We define four special colours $\lmin, \lmax, \mathit{Pdone}, \mathit{Cdone}$ that are not positive integers (i.e., they cannot be confused with any node's ID, and they cannot be confused with any non-negative integer colour chosen by a node during the algorithm's execution). Practically speaking, one possible implementation is to use $\lmin=-4$, $\lmax=-3$, $Pdone=-2$, and $Cdone=-1$.

\Cref{EarlyStopCV} provides the pseudocode for the algorithm's execution. Round 1 is an initialization round in which all nodes send their ID to their neighbours. In round 2, we say that the nodes \emph{perform Phase 0}, and they will all execute lines \ref{StartPhase0}--\ref{EndPhase0}. The purpose of this phase is for each node to classify itself based on the relationship between its own ID and the ID's of its neighbours. For any round $r \geq 3$, we say that a node $v$ \emph{performs Phase 1 during round $r$} if it executes the lines \ref{StartPhase1}--\ref{UpdateDP1} during round $r$. For any round $r \geq 3$, we say that a node $v$ \emph{performs Phase 2 during round $r$} if it executes the lines \ref{StartPhase2}--\ref{EndPhase2} during round $r$.

Each node $v$ will keep track of its neighbours' chosen colours. However, it is possible that a neighbour $w$ of $v$ will finish the algorithm while $v$ is still in Phase 1, and we want $v$ to be able to make its Phase 1 decisions based on the colour that $w$ chose for Phase 1. Moreover, when $v$ eventually enters Phase 2, we want $v$ to be able to make its Phase 2 decisions based on the final colour that $w$ chose (which $w$ previously advertised while $v$ was still in Phase 1). So, for neighbour $A$, node $v$ will have a variable called $\mathtt{Acol.P1}$ in which it will store neighbour $A$'s most recent colour choice when $A$ was in Phase 1, and a variable called $\mathtt{Acol.final}$ in which it will store neighbour $A$'s final colour choice. Similarly, for neighbour $B$, node $v$ will have variables called $\mathtt{Bcol.P1}$ and $\mathtt{Bcol.final}$. All of these variables are initially null (i.e., there is no default integer value that could be confused with an actual colour choice). If a node sends out a message containing its colour choice at the end of a round, it will make sure to include an additional string in the message (either ``P1'' or ``final'') to specify whether the colour is the node's Phase 1 colour or the node's final colour, so that receivers can store the value in the correct variable.

\begin{algorithm}
	\footnotesize
	\caption{\EarlyStopCV}\label{EarlyStopCV}
	\begin{algorithmic}[1]
		\Statex {\color{gray} \%\% $\mathtt{clockVal}$ is assumed to contain the current local round number, starting at 1.}
		\Statex {\color{gray} \%\% $\mathtt{myID}$ is assumed to contain the node's label.}
		\Statex {\color{gray} \%\% Initially, $\mathtt{Acol.P1}=\mathtt{Bcol.P1}=\mathtt{Acol.final}=\mathtt{Bcol.final}=\mathit{null}$}
		\If{$\mathtt{clockVal}==1$} \Comment{{\scriptsize Round 1: get IDs of neighbours}}
		\State {Send $\mathtt{myID}$ to both neighbours}
		\State {Receive $\mathtt{ID}_A$ from $A$ and receive $\mathtt{ID}_B$ from $B$}
		\ElsIf{$\mathtt{clockVal}==2$} \Comment{{\scriptsize Phase 0: detect if local max or local min, otherwise assign parent and child}}
		\If{$\mathtt{myID} < \mathtt{ID}_A$ \textbf{and} $\mathtt{myID} < \mathtt{ID}_B$}\label{StartPhase0} \Comment{{\scriptsize I'm a local minimum}}
		\State $\mathtt{myPhase1Col} \leftarrow \lmin$
		\ElsIf{$\mathtt{myID} > \mathtt{ID}_A$ \textbf{and} $\mathtt{myID} > \mathtt{ID}_B$} \Comment{{\scriptsize I'm a local maximum}}
		\State $\mathtt{myPhase1Col} \leftarrow \lmax$
		\ElsIf{$\mathtt{ID}_A < \mathtt{myID}$ \textbf{and} $\mathtt{myID} < \mathtt{ID}_B$} \Comment{{\scriptsize neighbourhood IDs increase towards B}}
		\State $\mathtt{myPhase1Col} \leftarrow \mathtt{myID}$
		\State $\mathtt{parent} \leftarrow A$ \label{setParentA}
		\State $\mathtt{child} \leftarrow B$ \label{setChildB}
		\Else \Comment{{\scriptsize neighbourhood IDs increase towards A}}
		\State $\mathtt{myPhase1Col} \leftarrow \mathtt{myID}$
		\State $\mathtt{parent} \leftarrow B$ \label{setParentB}
		\State $\mathtt{child} \leftarrow A$ \label{setChildA}
		\EndIf
		\State {Send $\mathtt{myPhase1Col}$ to both neighbours} \label{P0SendP1Col}
		\State {Receive $\mathtt{msg}_A$ from $A$ and receive $\mathtt{msg}_B$ from $B$}
		\State $\mathtt{Acol.P1} \leftarrow \mathtt{msg}_A$\label{P0StoreA}
		\State $\mathtt{Bcol.P1} \leftarrow \mathtt{msg}_B$\label{P0StoreB}
		\State $\mathtt{doPhase1} \leftarrow (\mathtt{myPhase1Col} \not\in \{0,\ldots,\maxCol,\lmin,\lmax\})$\label{EndPhase0}
		
		\ElsIf{$\mathtt{doPhase1} == \mathbf{true}$}\label{Phase1Condition} \Comment{{\scriptsize Phase 1: detect if one of my neighbours is settled, otherwise perform CV}}
		\If{$\mathtt{parent}==A$}\label{StartPhase1}
		\State $\mathtt{myPhase1Col} \leftarrow \ChooseNewColour (\mathtt{myPhase1Col},\mathtt{Acol.P1},\mathtt{Bcol.P1})$\label{NewColourParentA}
		\Else
		\State $\mathtt{myPhase1Col} \leftarrow \ChooseNewColour (\mathtt{myPhase1Col},\mathtt{Bcol.P1},\mathtt{Acol.P1})$\label{NewColourParentB}
		\EndIf \label{EndSwitch}
		\State {Send (``P1'', $\mathtt{myPhase1Col}$) to both neighbours}\label{P1SendP1}
		\State \textbf{if} received a message of the form (\textit{key}, \textit{val}) from $A$ \textbf{then}: $\mathtt{Acol.}\mathit{key} \leftarrow \mathit{val}$\label{fromAP1}
		\State \textbf{if} received a message of the form (\textit{key}, \textit{val}) from $B$ \textbf{then}: $\mathtt{Bcol.}\mathit{key} \leftarrow \mathit{val}$\label{fromBP1}
		\State $\mathtt{doPhase1} \leftarrow (\mathtt{myPhase1Col} \not\in \{0,\ldots,\maxCol,\mathit{Pdone},\mathit{Cdone}\})$\label{UpdateDP1}
		\Else \Comment{{\scriptsize Phase 2: colour reduction down to $\{0,1,2\}$}}
		\State $\mathtt{token} \leftarrow \mathtt{clockVal} \bmod 56$\label{StartPhase2}
		\If{($\mathtt{myPhase1Col} == \mathtt{token}$) \textbf{or}\newline 
			\hspace*{8mm}($(\mathtt{myPhase1Col} == \lmin) \wedge (\mathtt{token} == 52)$) \textbf{or}\newline
			\hspace*{8mm}($(\mathtt{myPhase1Col} == \lmax) \wedge (\mathtt{token} == 53)$) \textbf{or}\newline
			\hspace*{8mm}($(\mathtt{myPhase1Col} == \mathit{Pdone}) \wedge (\mathtt{token} == 54)$) \textbf{or}\newline
			\hspace*{8mm}($(\mathtt{myPhase1Col} == \mathit{Cdone}) \wedge (\mathtt{token} == 55)$)}\label{termCondition}
		\State\hspace*{8mm} $\mathtt{myFinalCol} \leftarrow $ smallest element in $\{0,1,2\} \setminus \{\mathtt{Acol.final},\mathtt{Bcol.final}\}$\label{ChooseFinal}
		\State\hspace*{8mm} {Send (``final'', $\mathtt{myFinalCol}$) to both neighbours}
		\State\hspace*{8mm} \textbf{terminate()}\label{terminate}
		\EndIf\label{EndPhase2}
		\State \textbf{if} received a message of the form (\textit{key}, \textit{val}) from $A$ \textbf{then}: $\mathtt{Acol.}\mathit{key} \leftarrow \mathit{val}$\label{fromAP2}
		\State \textbf{if} received a message of the form (\textit{key}, \textit{val}) from $B$ \textbf{then}: $\mathtt{Bcol.}\mathit{key} \leftarrow \mathit{val}$\label{fromBP2}
		\EndIf			
	\end{algorithmic}
\end{algorithm}

\Cref{ChooseNewColour} provides the pseudocode describing the procedure that each node will use to choose its new colour in each round that it performs Phase 1. When calling this procedure, the node must pass in its own Phase 1 colour as the first input parameter, the Phase 1 colour of its parent as the second input parameter, and the Phase 1 colour of its child as the third input parameter. 

\Cref{CVChoice} provides the pseudocode that describes our modified Cole and Vishkin strategy for choosing a new colour. The first two lines apply the suffix-free encoding, and the remaining lines implement the classic Cole and Vishkin colour choice. When calling this procedure, the node must pass in its own Phase 1 colour as the first input parameter, and the Phase 1 colour of its parent as the second input parameter.

\begin{algorithm}
	\caption{\footnotesize$\ChooseNewColour (\mathtt{myPhase1Col},\mathtt{ParentPhase1Col},\mathtt{ChildPhase1Col})$}\label{ChooseNewColour}
	\begin{algorithmic}[1]
		\If{$\mathtt{ParentPhase1Col} \in \{0,\ldots,\maxCol\}$}\label{IfPSettled}
		\State $\mathtt{newColour} \leftarrow \mathit{Pdone}$\label{ReturnPdone}
		\ElsIf{$\mathtt{ChildPhase1Col} \in \{0,\ldots,\maxCol\}$}\label{IfCSettled}
		\State $\mathtt{newColour} \leftarrow \mathit{Cdone}$\label{ReturnCdone}
		\ElsIf{$\mathtt{ParentPhase1Col} \in \{\mathit{Pdone},\mathit{Cdone},\lmin,\lmax\}$}\label{IfPdone}
		\State $\mathtt{newColour} \leftarrow \CVChoice(\mathtt{myPhase1Col},0)$\label{doCVWith0}
		\Else
		\State $\mathtt{newColour} \leftarrow \CVChoice (\mathtt{myPhase1Col},\mathtt{ParentPhase1Col})$\label{doCV}
		\EndIf
		\State \Return $\mathtt{newColour}$
	\end{algorithmic}
\end{algorithm}

\begin{algorithm}
	\caption{$\CVChoice(\mathtt{MyCol},\mathtt{OtherCol})$}\label{CVChoice}
	\begin{algorithmic}[1]
		\State $\mathtt{MyString} \leftarrow \textsc{EncodeSF}(\textsc{BinaryRep}(\mathtt{MyCol}))$\label{EncodeMS}
		\State $\mathtt{OtherString} \leftarrow \textsc{EncodeSF}(\textsc{BinaryRep}(\mathtt{OtherCol}))$\label{EncodeOS}
		\State $\mathtt{i} \leftarrow $ smallest $x \geq 0$ such that $\mathtt{MyString}[x] \neq \mathtt{OtherString}[x]$\label{FindDiffIndex}
		\State $\mathtt{newString} \leftarrow \textsc{BinaryRep}(\mathtt{i}) \cdot \mathtt{myString}[\mathtt{i}]$\label{SetCVChoice}
		\State \Return $\textsc{IntVal}(\mathtt{newString})$\label{ReturnCVChoice}
	\end{algorithmic}
\end{algorithm}

%--------------------------------------------------
%--------------------------------------------------
\subsection{Algorithm analysis: round complexity}\label{sec:linecolouringcomplexity}
%--------------------------------------------------
%--------------------------------------------------

The original Cole and Vishkin algorithm guarantees that each time a node chooses a new colour, the new colour comes from a drastically smaller range of colours. We now prove a similar result about the output of \CVChoice each time that it is called. In our version, we need to incorporate the fact that we are first applying the suffix-free encoding (which increases the length of the strings), and we express the size of the new colour in terms of the \emph{smaller} of the two input colours.
\begin{lemma}\label{CVChoiceOutput}
	For any two non-negative integers $I_1,I_2$ such that $I_1 \neq I_2$, consider the execution of \CVChoice ($I_1,I_2$). If $\min\{I_1,I_2\} \in \{0,\ldots,2^{\ell_{min}}-1\}$ for some integer $\ell_{min} \geq 1$, then the value returned at the end of the execution is an integer in the range $\{0,\ldots,8\ell_{min}+3\}$.
\end{lemma}
\begin{proof}
	Let $I = \min\{I_1,I_2\}$, and suppose that $I \in \{0,\ldots, 2^{\ell_{min}}-1\}$ for some $\ell_{min} \geq 1$. It follows that the binary representation of $I$ has length at most $\ell_{min}$, so, from the definition of \textsc{EncodeSF}, the result of $\textsc{EncodeSF}(\textsc{BinaryRep}(I))$ has length at most $2\ell_{min}+2$. By \Cref{SuffixFree}, line \ref{FindDiffIndex} sets $i$ to be a non-negative integer in the range $\{0,\ldots,2\ell_{min}+1\}$. Line \ref{SetCVChoice} sets $\mathtt{newString}$ to a string that has length $|\textsc{BinaryRep}(i)|+1$. But the length of the binary representation of $i$ is at most $\lfloor \log_2(i) \rfloor + 1$, so the length of $\mathtt{newString}$ is at most $\lfloor \log_2(i) \rfloor + 2 \leq \log_2(i)+2$. Since $i \leq 2\ell_{min}+1$, we get that the length of $\mathtt{newString}$ is at most $\log_2(2\ell_{min}+1)+2 = \log_2(2\ell_{min}+1)+\log_2(4) = \log_2(8\ell_{min}+4)$. When interpreted as an integer in base 2, we get an integer in the range $\{0,\ldots,2^{|\mathtt{newString}|}-1\} \subseteq \{0,\ldots,2^{\log_2(8\ell_{min}+4)}-1\} = \{0,\ldots,8\ell_{min}+3\}$, as desired.
\end{proof}

In each round that a node performs Phase 1, it chooses a new Phase 1 colour using the procedure \ChooseNewColour. The strategy that is used to choose the new colour depends on the most recent Phase 1 colours of its neighbours. The following result summarizes the range of possible colours that might be chosen.  
\begin{lemma}\label{NewColourOptions}
	Let $\mathtt{myPhase1Col} \in \mathbb{Z}_{\geq 0}$ and let $\ell = |\textsc{BinaryRep}(\mathtt{myPhase1Col})|$. Let $\mathtt{ParentPhase1Col},\mathtt{ChildPhase1Col} \in \mathbb{Z}_{\geq 0} \cup \{\mathit{Pdone},\mathit{Cdone}\} \cup \{\lmin,\lmax\}$. The execution of $$\ChooseNewColour (\mathtt{myPhase1Col},\mathtt{ParentPhase1Col},\mathtt{ChildPhase1Col})$$ returns a value from the set $\{\mathit{Pdone},\mathit{Cdone}\} \cup \{0,\ldots,11\} \cup \{0,\ldots,8\ell+3\}$.
\end{lemma}
\begin{proof}
	If the condition at line \ref{IfPSettled} evaluates to true, then the value returned by the execution is $\mathit{Pdone}$. If the condition at line \ref{IfCSettled} evaluates to true, then the value returned by the execution is $\mathit{Cdone}$. If the condition at line \ref{IfPdone} evaluates to true, then the value returned by the execution is the value returned by $\CVChoice (\mathtt{MyCol},0)$, which, by \Cref{CVChoiceOutput}, is an integer in the range $\{0,\ldots,8|\textsc{BinaryRep}(0)|+3\} = \{0,\ldots,11\}$. 
	
	In all other cases, the value that is returned by the execution is the value returned by $\CVChoice (\mathtt{MyPhase1Col},\mathtt{ParentPhase1Col})$. First, note that $\mathtt{ParentPhase1Col}$ is a non-negative integer since the condition at line \ref{IfPdone} evaluated to false, so both input parameters to \CVChoice  are non-negative integers.\\
	Let $\ell_{min} = \min\{|\textsc{BinaryRep}(\mathtt{MyPhase1Col})|,|\textsc{BinaryRep}(\mathtt{ParentPhase1Col})|\}$.\\
	By \Cref{CVChoiceOutput}, the value returned is an integer in the range $\{0,\ldots,8\ell_{min}+3\}$. But $\ell_{min} \leq \ell$, so we get that the return value is an integer in the range $\{0,\ldots,8\ell+3\}$.
\end{proof}

Consider the execution of the algorithm at any fixed node $v$. For each $j \geq 2$, denote by $\mathtt{myPhase1Col}_{j}$ the value of $\mathtt{myPhase1Col}$ at the end of round $j$ in $v$'s execution, and denote by $\ell_j$ the length of the binary representation of $\mathtt{myPhase1Col}_{j}$.

Note that if $\mathtt{doPhase1}$ is set to \textit{false} in any round $j \geq 2$, then $v$ never performs Phase 1 in any round after $j$: indeed, the value of $\mathtt{doPhase1}$ is only modified in round 2 (i.e., at line \ref{EndPhase2}) or in a round during which Phase 1 is performed (i.e., at line \ref{UpdateDP1}). This fact can be used to prove the following two results, which correspond to the two possible boolean values of $\mathtt{doPhase1}$ at the end of round 2.
\begin{proposition}\label{NeverPhase1}
	Node $v$ has $\mathtt{myPhase1Col}_2 \in \{0,\ldots,\maxCol\} \cup \{\lmin,\lmax\}$ if and only if $v$ never performs Phase 1.
\end{proposition}

\begin{proposition}\label{Phase1Rounds}
	If the number of times that an arbitrary node $v$ performs Phase 1 is $m \geq 1$, then the set of rounds in which $v$ performs Phase 1 is exactly $\{3,\ldots,m+2\}$.
\end{proposition}

The following result gives a bound on an arbitrary node $v$'s Phase 1 colour based on the current round number. The key to this result is to observe that a node applies the (modified) Cole and Vishkin approach in all Phase 1 rounds except for the last one, so we can repeatedly apply the bound from \Cref{CVChoiceOutput} to all but the last Phase 1 round.
\begin{lemma}\label{ColourDecreases}
	Suppose that the number of times that an arbitrary node $v$ performs Phase 1 is $m \geq 2$. Then, for each $j \in \{3,\ldots, m+1\}$, $\mathtt{myPhase1Col}_{j}$ at node $v$ is an integer in the range $\{0,\ldots,8\ell_{j-1}+3\}$.
\end{lemma}
\begin{proof}
	Consider an arbitrary $j \in \{3,\ldots,m+1\}$ and consider $v$'s execution of Phase 1 in round $j$. The value of $\mathtt{myPhase1Col}$ at the start of the current round is equal to $\mathtt{myPhase1Col}_{j-1}$, and so the length of the binary representation of $\mathtt{myPhase1Col}$ at the start of the current round is equal to $\ell_{j-1}$.
	
	Since $j \leq m+1$, it follows from \Cref{Phase1Rounds} that $v$ will perform Phase 1 in round $j+1$. In particular, this means that the value \textit{true} is assigned to $\mathtt{doPhase1}$ at the end of round $j$ (i.e., at line \ref{UpdateDP1}). It follows that $\mathtt{myPhase1Col}_{j}$ is not equal to $\mathit{Cdone}$ or $\mathit{Pdone}$. Thus, by \Cref{NewColourOptions}, the value of $\mathtt{myPhase1Col}_{j}$ is in $\{0,\ldots,8\ell_{j-1}+3\}$, as desired.
\end{proof}

Using the previous bound on the Phase 1 colour in each round, we can now reason about how many times Phase 1 will be performed. We will use the following notation for nested logarithms: define $\log_2^{(0)}(x) := x$, and, for all $i \geq 1$, define $\log_2^{(i)}(x) := \log_2(\log_2^{(i-1)}(x))$. The iterated logarithm, denoted by $\log^*(x)$, is defined as the smallest non-negative value of $i$ such that $\log_2^{(i)}(x) \leq 1$. We prove that an arbitrary node $v$ performs Phase 1 at most $\log^*(\mathrm{ID}_v)+1$ times.

\begin{lemma}\label{Phase1Bound}
	Consider any node $v$ and denote by $\mathrm{ID}_v$ the label of $v$. If $\mathrm{ID}_v > \maxCol$, then $v$ performs Phase 1 at most $\log^*(\mathrm{ID}_v)+1$ times.
\end{lemma}
\begin{proof}
	Suppose that ID$_v > \maxCol$. If $v$ performs Phase 1 at most once, then the desired statement is true since $1 \leq \log^*(\mathrm{ID}_v)+1$. So, the remainder of the proof deals with executions in which $v$ performs Phase 1 at least twice. 
	
	To obtain a contradiction, assume that $v$ performs Phase 1 at least $\log^*(\mathrm{ID}_v)+2 \geq 2$ times. 
	
	First, suppose that $\log^*(\mathrm{ID}_v) \leq 3$. Then $\mathrm{ID}_v \leq 16 < \maxCol$. However, by \Cref{NeverPhase1}, this means that $v$ never performs Phase 1, which contradicts our assumption that $v$ performs Phase 1 at least twice.
	
	Next, suppose that $\log^*(\mathrm{ID}_v) \geq 4$. We first prove a claim about how quickly the value of $\mathtt{myPhase1Col}$ decreases over the iterations of Phase 1.
	\begin{claim}\label{fastdecrease}
		For all $j \in \{2,\ldots,\log^*(\mathrm{ID}_v)-1\}$, we have the inequality\\ $\mathtt{myPhase1Col}_j \leq 20\log_2^{(j-2)}(\mathrm{ID}_v)$.
	\end{claim}
	To prove the claim, we proceed by induction on the value of $j$. 
	
	For the base case, consider $j=2$. Since $v$ performs Phase 1 at least twice, we can conclude from \Cref{NeverPhase1} that the value of $\mathtt{myPhase1Col}_2$ is neither $\lmin$ or $\lmax$, so from the code we see that the value of $\mathtt{myPhase1Col}_2$ is $\mathrm{ID}_v$. So we get that $20\log_2^{(j-2)}(\mathrm{ID}_v) = 20\log_2^{(0)}(\mathrm{ID}_v) = 20(\mathrm{ID}_v) \geq \mathrm{ID}_v = \mathtt{myPhase1Col}_2$, as desired.
	
	As induction hypothesis, assume that for some $j \in \{2,\ldots,\log^*(\mathrm{ID}_v)-2\}$, we have $\mathtt{myPhase1Col}_j \leq 20\log_2^{(j-2)}(\mathrm{ID}_v)$.
	
	For the induction step, we will prove that $\mathtt{myPhase1Col}_{j+1} \leq 20\log_2^{(j-1)}(\mathrm{ID}_v)$.
	
	Recall that $\ell_{j}$ denotes the length of the binary representation of $\mathtt{myPhase1Col}_{j}$. By the induction hypothesis, we have $\mathtt{myPhase1Col}_{j} \leq 20\log_2^{(j-2)}(\mathrm{ID}_v)$, so we conclude that $\ell_j \leq \lfloor \log_2(20\log_2^{(j-2)}(\mathrm{ID}_v)) \rfloor + 1$. By applying \Cref{ColourDecreases} with $m=\log^*(\mathrm{ID}_v)+2$, it follows that $\mathtt{myPhase1Col}_{j+1} \in \{0,\ldots,8\ell_j+3\}$. Combining these facts, we get
	\begin{align*}
		\mathtt{myPhase1Col}_{j+1} & \leq 8\ell_j+3\\
		& \leq 8(\lfloor \log_2(20\log_2^{(j-2)}(\mathrm{ID}_v)) \rfloor + 1)+3\\
		& \leq 8(\log_2(20\log_2^{(j-2)}(\mathrm{ID}_v)) + 1)+3\\
		& = 8(\log_2(20)+\log_2(\log_2^{(j-2)}(\mathrm{ID}_v)) + 1)+3\\
		& = 8\log_2(20)+8(\log_2^{(j-1)}(\mathrm{ID}_v)) + 11\\
		& < 46 + 8(\log_2^{(j-1)}(\mathrm{ID}_v))
	\end{align*}
	Finally, from the definition of $\log^*$ and the fact that $\log^*(\mathrm{ID}_v) \geq 4$, we get that\\ $\log_2^{(\log^*(\mathrm{ID}_v)-1)}(\mathrm{ID}_v) > 1$, so $\log_2^{(\log^*(\mathrm{ID}_v)-2)}(\mathrm{ID}_v) > 2$, and so $\log_2^{(\log^*(\mathrm{ID}_v)-3)}(\mathrm{ID}_v) > 4$. In particular, since $j-1 \leq \log^*(\mathrm{ID}_v)-3$ from our induction hypothesis, we conclude that $\log_2^{(j-1)}(\mathrm{ID}_v) > 4$. Therefore,
	\begin{align*}
		\mathtt{myPhase1Col}_{j+1} & < 46 + 8(\log_2^{(j-1)}(\mathrm{ID}_v))\\
		& < 12(4) + 8(\log_2^{(j-1)}(\mathrm{ID}_v))\\
		& < 12(\log_2^{(j-1)}(\mathrm{ID}_v))+ 8(\log_2^{(j-1)}(\mathrm{ID}_v))\\
		& = 20(\log_2^{(j-1)}(\mathrm{ID}_v))
	\end{align*}
	This concludes the proof of the induction step, and the proof of \Cref{fastdecrease}.
	
	By \Cref{fastdecrease} with $j=\log^*(\mathrm{ID}_v)-1$, we get $\mathtt{myPhase1Col}_{\log^*(\mathrm{ID}_v)-1} \leq 20\log_2^{(\log^*(\mathrm{ID}_v)-3)}(\mathrm{ID}_v)$. By the definition of $\log^*$, we get that $\log_2^{(\log^*(\mathrm{ID}_v))}(\mathrm{ID}_v) \leq 1$, so $\log_2^{(\log^*(\mathrm{ID}_v)-1)}(\mathrm{ID}_v) \leq 2$, so $\log_2^{(\log^*(\mathrm{ID}_v)-2)}(\mathrm{ID}_v) \leq 4$, thus $\log_2^{(\log^*(\mathrm{ID}_v)-3)}(\mathrm{ID}_v) \leq 16$. Thus, we conclude that $\mathtt{myPhase1Col}_{\log^*(\mathrm{ID}_v)-1} \leq 20\log_2^{(\log^*(\mathrm{ID}_v)-3)}(\mathrm{ID}_v) \leq 320$.
	
	The length of the binary representation of 320 is 9, so by \Cref{ColourDecreases}, we get that $\mathtt{myPhase1Col}_{\log^*(\mathrm{ID}_v)} \leq 8(9)+3 = 75$. The length of the binary representation of 75 is 7, so by \Cref{ColourDecreases}, we get that $\mathtt{myPhase1Col}_{\log^*(\mathrm{ID}_v)+1} \leq 8(7)+3 = 59$. The length of the binary representation of 59 is 6, so by \Cref{ColourDecreases}, we get that $\mathtt{myPhase1Col}_{\log^*(\mathrm{ID}_v)+2} \leq 8(6)+3 = \maxCol$. Therefore, the value \textit{false} is assigned to $\mathtt{doPhase1}$ at line \ref{UpdateDP1} when Phase 1 is performed in round $\log^*(\mathrm{ID}_v)+2$, so Phase 1 is not performed in round $\log^*(\mathrm{ID}_v)+3$. By \Cref{Phase1Rounds}, it follows that Phase 1 is performed no more than $\log^*(\mathrm{ID}_v)$ times. This contradicts the assumption that $v$ performs Phase 1 at least $\log^*(\mathrm{ID}_v)+2$ times.
\end{proof}

We now shift our attention to Phase 2. Since we use a round-robin strategy that cycles through token values $0,\ldots,55$, and each node chooses its final colour and terminates as soon as the token value corresponding to its Phase 1 colour comes up, we see that each node performs Phase 2 for at most 56 rounds.
\begin{lemma}\label{Phase2Bound}
	Each node $v$ performs Phase 2 at most 56 times.
\end{lemma}
\begin{proof}
	Consider an arbitrary node $v$. Since the only terminate instruction is in Phase 2, and $v$ eventually stops executing Phase 1 (by \Cref{Phase1Bound}), we conclude that $v$ executes Phase 2 at least once. This means that the value of $\mathtt{doPhase1}$ was set to \textit{false} in some round $t \geq 2$ (either in Phase 0 or Phase 1), and that Phase 1 will never be executed after round $t$. From lines \ref{EndPhase0} and \ref{UpdateDP1}, this implies that $\mathtt{myPhase1Col} \in \{0,\ldots,\maxCol\} \cup \{\lmin,\lmax\} \cup \{\mathit{Pdone},\mathit{Cdone}\}$ at the end of round $t$, and that this value does not change after round $t$ (since $\mathtt{myPhase1Col}$ is never modified in Phase 2, i.e., in lines \ref{StartPhase2}--\ref{EndPhase2}). In each of the rounds $t+1,\ldots,t+56$, the variable $\mathtt{token}$ will take on a distinct value from the set $\{0,\ldots,55\}$, and the variable $\mathtt{myPhase1Col}$ keeps its fixed value from $\{0,\ldots,\maxCol\} \cup \{\lmin,\lmax\} \cup \{\mathit{Pdone},\mathit{Cdone}\}$. Therefore, the condition on line \ref{termCondition} will evaluate to true in one of the rounds $t+1,\ldots,t+56$, and so the terminate instruction will be performed by round $t+56$, which implies the desired result.
\end{proof}

In each node's execution of the algorithm: there is one initialization round, Phase 0 is performed once, Phase 1 is performed at most $\log^*(\mathrm{ID}_v)+1$ times (by \Cref{Phase1Bound}), and Phase 2 is performed at most 56 times (by \Cref{Phase2Bound}). This gives us the following upper bound on the running time of \EarlyStopCV at an arbitrary node $v$.
\begin{theorem}\label{terminationbound}
	At each node $v$, algorithm \EarlyStopCV terminates within $\log^*(\mathrm{ID}_v)+59$ rounds.
\end{theorem}

%--------------------------------------------------
%--------------------------------------------------
\subsection{Algorithm analysis: correctness}\label{sec:linecolouringcorrectness}
%--------------------------------------------------
%--------------------------------------------------

For any node $x$ and any $j \geq 2$, denote by $\mathtt{myPhase1Col}_{x,j}$ the value of $\mathtt{myPhase1Col}_{j}$ for node $x$, that is, the value of $\mathtt{myPhase1Col}$ at node $x$ at the end of round $j$. 

Suppose that $v$ sets its $\mathtt{parent}$ variable to some non-null value in round 2 (i.e., at line \ref{setParentA} or \ref{setParentB}). If node $w$ is $v$'s $A$-neighbour and $v$ sets $\mathtt{parent}$ to $A$, or, if node $w$ is $v$'s $B$-neighbour and $v$ sets $\mathtt{parent}$ to $B$, then we say that $v$'s $\mathtt{parent}$ variable \emph{corresponds} to $w$. Similarly, suppose that $v$ sets its $\mathtt{child}$ variable to some non-null value in round 2 (i.e., at line \ref{setChildB} or \ref{setChildA}). If node $w$ is $v$'s $A$-neighbour and $v$ sets $\mathtt{child}$ to $A$, or, if node $w$ is $v$'s $B$-neighbour and $v$ sets $\mathtt{child}$ to $B$, then we say that $v$'s $\mathtt{child}$ variable \emph{corresponds} to $w$.

Consider any three nodes $u,v,w$ such that $u$ and $w$ are $v$'s neighbours. We say that a node $v$ is a \emph{local minimum} if $\mathrm{ID}_v < \mathrm{ID}_u$ and $\mathrm{ID}_v < \mathrm{ID}_w$. We say that a node $v$ is a \emph{local maximum} if $\mathrm{ID}_v > \mathrm{ID}_u$ and $\mathrm{ID}_v > \mathrm{ID}_w$. Notice that if $v$ is neither a local minimum or a local maximum, then it must be the case that either $\mathrm{ID}_u > \mathrm{ID}_v > \mathrm{ID}_w$ or $\mathrm{ID}_u < \mathrm{ID}_v < \mathrm{ID}_w$. Using these facts and lines \ref{StartPhase0}--\ref{EndPhase0} of \EarlyStopCV, we make the observations stated in \Cref{partition}: essentially, at the end of round 2, the graph has been partitioned into subpaths that are separated by nodes that are local maxima/minima. Moreover, the parent/child pointers at each node in the subpaths are assigned in a consistent way so that the subpaths are directed (in particular, $\mathtt{parent}$ points to a node with smaller ID, and $\mathtt{child}$ points to a node with a larger ID).

\begin{proposition}\label{partition}
	Consider any two neighbouring processes $v,w$.
	\begin{itemize}
		\item $v$ is a local minimum if and only if $\mathtt{myPhase1Col}_{v,2} = \lmin$
		\item $v$ is a local maximum if and only if $\mathtt{myPhase1Col}_{v,2} = \lmax$
		\item If $v$ is not a local maximum/minimum and $w$ is a local minimum, then $v$'s $\mathtt{parent}$ variable corresponds to $w$ in all rounds $t \geq 3$.
		\item If $v$ is not a local maximum/minimum and $w$ is a local maximum, then $v$'s $\mathtt{child}$ variable corresponds to $w$ in all rounds $t \geq 3$.
		\item If both $v$ and $w$ are not a local maximum/minimum, then:
		\begin{itemize}
			\item If $\mathrm{ID}_v > \mathrm{ID}_w$, then $v$'s $\mathtt{parent}$ variable corresponds to $w$ in all rounds $t \geq 3$, and $w$'s $\mathtt{child}$ variable corresponds to $v$ in all rounds $t \geq 3$.
			\item If $\mathrm{ID}_v < \mathrm{ID}_w$, then $v$'s $\mathtt{child}$ variable corresponds to $w$ in all rounds $t \geq 3$, and $w$'s $\mathtt{parent}$ variable corresponds to $v$ in all rounds $t \geq 3$.
		\end{itemize}
	\end{itemize}
\end{proposition}

Each node is choosing two colours: its ``Phase 1'' colour (which it will use when performing Phase 2) and its ``final'' colour (which is its colour upon terminating the algorithm). When choosing these colours, it is important for the node to have accurate information about its neighbours' most recently chosen colours. Our next goal is to prove that this is the case.

For each node $x$, denote by $\mathrm{StartP2}_x$ the earliest round in which node $x$ performs Phase 2, which is well-defined integer greater than 2, since $x$ performs Phase 2 at least once after round 2 (otherwise $x$ would not terminate its execution, which would contradict \Cref{terminationbound}). If a node $x$ never performs Phase 1, then $\mathrm{StartP2}_x = 3$. Otherwise, if the number of times that a node $x$ performs Phase 1 is $m \geq 1$, then $\mathrm{StartP2}_x = m+3$ (by \Cref{Phase1Rounds}).

The next result states that, as long as a node $v$ is still performing Phase 1, it will have the most up-to-date values of its neighbours' Phase 1 colours.

\begin{proposition}\label{ABStoresP1Col}
	Consider the execution of the algorithm by any node $v$, and suppose that node $y$ is node $v$'s $A$-neighbour, and that node $z$ is node $v$'s $B$-neighbour. Then, for any round $j \in \{2,\ldots,\mathrm{StartP2}_v-1\}$, the value stored in node $v$'s variable $\mathtt{Acol.P1}$ at the end of round $j$ is equal to $\mathtt{myPhase1Col}_{y,j}$, and the value stored in node $v$'s variable $\mathtt{Bcol.P1}$ at the end of round $j$ is equal to $\mathtt{myPhase1Col}_{z,j}$.
\end{proposition}
\begin{proof}
	The proof is by induction on $j$. For the base case, consider $j=2$. When node $y$ executes line \ref{P0SendP1Col} in round 2, it sends the value of $\mathtt{myPhase1Col}_{y,2}$ since the value stored in its $\mathtt{myPhase1Col}$ does not change in lines \ref{P0SendP1Col}--\ref{EndPhase0}. Similarly, when node $z$ executes line \ref{P0SendP1Col} in round 2, it sends the value of $\mathtt{myPhase1Col}_{z,2}$. When node $v$ executes lines \ref{P0StoreA} and \ref{P0StoreB} in round 2, it stores the value of $y$'s sent message in $\mathtt{Acol.P1}$, and it stores the value of $z$'s sent message in $\mathtt{Bcol.P1}$, and then does not change these values at line \ref{EndPhase0}. It follows that, at the end of round 2, node $v$ has $\mathtt{Acol.P1}=\mathtt{myPhase1Col}_{y,2}$ and has $\mathtt{Bcol.P1}=\mathtt{myPhase1Col}_{z,2}$.
	
	As induction hypothesis, assume that for some $j \in \{2,\ldots,\mathrm{StartP2}_v-2\}$, we have that the value stored in node $v$'s variable $\mathtt{Acol.P1}$ at the end of round $j$ is equal to $\mathtt{myPhase1Col}_{y,j}$, and the value stored in node $v$'s variable $\mathtt{Bcol.P1}$ at the end of round $j$ is equal to $\mathtt{myPhase1Col}_{z,j}$.
	
	For the induction step, consider round $j+1 \in \{3,\ldots,\mathrm{StartP2}_v-1\}$. From the definition of $\mathrm{StartP2}_v$, it follows that $v$ performs Phase 1 in round $j+1$, i.e., it executes lines \ref{StartPhase1}--\ref{UpdateDP1}. We consider two cases based on the behaviour of node $y$ in round $j+1$:
	\begin{itemize}
		\item {\bf Suppose that $y$ performs Phase 1 in round $j+1$:}\\
		Note that, if $y$ performs Phase 1 in round $j+1$, then it sends (``P1'',$\mathtt{myPhase1Col}$) to $v$ at line \ref{P1SendP1} in round $j+1$. Since the value of $\mathtt{myPhase1Col}$ is not updated in lines \ref{P1SendP1}--\ref{UpdateDP1}, it follows that the value of $\mathtt{myPhase1Col}$ that node $y$ sent at line \ref{P1SendP1} in round $j+1$ is equal to $\mathtt{myPhase1Col}_{y,j+1}$. Therefore, when node $v$ executes line \ref{fromAP1} in round $j+1$, it receives $y$'s message and stores the value $\mathtt{myPhase1Col}_{y,j+1}$ in $\mathtt{Acol.P1}$, and then does not change the value of $\mathtt{Acol.P1}$ in any of the lines \ref{fromBP1}--\ref{UpdateDP1}. It follows that, at the end of round $j+1$, node $v$ has $\mathtt{Acol.P1}=\mathtt{myPhase1Col}_{y,j+1}$.
		\item {\bf Suppose that $y$ does not perform Phase 1 in round $j+1$:}\\
		From the code, we see that, if $y$ does not perform Phase 1 in round $j+1 \geq 3$, then node $y$ does not send any message of the form (``P1'', \textit{val}) in round $j+1$ (since $y$ does not send such a message in Phase 2 or after terminating). It follows that node $v$ does not receive any message of the form (``P1'', \textit{val}) from its neighbour $A$, so node $v$ does not update the value of $\mathtt{Acol.P1}$ at line \ref{fromAP1} in round $j+1$. Since the value of $\mathtt{Acol.P1}$ is not updated at any other line in \ref{StartPhase1}--\ref{UpdateDP1}, it follows that $v$ does not update $\mathtt{Acol.P1}$ during round $j+1$. Thus, we have shown that the value of node $v$'s $\mathtt{Acol.P1}$ at the end of round $j+1$ is the same as it was at the end of round $j$, which, by the induction hypothesis, is $\mathtt{myPhase1Col}_{y,j}$. However, if $y$ does not perform Phase 1 in round $j+1 \geq 3$, then the value of $y$'s variable $\mathtt{myPhase1Col}$ is not modified in round $j+1$ (since $y$ does not modify $\mathtt{myPhase1Col}$ in Phase 2 or after terminating). It follows that $\mathtt{myPhase1Col}_{y,j+1} = \mathtt{myPhase1Col}_{y,j}$, which concludes the proof that, at the end of round $j+1$, node $v$ has $\mathtt{Acol.P1}=\mathtt{myPhase1Col}_{y,j+1}$.
	\end{itemize}
	We proved in both cases that, at the end of round $j+1$, $\mathtt{Acol.P1}=\mathtt{myPhase1Col}_{y,j+1}$ at node $v$. The same case analysis about the behaviour of node $z$ in round $j+1$ proves that, at the end of round $j+1$, node $v$ has $\mathtt{Bcol.P1}=\mathtt{myPhase1Col}_{z,j+1}$, which concludes the induction.
\end{proof}

The following fact confirms that each call to \ChooseNewColour uses accurate information about the Phase 1 colours of a node's parent and child. This is true by \Cref{ABStoresP1Col} (which confirms that $\mathtt{Acol.P1}$ and $\mathtt{Bcol.P1}$ correctly store the Phase 1 colours of a node's $A$-neighbour and $B$-neighbour from the previous round), along with lines \ref{StartPhase1}--\ref{EndSwitch} (which make the appropriate call to \ChooseNewColour depending on whether a node's parent is $A$ or $B$).
\begin{corollary}\label{CorrectParentChild}
	Consider any node $v$ that performs Phase 1 in some round $t \geq 3$, and let $w$ be a neighbour of $v$. If node $v$'s $\mathtt{parent}$ variable corresponds to $w$, then node $v$'s call to \ChooseNewColour in round $t$ provides $\mathtt{myPhase1Col}_{w,t-1}$ as the second input parameter, i.e., in the execution of \ChooseNewColour in round $t$, the value of $\mathtt{ParentPhase1Col}$ is $\mathtt{myPhase1Col}_{w,t-1}$. If node $v$'s $\mathtt{child}$ variable corresponds to $w$, then node $v$'s call to \ChooseNewColour in round $t$ provides $\mathtt{myPhase1Col}_{w,t-1}$ as the third input parameter, i.e., in the execution of \ChooseNewColour in round $t$, the value of $\mathtt{ChildPhase1Col}$ is $\mathtt{myPhase1Col}_{w,t-1}$.
\end{corollary}

Once a node stops performing Phase 1, it never changes its Phase 1 colour ever again, and so the node has ``settled'' on a Phase 1 colour. It uses this settled Phase 1 colour when choosing its final colour in Phase 2. More importantly, if a node $w$ terminates its execution in an earlier round than node $v$, we will depend on the fact that node $v$ knows $w$'s settled Phase 1 colour while node $v$ is still choosing its Phase 1 colour. Denote by $\mathrm{SettledPhase1Col}_x$ the value of $\mathtt{myPhase1Col}$ at node $x$ at the end of round $\mathrm{StartP2}_x-1$. The following result follows from the fact that the value of $\mathtt{myPhase1Col}$ is never modified in Phase 2, so its value is fixed in all rounds from the start of Phase 2.

\begin{proposition}\label{SettledP1Col}
	For any node $x$ and any $t \geq \mathrm{StartP2}_x-1$, we have\\ $\mathtt{myPhase1Col}_{x,t} = \mathrm{SettledPhase1Col}_x$.
\end{proposition}

The main fact we want to prove is that every pair of neighbouring nodes $v,w$ will settle on different Phase 1 colours, as this will enable them to choose different colours from the set $\{0,1,2\}$ in Phase 2. Our approach is to first show that if $v,w$ start any round $t \geq 3$ with different Phase 1 colours, then they will have different Phase 1 colours at the end of round $t$, and then use this fact in an induction argument to prove that the colours they eventually settled on are different. There are several situations to consider depending on various factors: (1) if at least one of $v,w$ is a local maximum or local minimum, (2) the parent/child relationship between the two nodes, and, (3) when each node stops performing Phase 1 (i.e., in the same round, or if one stops earlier than the other).

One easy case to consider is when neither $v$ or $w$ perform Phase 1 in a round $t \geq 3$. Then neither node modifies their Phase 1 colour, so their chosen Phase 1 colours will still be different at the end of the round.

\begin{lemma}\label{neitherDoPhase1}
	Consider two neighbouring nodes $v,w$ and any round $t \geq 3$. Suppose that $\mathtt{myPhase1Col}_{v,t-1} \neq \mathtt{myPhase1Col}_{w,t-1}$. If neither $v$ or $w$ perform Phase 1 in round $t$, then we have $\mathtt{myPhase1Col}_{v,t} \neq \mathtt{myPhase1Col}_{w,t}$.
\end{lemma}

Another relatively easy case to consider is when at least one of $v$ or $w$ is a local maximum or minimum.

\begin{lemma}\label{atLeastOneMaxMin}
	Consider two neighbouring nodes $v,w$ and any round $t \geq 3$. Suppose that $\mathtt{myPhase1Col}_{v,t-1} \neq \mathtt{myPhase1Col}_{w,t-1}$. If at least one of $v$ or $w$ is a local maximum or a local minimum, then we have $\mathtt{myPhase1Col}_{v,t} \neq \mathtt{myPhase1Col}_{w,t}$.
\end{lemma}
\begin{proof}
	Without loss of generality, assume that $v$ is a local minimum. Then, by \Cref{partition}, it follows that $\mathtt{myPhase1Col}_{v,2} = \lmin$. By \Cref{NeverPhase1}, node $v$ never performs Phase 1, so $\mathtt{myPhase1Col}_{v,t} = \lmin$ for all $t \geq 3$ since $\mathtt{myPhase1Col}$ is never modified in Phase 2.
	
	Next, note that $w$ cannot also be a local minimum, since otherwise we would have $\mathrm{ID}_v < \mathrm{ID}_w$ and $\mathrm{ID}_v > \mathrm{ID}_w$. So, by \Cref{partition}, it follows that $\mathtt{myPhase1Col}_{w,2} \neq \lmin$. Consider two possibilities:
	\begin{itemize}
		\item {\bf Suppose that $\mathtt{myPhase1Col}_{w,2} \in \{\lmax\} \cup \{0,\ldots,\maxCol\}$}\\
		By \Cref{NeverPhase1}, node $w$ never performs Phase 1, so $\mathtt{myPhase1Col}_{w,t} \in \{\lmax\} \cup \{0,\ldots,\maxCol\}$ for all $t \geq 3$ since $\mathtt{myPhase1Col}$ is never modified in Phase 2. This proves that $\mathtt{myPhase1Col}_{v,t} \neq \mathtt{myPhase1Col}_{w,t}$.
		\item {\bf Suppose that $\mathtt{myPhase1Col}_{w,2} \not\in \{\lmax\} \cup \{0,\ldots,\maxCol\}$}\\
		Then we have that $\mathtt{myPhase1Col}_{w,2} \not\in \{\lmin,\lmax\} \cup \{0,\ldots,\maxCol\}$, so, by \Cref{NeverPhase1}, we know that $v$ performs Phase 1 at least once. For any $t \geq 3$, if $\mathtt{myPhase1Col}$ is modified, it must occur at line \ref{NewColourParentA} or \ref{NewColourParentB}, and its new value is the value returned by a call to \ChooseNewColour . By \Cref{NewColourOptions}, the value returned is a non-negative integer, or, one of $\{\mathit{Pdone},\mathit{Cdone}\}$. In particular, the value of $\mathtt{myPhase1Col}$ is not assigned the value $\lmin$ in any round $t \geq 3$. This proves that $\mathtt{myPhase1Col}_{v,t} \neq \mathtt{myPhase1Col}_{w,t}$.
	\end{itemize}
\end{proof}

So, for the remaining cases, we assume that neither $v$ or $w$ is a local maximum or minimum, and, we assume that at least one of $v$ or $w$ performs Phase 1 in round $t$. These cases are relatively easy to handle if $v$ and $w$ both start round $t$ with large integer colours, since we will be able to apply an argument similar to the original Cole and Vishkin proof. However, additional care must be taken now that the algorithm can choose one of the special colours \textit{Pdone} or \textit{Cdone}, and the following fact will help us argue about situations that involve these colours.

\begin{proposition}\label{ParentPdone}
	Suppose that \ChooseNewColour  is called in some round $t \geq 3$ by some node $w$ whose $\mathtt{parent}$ variable corresponds to a node $v$. If this call to \ChooseNewColour returns the value \textit{Pdone}, then $\mathtt{myPhase1Col}_{v,t-1} \in \{0,\ldots,\maxCol\}$.
\end{proposition}
\begin{proof}
	Since node $w$'s $\mathtt{parent}$ variable corresponds to node $v$, it follows from \Cref{CorrectParentChild} that $\mathtt{ParentPhase1Col}$ is equal to $\mathtt{myPhase1Col}_{v,t-1}$.
	
	Suppose that $w$'s call to \ChooseNewColour in round $t$ returns the value \textit{Pdone}. From the code of \ChooseNewColour, this return value must have been chosen at line \ref{ReturnPdone}, which means that the condition on line \ref{IfPSettled} was satisfied. It follows that $\mathtt{ParentPhase1Col} \in \{0,\ldots,\maxCol\}$, and so $\mathtt{myPhase1Col}_{v,t-1} \in \{0,\ldots,\maxCol\}$.
\end{proof}

Next, we consider the situation where exactly one of $v$ or $w$ performs Phase 1 in round $t$. We designed the algorithm so that if a node $v$ already has a colour in the range $\{0,\ldots,\maxCol\}$ at the start of round $t-1$, then its child $w$ concludes that $v$ will not modify its Phase 1 colour in round $t-1$, so $w$ will choose the colour \textit{Pdone} in round $t-1$. The following fact proves that this is the only situation in which $w$ will choose colour $\textit{Pdone}$, i.e., if $v$ performs Phase 1 in round $t-1$, then it is not possible that its child $w$ chose colour $\textit{Pdone}$ in round $t-1$.

\begin{proposition}\label{ChildNotPdone}
	Consider two neighbouring nodes $v,w$ and suppose that neither $v$ or $w$ is a local maximum or minimum. Consider any round $t \geq 3$ in which $w$ does not perform Phase 1. If $v$'s $\mathtt{child}$ variable corresponds to node $w$ and $v$ performs Phase 1 in round $t-1$, then $\mathtt{myPhase1Col}_{w,t-1} \neq \textit{Pdone}$.
\end{proposition}
\begin{proof}
	First, consider the case where $t=3$. Then in round $t-1=2$, we see from lines \ref{StartPhase0}--\ref{EndPhase0} that $\mathtt{myPhase1Col}$ is never assigned the value of \textit{Pdone}, so we conclude that $\mathtt{myPhase1Col}_{w,t-1} \neq \textit{Pdone}$.
	
	Suppose that $t \geq 4$ and that $v$ performs Phase 1 in round $t-1$. If $v$'s $\mathtt{child}$ variable corresponds to node $w$, then from lines \ref{StartPhase0}--\ref{EndPhase0}, it must be the case that $\mathrm{ID}_w > \mathrm{ID}_v$. Since we assumed that $w$ is not a local maximum, we conclude from \Cref{partition} that node $w$'s $\mathtt{parent}$ variable corresponds to node $v$. 
	
	To obtain a contradiction, assume that $\mathtt{myPhase1Col}_{w,t-1} = \textit{Pdone}$. From the assumption that $w$ does not perform Phase 1 in round $t$, it follows that $\mathrm{StartP2}_w \leq t$, and so $\mathrm{StartP2}_w-1 \leq t-1$. By \Cref{SettledP1Col} and the definition of $\mathrm{SettledPhase1Col}_w$, we get that $\mathtt{myPhase1Col}_{w,\mathrm{StartP2}_w-1} = \mathrm{SettledPhase1Col}_w = \mathtt{myPhase1Col}_{w,t-1} = \textit{Pdone}$. Thus, in round $\mathrm{StartP2}_w-1$ of $w$'s execution, the value \textit{Pdone} was assigned to $\mathtt{myPhase1Col}$. But for this to happen, the call to \ChooseNewColour  must return \textit{Pdone} in round $\mathrm{StartP2}_w-1$ of $w$'s execution, so, by \Cref{ParentPdone}, it follows that $\mathtt{myPhase1Col}_{v,\mathrm{StartP2}_w-2} \in \{0,\ldots,\maxCol\}$. However, this means that $\mathtt{doPhase1}$ is \textit{false} in round $\mathrm{StartP2}_w-2$ of $v$'s execution, which implies that $\mathrm{StartP2}_v \leq \mathrm{StartP2}_w-1$. In other words, node $v$ does not perform Phase 1 in or after round $\mathrm{StartP2}_w-1$. But we showed above that $\mathrm{StartP2}_w-1 \leq t-1$, so we conclude that node $v$ does not perform Phase 1 in or after round $t-1$. This contradicts the fact that node $v$ performs Phase 1 in round $t-1$. Therefore, our assumption was incorrect, and we conclude that $\mathtt{myPhase1Col}_{w,t-1} \neq \textit{Pdone}$.
\end{proof}

We can now prove that, for any neighbouring nodes $v$ and $w$ that start some round $t \geq 3$ with different colours, if exactly one of them performs Phase 1 in round $t$, then it will choose a new Phase 1 colour in such a way that both nodes have different Phase 1 colours at the end of the round.

\begin{lemma}\label{ExactlyOneP1}
	Consider two neighbouring nodes $v,w$ and any round $t \geq 3$. Suppose that $\mathtt{myPhase1Col}_{v,t-1} \neq \mathtt{myPhase1Col}_{w,t-1}$, and that neither $v$ or $w$ is a local maximum or a local minimum. If exactly one of $v$ or $w$ perform Phase 1 in round $t$, then we have $\mathtt{myPhase1Col}_{v,t} \neq \mathtt{myPhase1Col}_{w,t}$.
\end{lemma}
\begin{proof}
	Without loss of generality, assume that $\mathrm{StartP2}_v \geq \mathrm{StartP2}_w$. In particular, we assume that node $v$ performs Phase 1 in round $t$ and that node $w$ does not perform Phase 1 in round $t$. 
	
	First, we will determine the value of $\mathtt{myPhase1Col}_{v,t-1}$. The fact that node $v$ performs Phase 1 in round $t \geq 3$ means that $\mathtt{doPhase1}$ must have been set to \textit{true} in all rounds $\{2,\ldots,t-1\}$. In particular, it follows that $\mathtt{myPhase1Col}_{v,t-1} \not\in \{0,\ldots,\maxCol\} \cup \{\lmin,\lmax\} \cup \{\mathit{Pdone},\mathit{Cdone}\}$. Thus, we conclude that $\mathtt{myPhase1Col}_{v,t-1}$ is an integer and $\mathtt{myPhase1Col}_{v,t-1} \geq 52$.
	
	Next, we will determine the value of $\mathtt{myPhase1Col}_{w,t-1}$. The fact that node $w$ does not perform Phase 1 in round $t \geq 3$ means that $t \geq \mathrm{StartP2}_w$, which implies that $t-1 \geq \mathrm{StartP2}_w-1$. By \Cref{SettledP1Col}, we conclude that\\
	$\mathtt{myPhase1Col}_{w,t-1} = \mathrm{SettledPhase1Col}_w = \mathtt{myPhase1Col}_{w,\mathrm{StartP2}_w-1}$. From the definition of $\mathrm{StartP2}_w$, we know that $w$ performs Phase 2 for the first time in round $\mathrm{StartP2}_w$, which means that $\mathtt{doPhase1}$ must have been set to \textit{false} in round $\mathrm{StartP2}_w-1$. From lines \ref{EndPhase0} and \ref{UpdateDP1}, we conclude that $\mathtt{myPhase1Col}_{w,\mathrm{StartP2}_w-1} \in \{0,\ldots,\maxCol\} \cup \{\lmin,\lmax\} \cup \{\mathit{Pdone},\mathit{Cdone}\}$. Thus, we have shown that $\mathtt{myPhase1Col}_{w,t-1} \in \{0,\ldots,\maxCol\} \cup \{\lmin,\lmax\} \cup \{\mathit{Pdone},\mathit{Cdone}\}$.
	
	Next, we will determine the value of $\mathtt{myPhase1Col}_{v,t}$ and prove that it is not equal to $\mathtt{myPhase1Col}_{w,t-1}$. Since $v$ performs Phase 1 in round $t \geq 3$, we conclude that $t \in \{3,\ldots,\mathrm{StartP2}_v-1\}$, so $t-1 \in \{2,\ldots,\mathrm{StartP2}_v-2\}$. By \Cref{ABStoresP1Col}, the value of $\mathtt{myPhase1Col}_{w,t-1}$ is stored in one of $\mathtt{Acol.P1}$ or $\mathtt{Bcol.P1}$ at the end of round $t-1$, so the same is true at the start of round $t$. 
	Therefore, when \ChooseNewColour  is executed at line \ref{NewColourParentA} or line \ref{NewColourParentB} in round $t$, the value of $\mathtt{myPhase1Col}_{w,t-1}$ will be passed in as the second parameter (if $w$ corresponds to $v$'s $\mathtt{parent}$) or third parameter (if $w$ corresponds to $v$'s $\mathtt{child}$). Moreover, when \ChooseNewColour  is executed at line \ref{NewColourParentA} or line \ref{NewColourParentB} in round $t$, the first input parameter is the value of $\mathtt{myPhase1Col}$ at the end of the previous round, i.e., $\mathtt{myPhase1Col}_{v,t-1}$. So, to determine the value of $\mathtt{myPhase1Col}_{v,t}$, we consider the execution of \ChooseNewColour , and separately consider cases based on the value of $\mathtt{myPhase1Col}_{w,t-1}$ (which, from the previous paragraph, comes from the set $\{0,\ldots,\maxCol\} \cup \{\lmin,\lmax\} \cup \{\mathit{Pdone},\mathit{Cdone}\}$):
	\begin{itemize}
		\item {\bf Suppose that $\mathtt{myPhase1Col}_{w,t-1} \in \{0,\ldots,\maxCol\}$}\\
		Since $\mathtt{myPhase1Col}_{w,t-1}$ is passed in as either the second or third input parameter, it follows that one of the conditions on lines \ref{IfPSettled} or \ref{IfCSettled} of \ChooseNewColour  will evaluate to true, so $\mathtt{myPhase1Col}_{v,t} \in \{\mathit{Pdone},\mathit{Cdone}\}$. Since $\mathit{Pdone}$ and $\mathit{Cdone}$ were defined in such a way that $\mathit{Pdone},\mathit{Cdone} \not\in \{0,\ldots,\maxCol\}$, it follows that $\mathtt{myPhase1Col}_{v,t} \neq \mathtt{myPhase1Col}_{w,t-1}$.
		\item {\bf Suppose that $\mathtt{myPhase1Col}_{w,t-1} \not\in \{0,\ldots,\maxCol\}$}\\
		Since we are assuming that $\mathtt{myPhase1Col}_{w,t-1} \not\in \{0,\ldots,\maxCol\}$, it follows that\\
		$\mathtt{myPhase1Col}_{w,t-1} \in \{\lmin,\lmax\} \cup \{\mathit{Pdone},\mathit{Cdone}\}$.
		
		First, suppose $w$ corresponds to $v$'s $\mathtt{parent}$. This means that $\mathtt{myPhase1Col}_{w,t-1}$ is passed as the $2^{nd}$ input parameter of \ChooseNewColour, i.e., $\mathtt{ParentPhase1Col} = \mathtt{myPhase1Col}_{w,t-1}$ in the execution of \ChooseNewColour. It follows that $\mathtt{ParentPhase1Col} \in \{\lmin,\lmax\} \cup \{\mathit{Pdone},\mathit{Cdone}\}$, so the condition on line \ref{IfPdone} of \ChooseNewColour evaluates to true, so $\mathtt{myPhase1Col}_{v,t}$ gets the return value of $\CVChoice (\mathtt{myPhase1Col}_{v,t-1},0)$. We proved earlier that $\mathtt{myPhase1Col}_{v,t-1} \geq 52$, so both inputs into \CVChoice  are non-negative integers. By \Cref{CVChoiceOutput}, since the length of the binary representation of 0 is 1, we get that the value returned by $\CVChoice (\mathtt{myPhase1Col}_{v,t-1},0)$ is in $\{0,\ldots,11\}$. Therefore, we have shown that $\mathtt{myPhase1Col}_{v,t} \in \{0,\ldots,11\}$. Since $\mathtt{myPhase1Col}_{w,t-1} \in \{\lmin,\lmax\} \cup \{\mathit{Pdone},\mathit{Cdone}\}$, it follows that $\mathtt{myPhase1Col}_{v,t} \neq \mathtt{myPhase1Col}_{w,t-1}$ .
		
		Next, suppose that $w$ corresponds to $v$'s $\mathtt{child}$, so $\mathtt{myPhase1Col}_{w,t-1}$ is passed in as the third input parameter of \ChooseNewColour . In particular, consider the execution of \ChooseNewColour  with $\mathtt{ChildPhase1Col} = \mathtt{myPhase1Col}_{w,t-1}$. Since we assumed that $w$ does not perform Phase 1 in round $t$, then by \Cref{ChildNotPdone}, we conclude that $\mathtt{myPhase1Col}_{w,t-1} \neq \mathit{Pdone}$. Thus, $\mathtt{ChildPhase1Col}=\mathtt{myPhase1Col}_{w,t-1} \in \{\lmin,\lmax,\mathit{Cdone}\}$. One consequence is that the condition on line \ref{IfCSettled} of \ChooseNewColour  does not evaluate to true. Therefore, the return value is set at one of lines \ref{ReturnPdone}, \ref{doCVWith0}, or \ref{doCV}. In all three cases, the return value is not in $\{\lmin,\lmax,\mathit{Cdone}\}$, which, in particular, implies that the return value is not equal to $\mathtt{myPhase1Col}_{w,t-1}$. Thus, we have shown that $\mathtt{myPhase1Col}_{v,t} \neq \mathtt{myPhase1Col}_{w,t-1}$.
	\end{itemize}
	
	Finally, we will determine the value of $\mathtt{myPhase1Col}_{w,t}$. The fact that node $w$ does not perform Phase 1 in round $t \geq 3$ means that $t \geq \mathrm{StartP2}_w$, which implies that $t-1 \geq \mathrm{StartP2}_w-1$. By \Cref{SettledP1Col}, we conclude that $\mathtt{myPhase1Col}_{w,t} = \mathrm{SettledPhase1Col}_w = \mathtt{myPhase1Col}_{w,t-1}$. As we showed above that $\mathtt{myPhase1Col}_{v,t} \neq \mathtt{myPhase1Col}_{w,t-1}$, we conclude that $\mathtt{myPhase1Col}_{v,t} \neq \mathtt{myPhase1Col}_{w,t}$.
\end{proof}

Next, we consider the situation where both $v$ and $w$ perform Phase 1 in round $t$. First, we prove the following useful fact that will help us argue about simultaneous executions of \CVChoice by $v$ and $w$.
\begin{lemma}\label{CVChoiceDifferent}
	For any three non-negative integers $a,b,c$ such that $a \neq b$ and $b \neq c$, the executions of $\CVChoice (a,b)$ and $\CVChoice (b,c)$ return different integer values.
\end{lemma}
\begin{proof}
	In the execution of $\CVChoice (a,b)$, denote by $a_{\mathrm{SF}}$ the string stored in $\mathtt{MyString}$ after executing line \ref{EncodeMS}, and denote by $b_{\mathrm{SF}}$ the string stored in $\mathtt{OtherString}$ after executing line \ref{EncodeOS}. From the fact that $a \neq b$ and by \Cref{SuffixFree}, it follows that $a_{\mathrm{SF}} \neq b_{\mathrm{SF}}$, and it follows that there exists at least one non-negative integer $i$ such that $a_{\mathrm{SF}}[i] \neq b_{\mathrm{SF}}[i]$. Denote by $i_{a,b}$ the least such integer $i$, and note that $i_{a,b}$ is the value stored in the variable $\mathtt{i}$ after executing line \ref{FindDiffIndex}. Denote by $S_{a,b}$ the binary representation of $i_{a,b}$. Denote by $\mathit{NS}_{a,b}$ the value stored in $\mathtt{newString}$ after executing line \ref{SetCVChoice}, which is the string $S_{a,b} \cdot a_{\mathrm{SF}}[i_{a,b}]$.
	
	In the execution of $\CVChoice (b,c)$, denote by $b_{\mathrm{SF}}$ the string stored in $\mathtt{MyString}$ after executing line \ref{EncodeMS}, and denote by $c_{\mathrm{SF}}$ the string stored in $\mathtt{OtherString}$ after executing line \ref{EncodeOS}. From the fact that $b \neq c$ and by \Cref{SuffixFree}, it follows that $b_{\mathrm{SF}} \neq c_{\mathrm{SF}}$, and it follows that there exists at least one non-negative integer $i$ such that $b_{\mathrm{SF}}[i] \neq c_{\mathrm{SF}}[i]$. Denote by $i_{b,c}$ the least such integer $i$, and note that $i_{b,c}$ is the value stored in the variable $\mathtt{i}$ after executing line \ref{FindDiffIndex}. Denote by $S_{b,c}$ the binary representation of $i_{b,c}$. Denote by $\mathit{NS}_{b,c}$ the value stored in $\mathtt{newString}$ after executing line \ref{SetCVChoice}, which is the string $S_{b,c} \cdot b_{\mathrm{SF}}[i_{b,c}]$.
	
	It suffices to prove that $\mathit{NS}_{a,b} \neq \mathit{NS}_{b,c}$, since this will imply that the integer values returned at line \ref{ReturnCVChoice} of each execution are different. There are two cases to consider:
	\begin{itemize}
		\item {\bf Suppose that $i_{a,b} \neq i_{b,c}$.}\\
		It follows that $S_{a,b} \neq S_{b,c}$, so the strings $S_{a,b} \cdot a_{\mathrm{SF}}[i_{a,b}]$ and $S_{b,c} \cdot b_{\mathrm{SF}}[i_{b,c}]$ cannot be equal (regardless of whether or not the rightmost bit $a_{\mathrm{SF}}[i_{a,b}]$ is equal to the rightmost bit $b_{\mathrm{SF}}[i_{b,c}]$). This proves that $\mathit{NS}_{a,b} \neq \mathit{NS}_{b,c}$.
		\item {\bf Suppose that $i_{a,b} = i_{b,c}$.}\\
		From the definition of $i_{a,b}$, we know that $a_{\mathrm{SF}}[i_{a,b}] \neq b_{\mathrm{SF}}[i_{a,b}]$. Since $i_{a,b} = i_{b,c}$, we conclude that $a_{\mathrm{SF}}[i_{a,b}] \neq b_{\mathrm{SF}}[i_{b,c}]$. Thus, the rightmost bit of $\mathit{NS}_{a,b}$ differs from the rightmost bit of $\mathit{NS}_{b,c}$, which proves that $\mathit{NS}_{a,b} \neq \mathit{NS}_{b,c}$.
	\end{itemize}
	In all cases, we proved that $\mathit{NS}_{a,b} \neq \mathit{NS}_{b,c}$, as desired.
\end{proof}

Now we consider two neighbouring nodes that both perform Phase 1 in any round $t \geq 3$. We assume that each of $v$ and $w$ start the round with Phase 1 colours that are different than their neighbours' Phase 1 colours, and prove that $v$ and $w$ choose different Phase 1 colours in round $t$.
\begin{lemma}\label{BothP1}
	Consider two neighbouring nodes $v,w$ and any round $t \geq 3$, let $\hat{v}$ be $v$'s other neighbour, and let $\hat{w}$ be $w$'s other neighbour. Suppose that neither $v$ or $w$ is a local maximum or a local minimum. Further, suppose that $\mathtt{myPhase1Col}_{v,t-1} \neq \mathtt{myPhase1Col}_{w,t-1}$, $\mathtt{myPhase1Col}_{v,t-1} \neq \mathtt{myPhase1Col}_{\hat{v},t-1}$, and $\mathtt{myPhase1Col}_{w,t-1} \neq \mathtt{myPhase1Col}_{\hat{w},t-1}$. If both $v$ and $w$ perform Phase 1 in round $t$, then we have $\mathtt{myPhase1Col}_{v,t} \neq \mathtt{myPhase1Col}_{w,t}$.
\end{lemma}
\begin{proof}
	The fact that $v$ performs Phase 1 in round $t$ means that, in each of the rounds $\{2,\ldots,t-1\}$, the variable $\mathtt{doPhase1}$ was assigned the value \textit{true}. From lines \ref{EndPhase0} and \ref{UpdateDP1} of \EarlyStopCV, it follows that $\mathtt{myPhase1Col}_{v,j} \not\in \{0,\ldots,\maxCol\} \cup \{\lmin,\lmax\} \cup \{\mathit{Pdone},\mathit{Cdone}\}$ for each $j \in \{2,\ldots,t-1\}$. In particular, we conclude that $\mathtt{myPhase1Col}_{v,t-1}$ is an integer greater than or equal to 52. Using the same argument, we also conclude that $\mathtt{myPhase1Col}_{w,t-1}$ is an integer greater than or equal to 52.
	
	Without loss of generality, assume that $\mathrm{ID}_v > \mathrm{ID}_w$. Since neither $v$ or $w$ is a local maximum or a local minimum, \Cref{partition} tells us that $v$'s $\mathtt{parent}$ variable corresponds to $w$, and $w$'s $\mathtt{child}$ variable corresponds to $v$. Moreover, since $w$ is not a local minimum, we know that $\mathrm{ID}_w > \mathrm{ID}_{\hat{w}}$, so $w$'s $\mathtt{parent}$ variable corresponds to $\hat{w}$. From these facts, along with \Cref{ABStoresP1Col}, we conclude:
	\begin{itemize}
		\item In round $t$, node $v$ calls \ChooseNewColour with first input parameter equal to $\mathtt{myPhase1Col}_{v,t-1}$ and second input parameter equal to $\mathtt{myPhase1Col}_{w,t-1}$. In particular, this means that, during the execution of \ChooseNewColour at $v$ in round $t$, we have $\mathtt{myPhase1Col}$ equal to $\mathtt{myPhase1Col}_{v,t-1}$ and $\mathtt{ParentPhase1Col}$ equal to $\mathtt{myPhase1Col}_{w,t-1}$.
		\item In round $t$, node $w$ calls \ChooseNewColour with first input value equal to $\mathtt{myPhase1Col}_{w,t-1}$, second input value equal to $\mathtt{myPhase1Col}_{\hat{w},t-1}$, and third input value equal to $\mathtt{myPhase1Col}_{v,t-1}$. In particular, this means that, during the execution of \ChooseNewColour at $w$ in round $t$, we have $\mathtt{myPhase1Col}$ equal to $\mathtt{myPhase1Col}_{w,t-1}$, $\mathtt{ParentPhase1Col}$ equal to $\mathtt{myPhase1Col}_{\hat{w},t-1}$, and $\mathtt{ChildPhase1Col}$ equal to $\mathtt{myPhase1Col}_{v,t-1}$.
	\end{itemize}
	We proved above that both $\mathtt{myPhase1Col}_{v,t-1}$ and $\mathtt{myPhase1Col}_{w,t-1}$ are integers that are greater than or equal to 52, so we can immediately deduce certain facts about the return value of \ChooseNewColour at $v$ and $w$:
	\begin{itemize}
		\item The condition on line \ref{IfCSettled} in $w$'s execution of \ChooseNewColour evaluates to false, so $w$'s execution does not return \textit{Cdone} in round $t$.
		\item The condition on line \ref{IfPSettled} in $v$'s execution of \ChooseNewColour evaluates to false, so $v$'s execution does not return \textit{Pdone} in round $t$.
		\item The condition on line \ref{IfPdone} in $v$'s execution of \ChooseNewColour evaluates to false, so $v$'s execution does not perform line \ref{doCVWith0} of \ChooseNewColour in round $t$.
	\end{itemize}
	We consider the remaining possible outcomes of $v$'s and $w$'s executions of \ChooseNewColour:
	\begin{itemize}
		\item {\bf Suppose that $w$'s execution of \ChooseNewColour returns \textit{Pdone}.}\\
		We proved above that $v$'s execution of \ChooseNewColour does not return \textit{Pdone}, so it follows that $\mathtt{myPhase1Col}_{v,t} \neq \mathtt{myPhase1Col}_{w,t}$.
		\item {\bf Suppose that $v$'s execution of \ChooseNewColour returns \textit{Cdone}.}\\
		We proved above that $w$'s execution of \ChooseNewColour does not return \textit{Cdone}, so it follows that $\mathtt{myPhase1Col}_{v,t} \neq \mathtt{myPhase1Col}_{w,t}$.
		\item {\bf Suppose that $v$'s execution does not return \textit{Cdone} and that $w$'s execution does not return \textit{Pdone}.}\\
		We proved above that $v$'s execution of \ChooseNewColour does not return \textit{Pdone} and does not perform line \ref{doCVWith0}, so the remaining possibility is that the return value of $v$'s execution of \ChooseNewColour is set at line \ref{doCV}. In particular, $v$'s execution returns the value returned by $\CVChoice (\mathtt{myPhase1Col}_{v,t-1},\mathtt{myPhase1Col}_{w,t-1})$. We proved above that $w$'s execution of \ChooseNewColour does not return \textit{Cdone}, so the only remaining possibilities are that the return value of $w$'s execution of \ChooseNewColour is set at line \ref{doCVWith0} or \ref{doCV}. In particular, $w$'s execution returns the value that was returned by $\CVChoice (\mathtt{myPhase1Col}_{w,t-1},z)$ for some non-negative integer $z \neq \mathtt{myPhase1Col}_{w,t-1}$ (since $z$ is either 0 or $\mathtt{myPhase1Col}_{\hat{w},t-1}$). By \Cref{CVChoiceDifferent} with $x = \mathtt{myPhase1Col}_{v,t-1}$ and $y=\mathtt{myPhase1Col}_{w,t-1}$, we conclude that the executions of \ChooseNewColour by $v$ and $w$ return different integer values. Therefore, we get that $\mathtt{myPhase1Col}_{v,t} \neq \mathtt{myPhase1Col}_{w,t}$.
	\end{itemize}
	In all cases, we proved that $\mathtt{myPhase1Col}_{v,t} \neq \mathtt{myPhase1Col}_{w,t}$, as desired.
\end{proof}

We now combine the above results to show that, in every round $t \geq 2$, no two neighbouring nodes have the same Phase 1 colour.

\begin{lemma}\label{NbrDiffP1Cols}
	For any neighbouring nodes $v,w$ and any $t \geq 2$, we have\\ $\mathtt{myPhase1Col}_{v,t} \neq \mathtt{myPhase1Col}_{w,t}$.
\end{lemma}
\begin{proof}
	We prove the statement by induction on $t$.
	
	For the base case, consider $t=2$ and any neighbouring nodes $v,w$. From lines \ref{StartPhase0}--\ref{EndPhase0}, we see that $\mathtt{myPhase1Col}_{v,2} \in \{\lmin,\lmax,\mathrm{ID}_v\}$ and $\mathtt{myPhase1Col}_{w,2} \in \{\lmin,\lmax,\mathrm{ID}_w\}$. If $\mathtt{myPhase1Col}_{v,2} = \lmin$, then $\mathtt{myPhase1Col}_{w,2} \neq \lmin$ since it is not possible that neighbouring nodes are both local minima. If $\mathtt{myPhase1Col}_{v,2} = \lmax$, then $\mathtt{myPhase1Col}_{w,2} \neq \lmax$ since it is not possible that neighbouring nodes are both local maxima. If $\mathtt{myPhase1Col}_{v,2} = \mathrm{ID}_v$, then our definition of $\lmin$ and $\lmax$ guarantees that $\mathrm{ID}_v \not\in \{\lmin,\lmax\}$, and our assumption that the initial node labels form a proper colouring guarantees that $\mathrm{ID}_v \neq \mathrm{ID}_w$, so $\mathtt{myPhase1Col}_{v,2} \not\in \{\lmin,\lmax,\mathrm{ID}_w\}$, and so $\mathtt{myPhase1Col}_{v,2} \neq \mathtt{myPhase1Col}_{w,2}$.
	
	As induction hypothesis, assume that for any neighbouring nodes $v,w$ and some $t-1 \geq 2$, we have $\mathtt{myPhase1Col}_{v,t-1} \neq \mathtt{myPhase1Col}_{w,t-1}$. For the induction step, we set out to prove that $\mathtt{myPhase1Col}_{v,t} \neq \mathtt{myPhase1Col}_{w,t}$. Consider the following exhaustive list of possible cases:
	\begin{itemize}
		\item {\bf Suppose that at least one of $v$ or $w$ is a local maximum or local minimum.} This case is handled by \Cref{atLeastOneMaxMin}.
		\item {\bf Suppose that neither $v$ or $w$ performs Phase 1 in round $t$.} This case is handled by \Cref{neitherDoPhase1}.
		\item {\bf Suppose that neither $v$ or $w$ is a local maximum or local minimum, and that exactly one of $v$ or $w$ performs Phase 1 in round $t$.} This case is handled by \Cref{ExactlyOneP1}.
		\item {\bf Suppose that neither $v$ or $w$ is a local maximum or local minimum, and that both $v$ and $w$ perform Phase 1 in round $t$.} This case is handled by \Cref{BothP1}.
	\end{itemize}
	In all cases, we proved that $\mathtt{myPhase1Col}_{v,t} \neq \mathtt{myPhase1Col}_{w,t}$, which concludes the induction.
\end{proof}

Using the fact that any neighbouring nodes $v,w$ have different Phase 1 colours in each round $t \geq 3$, we now show that they will settle on different Phase 1 colours, i.e., they each start Phase 2 with different Phase 1 colours.
\begin{lemma}\label{DifferentSettledP1}
	For any two neighbouring nodes $v$ and $w$, we have $\mathrm{SettledPhase1Col}_v \neq \mathrm{SettledPhase1Col}_w$.
\end{lemma}
\begin{proof}
	Without loss of generality, assume $\mathrm{StartP2}_v \geq \mathrm{StartP2}_w$. From the definition of $\mathrm{SettledPhase1Col}$, we have $\mathrm{SettledPhase1Col}_v = \mathtt{myPhase1Col}_{v,\mathrm{StartP2}_v-1}$. By \Cref{NbrDiffP1Cols}, we get that $\mathtt{myPhase1Col}_{v,\mathrm{StartP2}_v-1} \neq \mathtt{myPhase1Col}_{w,\mathrm{StartP2}_v-1}$. However, since $\mathrm{StartP2}_v-1 \geq \mathrm{StartP2}_w-1$, we conclude from \Cref{SettledP1Col} that $\mathtt{myPhase1Col}_{w,\mathrm{StartP2}_v-1} = \mathrm{SettledPhase1Col}_w$. To summarize, we have
	$\mathrm{SettledPhase1Col}_v = \mathtt{myPhase1Col}_{v,\mathrm{StartP2}_v-1} \neq \mathtt{myPhase1Col}_{w,\mathrm{StartP2}_v-1} = \mathrm{SettledPhase1Col}_w$, which implies that $\mathrm{SettledPhase1Col}_v \neq \mathrm{SettledPhase1Col}_w$.
\end{proof}

Using the code, we next observe that each node chooses a final colour and advertises this final colour to its neighbours exactly once.
\begin{proposition}\label{FinalExactlyOnce}
	Each node performs lines \ref{ChooseFinal}--\ref{terminate} exactly once.
\end{proposition}
\begin{proof}
	Consider an arbitrary node $v$. Node $v$ performs lines \ref{ChooseFinal}--\ref{terminate} at most once, since $v$'s execution terminates after performing line \ref{terminate} for the first time. Node $v$ performs lines \ref{ChooseFinal}--\ref{terminate} at least once since line \ref{terminate} is the only termination instruction, and \Cref{terminationbound} tells us that $v$ eventually terminates its execution.
\end{proof}

We now prove that, when choosing its final colour, a node will always be able to choose a value from $\{0,1,2\}$, which proves that the algorithm uses at most 3 colours.

\begin{proposition}\label{ValidFinalCol}
	After performing line \ref{ChooseFinal}, the value of $\mathtt{myFinalCol}$ is a value in $\{0,1,2\}$.
\end{proposition}
\begin{proof}
	The set $\{\mathtt{Acol.final},\mathtt{Bcol.final}\}$ can contain at most two values from $\{0,1,2\}$, so the set $\{0,1,2\} \setminus \{\mathtt{Acol.final},\mathtt{Bcol.final}\}$ is non-empty when line \ref{ChooseFinal} is performed.
\end{proof}

From the fact that every two neighbouring nodes will settle on different Phase 1 colours, we can now confirm the round-robin nature of Phase 2, i.e., we can guarantee that no two neighbouring nodes will choose their final colour in the same round.

\begin{lemma}\label{P2DifferentRounds}
	If two neighbouring nodes $v,w$ have\\ $\mathrm{SettledPhase1Col}_v \neq \mathrm{SettledPhase1Col}_w$, then nodes $v$ and $w$ execute line \ref{ChooseFinal} in different rounds.
\end{lemma}
\begin{proof}
	We prove the contrapositive of the given statement. Suppose that $v$ and $w$ execute line \ref{ChooseFinal} in the same round $r$. Due to the assumption that all nodes start the algorithm in the same round, we have that the value of $\mathtt{clockVal}$ at nodes $v$ and $w$ are equal to $r$ in round $r$, so the value of $\mathtt{token}$ is the same at nodes $v$ and $w$ in round $r$. Since both $v$ and $w$ execute line \ref{ChooseFinal} in round $r$, it follows that the \textbf{if} condition on line \ref{termCondition} evaluates to true at both nodes in round $r$. It follows that $v$ and $w$ have the same value of $\mathtt{myPhase1Col}$ in round $r$. By \Cref{SettledP1Col}, it follows that $\mathrm{SettledPhase1Col}_v = \mathrm{SettledPhase1Col}_w$, as desired.
\end{proof}

If two nodes choose their final colour in different rounds, then the node that chooses later can always avoid the colour that was chosen by the earlier node. We use this fact to confirm that, as long as two nodes choose their final colour in different rounds, they will terminate the algorithm with different final colours.

\begin{lemma}\label{DifferentFinalCol}
	If two neighbouring nodes $v,w$ execute line \ref{ChooseFinal} in different rounds, then they terminate the algorithm with different colours stored in $\mathtt{myFinalCol}$.
\end{lemma}
\begin{proof}
	For each $x \in \{v,w\}$, denote by $t_x$ the round in which node $x$ performs line \ref{ChooseFinal}, which is well-defined by \Cref{FinalExactlyOnce}. Suppose that $t_v \neq t_w$, and without loss of generality, assume that $t_v > t_w$. 
	
	In round $t_w$, node $w$ sets its value of $\mathtt{myFinalCol}$ to a value in $\{0,1,2\}$ (by \Cref{ValidFinalCol}) and sends this value to its neighbours with the string ``final''. Therefore, in round $t_w$, node $v$ receives a message of the form (``final'',$c_w$) where $c_w$ denotes the final colour chosen by $w$. At one of lines \ref{fromAP1}, \ref{fromBP1}, \ref{fromAP2}, \ref{fromBP2}, node $v$ will store the value of $c_w$ in either $\mathtt{Acol.final}$ or $\mathtt{Bcol.final}$ (depending on whether $w$ corresponds to $v$'s neighbour $A$ or neighbour $B$). Without loss of generality, assume that $v$ stores $c_w$ in $\mathtt{Acol.final}$. Since $w$ terminates its execution at line \ref{terminate} in round $t_w$, node $w$ will not modify its value of $\mathtt{myFinalCol}$ in any round after $t_w$, so $\mathtt{myFinalCol}$ at $w$ is $c_w$ at all times after round $t_w$. Moreover, since $w$ terminates its execution at line \ref{terminate} in round $t_w$, node $w$ will not send out another message after round $t_w$, so $v$'s value of $\mathtt{Acol.final}$ is $c_w$ at all times after round $t_w$.
	
	Finally, in round $t_v > t_w$, node $v$ sets its value of $\mathtt{myFinalCol}$ at line \ref{ChooseFinal}, and from the code, we see that $v$'s final colour is set to a value that is not equal to $\mathtt{Acol.final}$. But we proved above that $\mathtt{Acol.final}$ is equal to $c_w$ in round $t_v$, so $v$ chooses a final colour that is different from $w$'s. Since $v$ terminates its execution at line \ref{terminate} in round $t_v$, the value of $\mathtt{myFinalCol}$ at node $v$ is never modified again, which concludes the proof that $v$ and $w$ terminate the algorithm with different values stored in $\mathtt{myFinalCol}$.
\end{proof}

We are now able to prove the correctness of the algorithm.

\begin{theorem}
	For all pairs of neighbouring nodes $v,w$, the value of $\mathtt{myFinalCol}$ at $v$ upon $v$'s termination of the algorithm is different than the value of $\mathtt{myFinalCol}$ at $w$ upon $w$'s termination of the algorithm, and both of these values are in $\{0,1,2\}$.
\end{theorem}
\begin{proof}
	By \Cref{FinalExactlyOnce,ValidFinalCol}, both $v$ and $w$ set their value of $\mathtt{myFinalCol}$ to a value from $\{0,1,2\}$. By \Cref{DifferentSettledP1}, it follows that $\mathrm{SettledPhase1Col}_v \neq \mathrm{SettledPhase1Col}_w$, so, by \Cref{P2DifferentRounds}, nodes $v$ and $w$ execute line \ref{ChooseFinal} in different rounds. By \Cref{DifferentFinalCol}, it follows that $v$ and $w$ choose different values for $\mathtt{myFinalCol}$.
\end{proof}

%--------------------------------------------------
%--------------------------------------------------
\subsection{Extending to other graphs}\label{sec:linecolouringextend}
%--------------------------------------------------
%--------------------------------------------------
The exact same algorithm and proofs work when the graph is a cycle. In the case where the graph is a finite path, the only change needed is to modify the definition of local minimum and local maximum: a node is a local minimum if it has no neighbour with smaller ID, and a node is a local maximum if it has no neighbour with larger ID. Then, each endpoint of the path is considered a local maximum or local minimum, and the remaining details of the algorithm and proof work without further changes.

%--------------------------------------------------
%--------------------------------------------------
\section{Conclusion}\label{sec:conclusion}
%--------------------------------------------------
%--------------------------------------------------
We presented rendezvous algorithms for three scenarios: the scenario of the canonical line,  the scenario of arbitrary labeling with known initial distance $D$, and the scenario where each agent knows {\em a priori} only the label of its starting node. While for the first two scenarios the complexity of our algorithms is optimal (respectively $O(D)$ and $O(D\log^*\ell)$, where $\ell$ is the larger label of the two starting nodes), for the most general scenario, where each agent knows {\em a priori} only the label of its starting node, the complexity of our algorithm is $O(D^2(\log^*\ell)^3)$, for arbitrary unknown $D$, while the best known lower bound, valid also in this scenario, is $\Omega(D\log^*\ell)$ .

The natural open problem is the optimal complexity of rendezvous in the most general scenario (with arbitrary labeling and unknown $D$), both for the infinite labeled line and for the finite labeled lines and cycles. This open problem can be generalized to the class of arbitrary trees or even arbitrary graphs.

\bibliographystyle{plainurl}
\bibliography{references}

\begin{thebibliography}{10}

\bibitem{alpern95a}
Steve Alpern.
\newblock The rendezvous search problem.
\newblock {\em SIAM Journal on Control and Optimization}, 33(3):673--683, 1995.
\newblock \href {https://doi.org/10.1137/S0363012993249195}
  {\path{doi:10.1137/S0363012993249195}}.

\bibitem{alpern02a}
Steve Alpern.
\newblock Rendezvous search on labeled networks.
\newblock {\em Naval Research Logistics (NRL)}, 49(3):256--274, 2002.
\newblock \href {https://doi.org/10.1002/nav.10011}
  {\path{doi:10.1002/nav.10011}}.

\bibitem{alpern02b}
Steve Alpern and Shmuel Gal.
\newblock {\em The theory of search games and rendezvous}, volume~55 of {\em
  International series in operations research and management science}.
\newblock Kluwer, 2003.

\bibitem{anderson90}
E.~J. Anderson and R.~R. Weber.
\newblock The rendezvous problem on discrete locations.
\newblock {\em Journal of Applied Probability}, 27(4):839--851, 1990.
\newblock URL: \url{http://www.jstor.org/stable/3214827}.

\bibitem{anderson98a}
Edward~J. Anderson and S{\'{a}}ndor~P. Fekete.
\newblock Asymmetric rendezvous on the plane.
\newblock In Ravi Janardan, editor, {\em Proceedings of the Fourteenth Annual
  Symposium on Computational Geometry, Minneapolis, Minnesota, USA, June 7-10,
  1998}, pages 365--373. {ACM}, 1998.
\newblock \href {https://doi.org/10.1145/276884.276925}
  {\path{doi:10.1145/276884.276925}}.

\bibitem{anderson98b}
Edward~J. Anderson and S{\'{a}}ndor~P. Fekete.
\newblock Two dimensional rendezvous search.
\newblock {\em Oper. Res.}, 49(1):107--118, 2001.
\newblock \href {https://doi.org/10.1287/opre.49.1.107.11191}
  {\path{doi:10.1287/opre.49.1.107.11191}}.

\bibitem{BCGIL}
Evangelos Bampas, Jurek Czyzowicz, Leszek Gasieniec, David Ilcinkas, and Arnaud
  Labourel.
\newblock Almost optimal asynchronous rendezvous in infinite multidimensional
  grids.
\newblock In Nancy~A. Lynch and Alexander~A. Shvartsman, editors, {\em
  Distributed Computing, 24th International Symposium, {DISC} 2010, Cambridge,
  MA, USA, September 13-15, 2010. Proceedings}, volume 6343 of {\em Lecture
  Notes in Computer Science}, pages 297--311. Springer, 2010.
\newblock \href {https://doi.org/10.1007/978-3-642-15763-9\_28}
  {\path{doi:10.1007/978-3-642-15763-9\_28}}.

\bibitem{baston98}
Vic Baston and Shmuel Gal.
\newblock Rendezvous on the line when the players' initial distance is given by
  an unknown probability distribution.
\newblock {\em SIAM Journal on Control and Optimization}, 36(6):1880--1889,
  1998.
\newblock \href {https://doi.org/10.1137/S0363012996314130}
  {\path{doi:10.1137/S0363012996314130}}.

\bibitem{baston01}
Vic Baston and Shmuel Gal.
\newblock Rendezvous search when marks are left at the starting points.
\newblock {\em Naval Research Logistics}, 48(8):722--731, December 2001.
\newblock \href {https://doi.org/10.1002/nav.1044}
  {\path{doi:10.1002/nav.1044}}.

\bibitem{BP1}
Subhash Bhagat and Andrzej Pelc.
\newblock Deterministic rendezvous in infinite trees.
\newblock {\em CoRR}, abs/2203.05160, 2022.
\newblock \href {https://doi.org/10.48550/arXiv.2203.05160}
  {\path{doi:10.48550/arXiv.2203.05160}}.

\bibitem{BP2}
Subhash Bhagat and Andrzej Pelc.
\newblock How to meet at a node of any connected graph.
\newblock In Christian Scheideler, editor, {\em 36th International Symposium on
  Distributed Computing, {DISC} 2022, October 25-27, 2022, Augusta, Georgia,
  {USA}}, volume 246 of {\em LIPIcs}, pages 11:1--11:16. Schloss Dagstuhl -
  Leibniz-Zentrum f{\"{u}}r Informatik, 2022.
\newblock \href {https://doi.org/10.4230/LIPIcs.DISC.2022.11}
  {\path{doi:10.4230/LIPIcs.DISC.2022.11}}.

\bibitem{BBDDP}
S{\'{e}}bastien Bouchard, Marjorie Bournat, Yoann Dieudonn{\'{e}}, Swan Dubois,
  and Franck Petit.
\newblock Asynchronous approach in the plane: a deterministic polynomial
  algorithm.
\newblock {\em Distributed Comput.}, 32(4):317--337, 2019.
\newblock \href {https://doi.org/10.1007/S00446-018-0338-2}
  {\path{doi:10.1007/S00446-018-0338-2}}.

\bibitem{BDPP}
S{\'{e}}bastien Bouchard, Yoann Dieudonn{\'{e}}, Andrzej Pelc, and Franck
  Petit.
\newblock Almost universal anonymous rendezvous in the plane.
\newblock In Christian Scheideler and Michael Spear, editors, {\em {SPAA} '20:
  32nd {ACM} Symposium on Parallelism in Algorithms and Architectures, Virtual
  Event, USA, July 15-17, 2020}, pages 117--127. {ACM}, 2020.
\newblock \href {https://doi.org/10.1145/3350755.3400283}
  {\path{doi:10.1145/3350755.3400283}}.

\bibitem{CFPS}
Mark Cieliebak, Paola Flocchini, Giuseppe Prencipe, and Nicola Santoro.
\newblock Distributed computing by mobile robots: Gathering.
\newblock {\em {SIAM} J. Comput.}, 41(4):829--879, 2012.
\newblock \href {https://doi.org/10.1137/100796534}
  {\path{doi:10.1137/100796534}}.

\bibitem{CV}
Richard Cole and Uzi Vishkin.
\newblock Deterministic coin tossing with applications to optimal parallel list
  ranking.
\newblock {\em Inf. Control.}, 70(1):32--53, 1986.
\newblock \href {https://doi.org/10.1016/S0019-9958(86)80023-7}
  {\path{doi:10.1016/S0019-9958(86)80023-7}}.

\bibitem{CCGKM}
Andrew Collins, Jurek Czyzowicz, Leszek Gasieniec, Adrian Kosowski, and
  Russell~A. Martin.
\newblock Synchronous rendezvous for location-aware agents.
\newblock In David Peleg, editor, {\em Distributed Computing - 25th
  International Symposium, {DISC} 2011, Rome, Italy, September 20-22, 2011.
  Proceedings}, volume 6950 of {\em Lecture Notes in Computer Science}, pages
  447--459. Springer, 2011.
\newblock \href {https://doi.org/10.1007/978-3-642-24100-0\_42}
  {\path{doi:10.1007/978-3-642-24100-0\_42}}.

\bibitem{CGKK}
Jurek Czyzowicz, Leszek Gasieniec, Ryan Killick, and Evangelos Kranakis.
\newblock Symmetry breaking in the plane: Rendezvous by robots with unknown
  attributes.
\newblock In Peter Robinson and Faith Ellen, editors, {\em Proceedings of the
  2019 {ACM} Symposium on Principles of Distributed Computing, {PODC} 2019,
  Toronto, ON, Canada, July 29 - August 2, 2019}, pages 4--13. {ACM}, 2019.
\newblock \href {https://doi.org/10.1145/3293611.3331608}
  {\path{doi:10.1145/3293611.3331608}}.

\bibitem{CKP}
Jurek Czyzowicz, Adrian Kosowski, and Andrzej Pelc.
\newblock How to meet when you forget: log-space rendezvous in arbitrary
  graphs.
\newblock {\em Distributed Comput.}, 25(2):165--178, 2012.
\newblock \href {https://doi.org/10.1007/s00446-011-0141-9}
  {\path{doi:10.1007/s00446-011-0141-9}}.

\bibitem{DFKP}
Anders Dessmark, Pierre Fraigniaud, Dariusz~R. Kowalski, and Andrzej Pelc.
\newblock Deterministic rendezvous in graphs.
\newblock {\em Algorithmica}, 46(1):69--96, 2006.
\newblock \href {https://doi.org/10.1007/s00453-006-0074-2}
  {\path{doi:10.1007/s00453-006-0074-2}}.

\bibitem{DP}
Yoann Dieudonn{\'{e}} and Andrzej Pelc.
\newblock Anonymous meeting in networks.
\newblock {\em Algorithmica}, 74(2):908--946, 2016.
\newblock \href {https://doi.org/10.1007/s00453-015-9982-0}
  {\path{doi:10.1007/s00453-015-9982-0}}.

\bibitem{DPV}
Yoann Dieudonn{\'{e}}, Andrzej Pelc, and Vincent Villain.
\newblock How to meet asynchronously at polynomial cost.
\newblock {\em {SIAM} J. Comput.}, 44(3):844--867, 2015.
\newblock \href {https://doi.org/10.1137/130931990}
  {\path{doi:10.1137/130931990}}.

\bibitem{fpsw}
Paola Flocchini, Giuseppe Prencipe, Nicola Santoro, and Peter Widmayer.
\newblock Gathering of asynchronous robots with limited visibility.
\newblock {\em Theor. Comput. Sci.}, 337(1-3):147--168, 2005.
\newblock \href {https://doi.org/10.1016/j.tcs.2005.01.001}
  {\path{doi:10.1016/j.tcs.2005.01.001}}.

\bibitem{FP2}
Pierre Fraigniaud and Andrzej Pelc.
\newblock Delays induce an exponential memory gap for rendezvous in trees.
\newblock {\em {ACM} Trans. Algorithms}, 9(2):17:1--17:24, 2013.
\newblock \href {https://doi.org/10.1145/2438645.2438649}
  {\path{doi:10.1145/2438645.2438649}}.

\bibitem{KM}
Dariusz~R. Kowalski and Adam Malinowski.
\newblock How to meet in anonymous network.
\newblock {\em Theor. Comput. Sci.}, 399(1-2):141--156, 2008.
\newblock \href {https://doi.org/10.1016/j.tcs.2008.02.010}
  {\path{doi:10.1016/j.tcs.2008.02.010}}.

\bibitem{KKPM08}
Evangelos Kranakis, Danny Krizanc, and Pat Morin.
\newblock Randomized rendezvous with limited memory.
\newblock {\em {ACM} Trans. Algorithms}, 7(3):34:1--34:12, 2011.
\newblock \href {https://doi.org/10.1145/1978782.1978789}
  {\path{doi:10.1145/1978782.1978789}}.

\bibitem{KKSS}
Evangelos Kranakis, Nicola Santoro, Cindy Sawchuk, and Danny Krizanc.
\newblock Mobile agent rendezvous in a ring.
\newblock In {\em 23rd International Conference on Distributed Computing
  Systems {(ICDCS} 2003), 19-22 May 2003, Providence, RI, {USA}}, pages
  592--599. {IEEE} Computer Society, 2003.
\newblock \href {https://doi.org/10.1109/ICDCS.2003.1203510}
  {\path{doi:10.1109/ICDCS.2003.1203510}}.

\bibitem{LS}
Juhana Laurinharju and Jukka Suomela.
\newblock Linial's lower bound made easy.
\newblock {\em CoRR}, abs/1402.2552, 2014.
\newblock URL: \url{http://arxiv.org/abs/1402.2552}.

\bibitem{lim96}
Wei~Shi Lim and Steve Alpern.
\newblock Minimax rendezvous on the line.
\newblock {\em SIAM Journal on Control and Optimization}, 34(5):1650--1665,
  1996.
\newblock \href {https://doi.org/10.1137/S036301299427816X}
  {\path{doi:10.1137/S036301299427816X}}.

\bibitem{Li}
Nathan Linial.
\newblock Locality in distributed graph algorithms.
\newblock {\em SIAM Journal on Computing}, 21(1):193--201, 1992.
\newblock \href {https://doi.org/10.1137/0221015} {\path{doi:10.1137/0221015}}.

\bibitem{DGKKP}
Gianluca~De Marco, Luisa Gargano, Evangelos Kranakis, Danny Krizanc, Andrzej
  Pelc, and Ugo Vaccaro.
\newblock Asynchronous deterministic rendezvous in graphs.
\newblock {\em Theor. Comput. Sci.}, 355(3):315--326, 2006.
\newblock \href {https://doi.org/10.1016/J.TCS.2005.12.016}
  {\path{doi:10.1016/J.TCS.2005.12.016}}.

\bibitem{MP}
Avery Miller and Andrzej Pelc.
\newblock Tradeoffs between cost and information for rendezvous and treasure
  hunt.
\newblock {\em J. Parallel Distributed Comput.}, 83:159--167, 2015.
\newblock \href {https://doi.org/10.1016/j.jpdc.2015.06.004}
  {\path{doi:10.1016/j.jpdc.2015.06.004}}.

\bibitem{Pe}
Andrzej Pelc.
\newblock Deterministic rendezvous in networks: {A} comprehensive survey.
\newblock {\em Networks}, 59(3):331--347, 2012.
\newblock \href {https://doi.org/10.1002/net.21453}
  {\path{doi:10.1002/net.21453}}.

\bibitem{Pe2}
Andrzej Pelc.
\newblock Deterministic rendezvous algorithms.
\newblock In Paola Flocchini, Giuseppe Prencipe, and Nicola Santoro, editors,
  {\em Distributed Computing by Mobile Entities, Current Research in Moving and
  Computing}, volume 11340 of {\em Lecture Notes in Computer Science}, pages
  423--454. Springer, 2019.
\newblock \href {https://doi.org/10.1007/978-3-030-11072-7\_17}
  {\path{doi:10.1007/978-3-030-11072-7\_17}}.

\bibitem{Pel}
David Peleg.
\newblock {\em Distributed Computing: A Locality-Sensitive Approach}.
\newblock Society for Industrial and Applied Mathematics, 2000.
\newblock \href {https://doi.org/10.1137/1.9780898719772}
  {\path{doi:10.1137/1.9780898719772}}.

\bibitem{TSZ07}
Amnon Ta{-}Shma and Uri Zwick.
\newblock Deterministic rendezvous, treasure hunts, and strongly universal
  exploration sequences.
\newblock {\em {ACM} Trans. Algorithms}, 10(3):12:1--12:15, 2014.
\newblock \href {https://doi.org/10.1145/2601068} {\path{doi:10.1145/2601068}}.

\bibitem{thomas92}
L.~C. Thomas.
\newblock Finding your kids when they are lost.
\newblock {\em Journal of the Operational Research Society}, 43(6):637--639,
  1992.
\newblock \href {https://doi.org/10.1057/jors.1992.89}
  {\path{doi:10.1057/jors.1992.89}}.

\bibitem{DisCo}
Roger Wattenhofer.
\newblock Principles of distributed computing, 2023.
\newblock Accessed on 2023-04-16.
\newblock URL: \url{https://disco.ethz.ch/courses/fs23/podc/}.

\end{thebibliography}
\end{document}